\documentclass[a4paper,11pt,final]{article}
\pdfoutput=1

\usepackage{tcs}
\usepackage{hyperref}
\usepackage{microtype}
\begin{document}

\pagenumbering{gobble}
\title{An Improved Classical Singular Value Transformation \\
for Quantum Machine Learning}
\author{
Ainesh Bakshi \\
\texttt{ainesh@mit.edu} \\
MIT
\and
Ewin Tang \\
\texttt{ewint@cs.washington.edu} \\
University of Washington
}
\date{}

\maketitle

\begin{abstract} 
    The field of quantum machine learning (QML) produces many proposals for attaining quantum speedups for tasks in machine learning and data analysis.
    Such speedups can only manifest if classical algorithms for these tasks perform significantly slower than quantum ones.
    We study quantum-classical gaps in QML through the quantum singular value transformation (QSVT) framework.
    QSVT, introduced by Gily\'en, Su, Low and Wiebe~\cite{gslw18}, unifies all major types of quantum speedup~\cite{MartynRTC21}; in particular, a wide variety of QML proposals are applications of QSVT on low-rank classical data.
    We challenge these proposals by providing a classical algorithm that matches the performance of QSVT in this regime up to a small polynomial overhead.
    
    We show that, given a matrix $A \in \mathbb{C}^{m\times n}$, a vector $b \in \mathbb{C}^{n}$, a bounded degree-$d$ polynomial $p$, and linear-time pre-processing, we can output a description of a vector $v$ such that $\|v - p(A) b\| \leq \varepsilon\|b\|$ in $\widetilde{\mathcal{O}}(d^{11} \fnorm{A}^4 / (\varepsilon^2 \norm{A}^4 ))$ time.
    This improves upon the best known classical algorithm~\cite{cglltw19}, which requires $\widetilde{\mathcal{O}}(d^{22} \fnorm{A}^6 /(\varepsilon^6 \norm{A}^6 ) )$ time, and narrows the gap with QSVT, which, after linear-time pre-processing to load input into a quantum-accessible memory, can estimate the magnitude of an entry $p(A)b$ to $\eps\norm{b}$ error in $\widetilde{\mathcal{O}}(d\fnorm{A}/(\varepsilon \norm{A}))$ time.
    Instantiating our algorithm with different polynomials, we improve on prior classical algorithms for quantum-inspired regression~\cite{cglltw19,gst20}, recommendation systems~\cite{tang18a,cglltw19}, and Hamiltonian simulation~\cite{cglltw19}.

    Our key insight is to combine the \emph{Clenshaw recurrence}, an iterative method for computing matrix polynomials, with sketching techniques to simulate QSVT classically.
    We introduce several new classical techniques in this work, including (a) a \emph{non-oblivious} matrix sketch for approximately preserving bi-linear forms, (b) a new stability analysis for the Clenshaw recurrence, and (c) a new technique to bound arithmetic progressions of the coefficients appearing in the Chebyshev series expansion of bounded functions, each of which may be of independent interest.
\end{abstract}

\thispagestyle{empty}

\clearpage

\microtypesetup{protrusion=false}
\tableofcontents{}
\thispagestyle{empty}
\microtypesetup{protrusion=true}

\clearpage

\pagestyle{plain}
\pagenumbering{arabic}


\section{Introduction}

Quantum machine learning (QML) has rapidly developed into an active field of study with numerous proposals for speeding up machine learning tasks with quantum computers~\cite{DunjkoWittek20,chiprsw18}.
These proposals include quantum algorithms for basic tasks in machine learning, including regression~\cite{hhl09}, perceptron learning~\cite{wks16}, support vector machines~\cite{rml14}, recommendation systems~\cite{kerenidis2016QRecSys}, and semi-definite programming~\cite{bkllsw19}.
A central goal of QML is to demonstrate a problem on which quantum computers obtain a substantial practical speedup over classical computers.
A successful resolution of this goal would provide compelling motivation to invest more resources into developing scalable quantum computers (i.e.\ be a \emph{killer app} for quantum computing).

The quantum singular value transformation (QSVT) framework~\cite{LowC17,cgj18,gslw18} uses ideas from signal processing to unify algorithm design for quantum linear algebra and, by extension, QML.
This framework has been called a ``grand unification'' of quantum algorithms as it captures all major classes of quantum advantage~\cite{MartynRTC21}.
With respect to QML, QSVT captures essentially all known techniques in quantum linear algebra, so it will be the focus of our investigation.
QSVT generalizes the simple observation that, for a quantum state $\ket{b} = \sum_i b_i\ket{i} \in \mathbb{C}^n$ encoding the vector $b$ into its amplitudes, a quantum gate corresponds to applying a unitary matrix $U \in \mathbb{C}^{n\times n}$ to the vector, $\ket{b} \mapsto \ket{Ub}$.
By including the non-unitary operation of measurement, we can consider quantum circuits that perform $\ket{b} \mapsto \ket{Ab}$ for general (bounded) matrices $A \in \mathbb{C}^{m\times n}$; such circuits are called \emph{block-encodings}.
The fundamental result of QSVT is that, given a block-encoding of $A$ and a bounded, even or odd, degree-$d$ polynomial $p$, we can form a block-encoding of $p(A)$ with an overhead of $d$.\footnote{When $A$ is asymmetric, we can interpret QSVT as applying the matrix function that applies $p$ to each singular value of $A$ (\cref{def:qsvt}).}
This new quantum circuit can then be applied to a quantum state $\ket{b}$ to get $\ket{p(A)b}$ with probability $\norm{p(A)b}^2$.
In this way, quantum linear algebra algorithms can apply generic functions to the singular values of a matrix efficiently, provided that we have an efficient block-encoding and the function is smooth enough to be approximated well by a low-degree polynomial.

Though QSVT is an immensely powerful tool for quantum linear algebra, it requires having an efficient block-encoding for the input matrix $A$.
These efficient block-encodings do not exist in general, but they do exist for two broad classes of matrices, assuming appropriately strong forms of coherent access: matrices with low sparsity~\cite[Lemma 48]{gslw18} (a typical setting for quantum simulation) and matrices with low stable rank~\cite[Lemma 50]{gslw18} (a typical setting for quantum machine learning).
We treat the latter case; specifically, assuming the existence of \emph{quantum random access memory} (QRAM), a piece of hardware that supports queries to its memory in superposition~\cite{glm08},\footnote{
    Having quantum random access memory is an assumption similar to the classical random access memory assumption in algorithm design, that querying any piece of memory has cost only logarithmic (or constant) in input size.
    It's unclear whether we can attain the effectively constant cost of modern RAM when we want to query memory in superposition, making this assumption somewhat speculative~\cite{jr23}.
    In this paper, we will be fine giving quantum algorithms QRAM assumptions, and these matters will not affect the classical algorithms we present.
} we can process a matrix $A$ given as a list of its entries $A_{i,j}$ into a data structure in linear time such that preparation of a block-encoding of $A/\fnorm{A}$ is efficient (e.g.\ using $\bigO{\log(mn)}$ queries to the QRAM).
Here, $\fnorm{A} = \sqrt{\sum_{i,j} \abs{A_{i,j}}^2}$ denotes the Frobenius norm of $A$.
This type of block-encoding is the one commonly used for quantum linear algebra algorithms on classical data~\cite{kllp19,CdW21}, since it works for arbitrary matrices and vectors, paying only a $\fnorm{A}/\norm{A}$ factor in sub-normalization.

\ewin{What is the speedup attained by QSVT?}
\ainesh{we need to soup up this part, the belief in faster quantum algorithms comes from this perception of problems being hard classically. However, we challenge this assumption by coming up with a non-trivial classical algorithm that shows the gap between classical and quantum is at most quadratic in $1/\eps$, and quartic in stable rank.  }
A natural question then arises: how well can classical computers simulate QSVT?
Sparsity-based QSVT supports universal quantum computation, and therefore cannot be simulated classically in polynomial time unless BPP=BQP.
On the other hand, prior work by Chia, Gily\'en, Li, Lin, Tang and Wang~\cite[Section 6.1]{cglltw19} gives a classical algorithm that outputs a vector $v$ such that $\norm{ p(A)b  -v} \leq \epsilon\norm{b}$ in $\bigO{d^{22}\norm{A}_F^6 /(\epsilon^6 \norm{A}^6)}$ time, after a linear-time pre-processing phase.
More generally, all of the ``quantum-inspired'' algorithms generated through this framework have large polynomial slowdowns, as they require computing a singular value decomposition of a matrix with $\Omega((\frac{\fnorm{A}}{\norm{A} \eps})^2)$ rows and columns, immediately incurring a power-six dependence in $\eps$ and the Frobenius norm of $A$.
In fact, without assuming that the input matrix is strictly low-rank, all prior quantum-inspired algorithms incur a $1/\eps^6$ dependence, except for one for linear regression which incurs an $1/\eps^4$~\cite{gst20}.
\ewin{To be pedantic, technically I think the prior framework could give a $\fnorm{A}^4/\eps^4$ dependence, so I still think this sell is not super great.}
This gives the perception that a large polynomial running time (significantly larger than quartic) to simulate QSVT may be inherent.
In fact, this classical hardness has been conjectured~\cite{kllp19,kerenidis2022quantum}.
Such a conclusion would be significant, as this suggests a potential regime for a practical speedup for many problems in QML.
Therefore, the central question we address in this work is as follows: 

\begin{center}
    \textit{How large is the quantum-classical gap for singular value transformation \\ and the machine learning problems captured by it?}
\end{center}

\subsection{Results}

Our main result addresses the aforementioned question by providing an improved classical algorithm for QSVT that avoids computing a full singular value decomposition.
As a consequence, we also obtain better bounds on the quantum-classical gap for regression~\cite{cglltw20,gst20}, recommendation systems~\cite{tang18a,cglltw19}, and Hamiltonian simulation~\cite{cglltw19} with improved dependence in each relevant parameter.
Additionally, we improve over prior quantum-inspired algorithms with atypical guarantees in various parameter regimes~\cite{cchw20,sm21}.
We begin by stating our result for simulating QSVT on inputs with low stable rank.\footnote{The stable rank of a matrix $A$ is $\fnorm{A}^2/\norm{A}^2$.}

\begin{theorem}[Classical Singular Value Transformation, informal version of \cref{main-svt-alg}]
\label{thm:main-thm-intro}
Given a Hermitian $A \in \mathbb{C}^{n\times n}$ and $b \in \mathbb{C}^n$, and an accuracy parameter $0< \eps<1$, after $\bigO{\nnz(A) + \nnz(b)}$ pre-processing time to create a data structure,\footnote{If we are already given $A$ and $b$ in the QRAM data structures needed to prepare a block-encoding of $A/\fnorm{A}$ and a quantum state of $b/\norm{b}$, this pre-processing can be done in $\bigOt{d^{12}\fnorm{A}^4/( \norm{A}^4 \eps^4)}$ time.} for a degree-$d$ polynomial $p$ such that $\abs{p(x)} \leq 1$ for $x \in [-\norm{A},\norm{A}]$, we can output a description of a vector $y\in \C^n$ such that with probability at least $0.9$, $\|y - p(A)b\| \leq \eps\norm{b}$.
The algorithm takes
\begin{align*}
    \bigOt[\Big]{ \frac{d^{11}\fnorm{A}^4}{\norm{A}^4\eps^2}}
\end{align*}
time to output a description of $y$ as $Ax$ for a sparse vector $x$.
This description allows us to compute entries of $y$ in $\bigOt[\Big]{ \frac{d^6\fnorm{A}^2}{\norm{A}^2\eps^2}}$ time and obtain an $\ell_2^2$ sample from $y$ in $\bigOt[\Big]{ \frac{d^8 \fnorm{A}^4\norm{b}^2}{\norm{A}^4\epsilon^2\norm{y}^2} }$ time.\footnote{Throughout the paper, we use $\bigOt{\cdot}$ to surpress poly-logarithmic factors in $d, 1/\eps$ and $\fnorm{A}^2/\norm{A}^2$.}
\end{theorem}
The assumptions that $A$ is square and Hermitian are not needed; for other $A$, the definition of $p(A)$ needs to be adjusted appropriately and restricted to $p$ either even or odd.
For the discussion that follows, we assume $\norm{A} \leq 1$ and that $p(x)$ is bounded in the interval $[-1, 1]$, which allows us to analyze $p(x)$ in terms of its Chebyshev coefficients, without rescaling.

\paragraph{Comparison to QSVT.}
We compare our main result to the one achievable by QSVT.
The quantum computer is able to produce a quantum state approximating $\ket{p(A)b}$, but to give a concrete comparison, we consider the task of estimating an overlap with a given vector $v$.
Classically, we can use \cref{thm:main-thm-intro} to compute a description of $y \approx p(A)b$ and compute entries $y_i$ or estimate overlaps $\angles{v|y}$ of that description, all in $\bigOt{d^{11} \fnorm{A}^4/\eps^2}$ time (\cref{rmk:overlaps-description}).

In the same setting, a quantum computer equipped with a QRAM can output an $\eps$-good estimate of $\abs*{\langle v | p(A/2)b\rangle}^2$ in $\bigO{d\fnorm{A}\log(mn)/\eps}$ gates and queries to the QRAM.\footnote{
    We will not concern ourselves with $\log(mn)$ and $\log\frac{1}{\eps}$ factors: quantum algorithms typically count bit complexity where classical algorithms count word RAM complexity, which muddles any such comparisons.
    We will also ignore issues of sub-normalization (i.e.\ the $A$ vs $A/2$ in the comparison), though this seeming constant factor incurred by the quantum algorithm can lead to quadratic losses depending on the application.
}
We get this from the following procedure: in the pre-processing phase, place $A$ and $b$ in data structures in QRAM such that we can prepare states $\ket{b}$ and perform a block-encoding for $A/\fnorm{A}$ in $\bigO{\log(mn)}$ QRAM queries~\cite{prakash14}.
Using QSVT, we can get a block-encoding of an $\eps$ approximation of $A/2$ with cost inflated by a factor of $\bigO{\fnorm{A}\log(\fnorm{A}/\eps)}$ through singular value amplification~\cite[Theorem~30]{gslw18}, from which we can get a block-encoding of (approximately) $p(A/2)$ with cost inflated by a further factor of $\bigO{d}$.
Applying this to a state $\ket{b}$, we get an output state approximating $\ket{p(A/2)b}$ (thought of as a sub-normalized state) with $\bigOt{d\fnorm{A}\log(mn)\log(1/\eps)}$ gates.
We can then produce a sample from $p(A/2)b$ by paying an additional $1/\norm{p(A/2)b}^2$ factor.
So far, the dependence on error $\eps$ is merely logarithmic.
However, this is not a realizable quantum speedup (except possibly for sampling tasks) since the output is a quantum state: estimating some statistic of the quantum state requires incurring a polynomial dependence on $\eps$.
For example, if we want to estimate the overlap $\abs{\langle v | p(A/2)b \rangle}^2$, then this can be done with $1/\eps^2$ invocations of a swap test (or $1/\eps$ if one uses amplitude amplification).
More generally, distinguishing a state from one $\eps$-far in trace distance requires $\Omega(1/\eps)$ additional overhead, even when given an oracle efficiently preparing that state,\footnote{
    This holds because the output state of a quantum circuit that applies a unitary $U$ $T$ times is perturbed by at most $T\eps$ when $U$ is perturbed by $\eps$.
    So, $1/\eps$ applications of $U$ are needed to distinguish between $U$ and some $\eps$-close $\tilde{U}$ with constant probability.
} so estimating quantities to this sensitivity requires polynomial dependence on $\eps$.

To summarize, the quantum-classical gap is is 1-to-11 for the degree $d$, 1-to-4 for $\fnorm{A}$ (which we think of as square root of stable rank $\fnorm{A}^2/\norm{A}^2$), and 1-to-2 for $\eps$.\footnote{
    Note that the quantum-classical comparison can be very different (in both directions, favoring quantum or classical) depending on what property of the output vector we wish to learn.
    See the related work section for more discussion of this; we choose the ``overlap'' task, as this is a natural linear algebraic quantity to estimate, and is the quantity QML algorithms typically try to compute.
    \ewin{Add citations eventually.}
}
The quadratic gap in $\eps$ seems inherent, and the quartic gap in stable rank bears resemblance to quartic quantum speedups noted for other spectral algorithms~\cite{Hastings2020classicalquantum}.
The $d^{11}$ degree dependence could potentially be improved, though our techniques have an inherent limit of $d^5$; see~\cref{sec:tech} for more details.

\paragraph{Comparison to \cite{cglltw19,jgs20,gl22}.}
There are three papers that get similar results about ``dequantizing the quantum singular value transformation''.
The work of Chia, Gily\'{e}n, Li, Lin, Tang, and Wang~\cite{cglltw19} gives a running time of $\bigOt{d^{22}\fnorm{A}^6/\eps^6}$ (after renormalizing so that $\norm{A} = 1$ instead of $\fnorm{A} = 1$).
We improve in all parameters over this work: degree of the polynomial $p$, Frobenius norm of $A$, and accuracy parameter $\epsilon$.

The work of Jethwani, Le Gall, and Singh~\cite{jgs20} provides two algorithms for applying $f(A)b$ for Lipschitz-continuous $f$ and $A$ with condition number $\kappa$.
(Standard results in function approximation state that such functions can be approximated by polynomials of degree $\bigO{L}$, where $L$ is the Lipschitz constant of $f$ and the tail is either polynomial or exponential depending on how smooth $f$ is~\cite{trefethen19}.)
In particular, they achieve a running time of $\bigO{\fnorm{A}^6\kappa^{20}(d^2 + \kappa)^6/\eps^6}$ to apply a degree-$d$ polynomial.
Again, we improve in all parameters over this work: degree of the polynomial $p$, Frobenius norm of $A$, and accuracy parameter $\epsilon$.
We also do not incur condition number dependence.

Finally, the work of Gharibian and Le Gall~\cite{gl22} considers QSVT when the input matrices are sparse, which is the relevant regime for quantum chemistry and other problems in many-body systems, where the matrix is a local Hamiltonian.
See also Van den Nest's work which uses similar techniques to simulate restricted classes of quantum circuits~\cite{vdNest11}.
The setting where $A$ is sparse (corresponding to QSVT with sparse input) is significantly different from our setting, where $A$ has low stable rank.
In our setting, algorithms run in time polynomial in degree and stable rank, whereas in the low sparsity case, algorithms run in time exponential in degree and polynomial in sparsity instead of stable rank.
These are incomparable: algorithms specialized to one setting perform worse than the naive algorithm in the other setting.

\subsection{Applications to Dequantizing QML } \label{subsec:intro-apps}

Now, we describe the implications of \cref{thm:main-thm-intro} to specific problems in QML.
We obtain faster \emph{quantum-inspired} algorithms for linear regression, recommendation systems, and Hamiltonian simulation.

We begin with quantum recommendation systems, where the goal is to output a sample from a row of a low-rank approximation of the input vector $A$.
The original quantum algorithm has running time of $\bigOt{\fnorm{A}/\sigma}$ to output $\ket{[\operatorname{thresh}_\sigma(A)]_{i,*}}$, where $\operatorname{thresh}_\sigma$ is a polynomial approximating a threshold function---close to identity for values $\geq \sigma$, close to zero for values $\leq 5\sigma/6$, and bounded in between \cite{kerenidis2016QRecSys,cgj18}.
We then obtain the following corollary:

\ewin{Even comparing with my previous work $O(\frac{\fnorm{A}^6}{\sigma^{16}\eps^6})$ we're not better in the implicit parameter $\eta$ (which is set to a constant for the purpose of the application). I'm not sure why we're so lossy in the degree, ideally we'd get the $d$ exponent to $6$ or $5$ or $4$ but for some reason that's not happening. Maybe Lanczos would be better? annoyin}

\begin{corollary}[Dequantized recommendation systems, informal version of \cref{cor:dq-rec-sys}]
\label{cor:dq-rec-sys-intro}
    Given a matrix $A \in \mathbb{C}^{m\times n}$ such that $\norm{A} = 1$, a sufficiently small accuracy parameter $\eps > 0$, and an $i \in [n]$, we can produce a data structure in $\bigO{\nnz(A)}$ time such that, we can compute a description of a vector $y$ such that with probability at least $0.9$, $\Norm{ y - [\operatorname{thresh}_\sigma(A)]_{i,*} } \leq \eps\norm{A_{i,*}}$ in 
    $\bigOt{ \fnorm{A}^4 / \Paren{ \sigma^{11} \eps^2}}$ time. 
    From this description we can output a sample from $y$ in $\bigOt[\Big]{\frac{\fnorm{A}^4 \norm{A_{i,*}}^2}{ \sigma^8 \eps^2 \norm{y}^2}}$ time.
\end{corollary}

\begin{remark}[Comparison to \cite{cglltw19}]
    \cite[Corollary 6.7]{cglltw19} achieves a running time of $\bigOt{\fnorm{A}^6/(\sigma^{16}\eps^6)}$.
    We improve upon it in every parameter, including error $\eps$, the threshold $\sigma$, and the Frobenius norm $\fnorm{A}$. We defer comparison to the work of Chepurko, Clarkson, Horesh, Lin, and Woodruff~\cite{cchw20} to \cref{cchw-comparison-recsystems}, as it achieves different guarantees from QSVT.
    To summarize the discussion there, if we attempt to translate the guarantees of \cite{cchw20} to this setting, we find that our running times are better, under the mild assumption that $\eps$ is roughly smaller than $\sigma$.
\end{remark}

Next, we state the dequantization result we obtain for regression.
Quantum algorithms can produce a state $\eps$-close to $\ket{\operatorname{inv}_{\kappa}(A)b}$ in $\bigO{\kappa\fnorm{A}\log\frac{1}{\eps}}$ time, where $\operatorname{inv}_{\kappa}$ is a polynomial close to $1/x$ on $[1/\kappa, 1]$ and and smoothly thresholds away all singular values below $1/\kappa$~\cite{wossnig2018QLinSysAlgForDensMat,cgj18}.

\begin{corollary}[Dequantized regression, informal version of \cref{cor:dequant-regression}]
Given $A \in \mathbb{C}^{m \times n}$ and $b \in \mathbb{C}^{n}$ such that $\Norm{A} = 1, \norm{b}\leq 1$, a singular value threshold $0<1/\kappa<1$, and a sufficiently small accuracy parameter $\epsilon>0$.
Then after $\bigO{\nnz(A) + \nnz(b)}$ pre-processing time, we can compute a description of a vector $y \in \mathbb{C}^n$ such that, with probability $\geq 0.9$, $\norm{y - \operatorname{inv}_\kappa(A) b } \leq \eps\norm{b}/\kappa$.
The algorithm runs in $\bigOt{\kappa^{11}\norm{A}_F^4 /\epsilon^2 }$ time.
From this description we can output a sample from $y$ in $\bigOt[\Big]{\frac{\kappa^{10}\fnorm{A}^4\norm{b}^2}{\eps^2\norm{y}^2}}$ time.
\end{corollary}

\begin{remark}[Comparison with \cite{cglltw19,gst20,sm21}] \label{intro-regression-comparison}
Prior work \cite[Corollary 6.15]{cglltw19} achieves a running time of $\bigO{\norm{A}_F^6  \kappa^{22}/ \eps^6}$ to achieve an error bound of $\eps\norm{b}/\norm{A}$ ($\sigma$ in that setting is our $1/\kappa$).
Converting our result there increases our runtime by $\kappa^2$, but we still improve over this result in all parameters.

Other quantum-inspired algorithms give improved results under fairly strong additional assumptions.
When $A$ has no non-zero singular values larger than $\sigma$, then \cite{gst20} achieves a running time of $\bigO{\fnorm{A}^6/(\sigma^{12}\eps^4)}$, which Shao and Montanaro~\cite{sm21} improve further to $\bigO{\fnorm{A}^6/(\sigma^8\eps^2)}$ when $b$ is in the image of $A$.
We improve over the former in all parameters, and for the latter we obtain a better $\fnorm{A}$ dependence at the cost of a worse $\sigma$ dependence.
Both of these provide the stronger guarantee that the output vector is close up to relative error $\eps\norm{A^+b}$, though \cite{gst20} requires factors of $\norm{A^+}\norm{b}/\norm{A^+b}$ to compensate (which we elided in the above expression).
\end{remark}

Finally, we state the dequantization result we obtain for Hamiltonian simulation.

\begin{corollary}[Dequantized Hamiltonian simulation, informal version of \cref{cor:dequant-ham}]
\label{cor:dequant-ham-intro}
    Given a Hamiltonian $H \in \mathbb{C}^{n\times n}$ with $\norm{H} \leq 1$, a vector $b \in \mathbb{C}^n$, and a sufficiently small accuracy parameter $\eps > 0$, after $\bigO{\nnz(H) + \nnz(b)}$ pre-processing time, we can output a description of a vector $v$ such that, with probability $\geq 0.9$, $\norm{v - e^{iHt}b} \leq \eps\norm{b}$ with running time $\bigOt{t^{11}\fnorm{H}^4/\eps^2}$. 
    From this description we can output a sample from $v$ in $\bigOt[\Big]{\frac{t^{8}\fnorm{H}^4\norm{b}^2}{\eps^2\norm{v}^2}}$ time.
\end{corollary}

\begin{remark}[Comparison with \cite{cglltw19}]
The only prior work~\cite{cglltw19} we are aware of in the low-rank regime obtains a running time $\bigO{t^{16} \norm{H }_F^6/\eps^6}$, and we improve upon it in every parameter. 
\end{remark}

\section{Technical Overview} \label{sec:tech}

\begin{figure}[h]
    \begin{center}
    \begin{tabular}{r|l}
        Prior work~\cite{cglltw19} & $d^{22} \|A\|_F^6 /\varepsilon^6$ \\
        Using the polynomial structure & $d^{\diff{\mathbf{17}}}\|A\|_F^{\diff{\mathbf{4}}}/\varepsilon^{\diff{\mathbf{4}}}$ \\
        Tightening Clenshaw's stability analysis & $d^{\diff{\mathbf{13}}}\|A\|_F^4/\varepsilon^4$ \\
        Sparsifying matrices & $d^{\diff{\mathbf{11}}}\|A\|_F^4/\varepsilon^{\diff{\mathbf{2}}}$
    \end{tabular}
    \end{center}
    \caption{A list of our technical contributions, along with the improvements they each make to the running time of the final algorithm, ignoring log factors.}
\end{figure}

In this section, we describe our classical framework for simulating QSVT and provide an overview of our key new contributions.
For ease of exposition, we assume $\norm{A}=1$ and consider when $p$ is odd.
We have three conceptual contributions.
First, we use an iterative method (the Clenshaw recurrence) instead of a pure sketching algorithm, which improves the running time to $\bigO{d^{17}\norm{A}_F^4/\eps^4}$.
This corresponds to $d$ iterations of matrix-vector products of size $\fnorm{A}^2/(\eps/d^4)^2$ by $\fnorm{A}^2/(\eps/d^4)^2$.
Second, we develop new insights into the stability of the Clenshaw recurrence and arithmetic progressions of Chebyshev coefficients to improve the size from $\fnorm{A}^2/( \eps/d^4)^2$ to an $\fnorm{A}^2/( \eps/(d^3\log^2(d)))^2$.
Third, we give a subtle argument that we can sparsify these matrices, improving the $\eps^4$ to an $\eps^2$.
Together,\footnote{It is natural to wonder here why the complexity is not something like $d\norm{A}_F^4/(\eps/(d^3\log^2(d)))^2 = d^7\log^2(d)\norm{A}_F^4/\eps^2$. Such a running time is conceivable, but our analysis essentially replaces two factors of $1/\eps$ with factors of $1/d^2$, so the sparsification only saves a factor of $d^2$.} these give the final running time of $\bigO{d^{11}\log^4(d)\norm{A}_F^4/\eps^2}$.
We now walk through these steps in more detail.

\paragraph{Computing matrix polynomials through sketches and iterative algorithms.}  
Recall our goal of simulating QSVT: given a matrix $A \in \C^{m\times n}$, a vector $b \in \C^n$, and a polynomial $p: [-1,1] \to \mathbb{R}$, compute a description of a vector $y$ such that $\norm{y - p(A)b} \leq \eps\supnorm{p}\norm{b}$, where $\supnorm{p} \coloneqq \max_{x \in [-1,1]} \abs{p(x)}$.
Specifically, we aim for our algorithm to run in $\poly(\fnorm{A}, \frac{1}{\eps}, d)$ time after $\bigO{\nnz(A) + \nnz(b)}$ pre-processing, and our output description is some sparse vector $x$ such that $y = Ax$, since this allows us to simulate some tasks that the quantum algorithm can do with copies of $\ket{y}$, like estimating overlaps and performing measurements in the computational basis.

We require the running time of the algorithm to be independent of input dimension (after the pre-processing) and therefore are compelled to create sketches of $A$ and $b$ and work with these sketches. We note that prior work~\cite{tang18a,tang18b,cglltw19} stops here, and directly computes a SVD of the sketches, and applies the relevant polynomial $p$ to the singular values of the sketch. As noted in the previous section, this approach loses large polynomial factors in the relevant parameters. 
\ewin{Maybe add discussion of the $g(x)=f(x)/x$ issue with chia}

Our main conceptual insight is to run iterative algorithms on the resulting sketches of $A$ and $b$ in order to approximate matrix polynomials.
In the canonical numerical linear algebra regime (working with matrices directly, instead of sketches), there are two standard methods to achieve this goal: (1) compute the polynomial explicitly through something like a Clenshaw recurrence~\cite{clenshaw55}, or (2) use the Lanczos method~\cite{lanczos50} to create a roughly-orthogonal basis for the Krylov subspace $\{b, Ab, A^2b, \ldots\}$ and then apply the function exactly to the matrix in this basis, in which $A$ is tridiagonal, implicitly using a polynomial approximation in the analysis.
We note that in a world where we are allowed to pay $\bigO{\nnz(A)}$ (or even $\bigO{m + n}$) time per-iteration, we can simply run either of these algorithms and call it a day.
The main challenge in the setting we consider is that each iterative step must run in time that is entirely dimension-independent.

\begin{mdframed}
  \begin{algorithm}[Singular value transformation for odd polynomials. Informal version of \cref{algo:odd-dquantizing}]
    \label{algo:intro-odd-dquantizing-informal}\mbox{}
    \begin{description}
    \item[Input (pre-processing):] A  matrix $A \in \mathbb{C}^{m \times n}$, vector $b \in \mathbb{C}^{n}$, and parameters $0< \eps <1$. An odd degree $2d+1$ polynomial given as its Chebyshev coefficients $a_{2i+1}$ (so that $p(x) = \sum_{i = 0}^{d} a_{2i+1} T_{2i+1}(x)$).
    \item[Pre-processing sketches:]
    Let $s, t = \bigOt[\big]{\frac{d^6 \fnorm{A}^2}{ \eps^2} }$.
    \begin{enumerate}[label=P\arabic*., ref=A\thealgorithm.P\arabic*]
        \item Let $S \in \C^{n \times s}$ be a sampling matrix that samples $s$ columns of $A$ such that the $j$-th column is sampled with probability $ \tfrac12 \Paren{ \norm*{A_{*,j}}^2 / \fnorm{A}^2   + \abs*{b_j}^2/ \norm{b}^2 }$. 
        \item Let $T \in \C^{t \times m}$ be a sampling matrix that samples $t$ columns of $AS$ such that the $i$-th column is sampled with probability $\norm{ [AS]_{i,*} }^2 / \fnorm{AS}^2$. 
        \item Compute $TAS$. 
    \end{enumerate}

    \item[Clenshaw iteration:]
    Let $r = \widetilde{\mathcal{O}}\Paren{ d^{10} \fnorm{A}^4 / \eps^2 }$. Let $v_{d+1} = v_{d+2} = \vec{0}^s$. For $k  \in [d, 0]$,
    \begin{enumerate}[label=I\arabic*., ref=A\thealgorithm.I\arabic*]
        \item Let $B^{(k)} \in \mathbb{C}^{t \times s}$ be the estimator formed by sampling $r$ entries of $TAS$ such that entry $(i,j)$ is sampled with probability $\abs*{[TAS]_{i,j}}^2/\fnorm{TAS}^2$.
        Construct $B_\dagger^{(k)}$ from $(TAS)^\dagger$ similarly.
        \item Compute $v_k = 2(2B_\dagger^{(k)} B^{(k)} - I)v_{k+1} - v_{k+2} + a_{2k+1} S^\dagger b$. 
    \end{enumerate}
    \item[Output:] 
    Output $x = \tfrac12 S(v_0 - v_1)$ that satisfies $\Norm{ A x - p(A) b } \leq \eps\supnorm{p}\norm{b}$.
    \end{description}
  \end{algorithm}
\end{mdframed}

\paragraph{Sketching down the Clenshaw recurrence.}
For reasons explained below, we use the Clenshaw recurrence instead of the Lanczos method.
Given a degree-$d$ polynomial $p$ given in terms of its Chebyshev coefficients $a_\ell$ (i.e.\ $p(x) = \sum_{\ell = 0}^d a_\ell T_\ell(x)$, where $T_\ell(x)$ is the degree $\ell$ Chebyshev polynomial) and a value $x \in [-1, 1]$, the Clenshaw recurrence computes $p(x)$.
We use a modified recurrence for technical reasons.
\ewin{Maybe explain this eventually.}
Concretely, our recurrence for odd $p$ (so that $a_\ell = 0$ for $\ell$ even), is the following:
\begin{align*}
    q_{(d-1)/2} &= q_{(d+1)/2} = 0; \\
    q_k &= 2 (2x^2 - 1) q_{k+1} - q_{k+2} + 2a_{2k+1} x; \\
    p(x) &= \tfrac12(q_0 - q_1).
\intertext{The scalar recurrences we discuss lift to computing matrix polynomials in a natural way:}
    u_{(d-1)/2} &= u_{(d+1)/2} = 0; \\
    u_k &= 2 (2A A^\dagger - \Id) u_{k+1} - u_{k+2} + 2a_{2k+1} A b; \\
    p(A)b &= \tfrac12(u_0 - u_1).
\end{align*}
Each iteration (to get $u_k$ from $u_{k+1}$ and $u_{k+2}$) can be performed in $\bigO{\nnz(A)}$ arithmetic operations, so this can be used to compute $p(A)b$ in $\bigO{d\nnz(A)}$ operations.
We would like to do this approximately in time independent of $\nnz(A)$ and $n$.
We begin by sketching down our matrix and vector: we show that it suffices to maintain a sparse description of $u_k$ of the form $u_{k} = A v_k$ where $v_k$ is sparse.
In particular, we produce sketches $S \in \mathbb{C}^{n\times s}$ and $T \in \mathbb{C}^{t \times m}$ such that 
\begin{enumerate}
	\item $\norm*{ AS \Paren{AS}^\dagger - A A^\dagger   } \leq \eps$;
	\item $\norm*{ AS S^{\dagger} b - Ab  } \leq \eps \norm{b}$;
	\item $\norm*{ TAS \Paren{TAS}^\dagger - AS \Paren{AS}^\dagger  } \leq \eps$.
\end{enumerate}
Sketches that satisfy the above type of guarantees are called \emph{approximate matrix product} ($\aamp$) sketches, and are standard in the quantum-inspired algorithms literature~\cite{cglltw19}.
In the linear-time pre-processing phase, we can produce these sketches of size $s, t = \bigOt{\frac{\fnorm{A}^2}{\eps^2}\log\frac{1}{\delta}}$, and then compute $TAS$.
If the input is given in the quantum-inspired access model of \emph{oversampling and query access}, this can even be done in time independent of dimension.
All of these guarantees follow from \cref{thm:stable_amm}, which shows $\ell_2^2$ sampling gives an asymmetric approximate matrix product property (in operator norm).
We do not need this generalization (prior ``symmetric'' results suffice), but we use it for convenience. 

Using these guarantees we can sketch the iterates as follows:
\begin{equation}
\begin{split}
    u_k &= 2(2AA^\dagger- \Id)u_{k+1} - u_{k+2} + 2a_{2k+1} Ab\\
    &= 4AA^\dagger Av_{k+1}- 2Av_{k+1} - Av_{k+2} + 2a_{2k+1} Ab\\
    &\approx AS[4(TAS)^\dagger(TAS)v_{k+1}- 2v_{k+1} - v_{k+2} + 2a_{2k+1} S^\dagger b].
\end{split}
\end{equation}
Therefore, we can interpret Clenshaw iteration as the recursion on the dimension-independent term $v_k \approx 4(TAS)^\dagger(TAS)v_{k+1}- 2v_{k+1} - v_{k+2} + 2a_{2k+1} S^\dagger b$, and then applying $AS$ on the left to lift it back to $m$ dimensional space.
We can recognize this as roughly the recursion performed in \cref{algo:intro-odd-dquantizing-informal}.
As desired, we can perform the iteration to produce $v_k$ in $\bigO{st} = \bigOt{\frac{\fnorm{A}^4}{\eps^4}\log^2\frac{1}{\delta}}$ time, which is independent of dimension, at the cost of incurring $\bigO{\eps (\norm{v_{k+1}} + \norm{v_{k+2}} + a_{2k+1}\norm{b})}$ error.
To bound the effect of these per-iteration errors on the final output, we need a stability analysis of the Clenshaw recurrence.

\paragraph{Connecting sketching error to finite precision.}
Given that we sketch each iterate down to a size that is dimension-independent, we introduce additive error at each step.
We can re-interpret this additive error as truncation error, and forge a conceptual connection between running \textit{iterative algorithms} on sketches and \textit{stability analyses} (also known as \emph{finite-precision analyses}) in the classical NLA literature. 
Stability analyses of Clenshaw and Lanczos iteration are well-understood~\cite{clenshaw55,Paige76,mms18}.
However, we cannot use them in a black-box manner, since they are concerned with optimizing the ``number of bits'' required to maintain an accurate solution, and so are content with error bounds that are only polynomially tight.
Translating these results to our setting results in a significantly worse polynomial running time.
Therefore, for our purposes, we must revisit classical stability analyses and refine our understanding of how error propagates in these iterative methods.

Folklore intuition suggests that Lanczos is a stabler algorithm for applying matrix functions, but the state-of-the-art analysis of it~\cite{mms18} relies on the stability of the Clenshaw recurrence as a subroutine, and therefore gives a strictly worse error-accumulation bounds than the Clenshaw recurrence.
In particular, if we wish to compute a generic $p(A)b$ to $\eps\supnorm{p}\norm{b}$ error in the regime where every matrix-vector product $Ax$ incurs an error of $\eps\norm{A}\norm{x}$ using Lanczos, the stability analysis of Musco, Musco, and Sidford suggests that the error of the output is $\bigO{d^{5.5}\eps}$, which would introduce a $d^{11}$ in our setting~\cite{mms18}.\footnote{
    This computation arises from taking Lemma 9 of \cite{mms18} to evaluate a degree-$d$ polynomial, say, bounded by 1 in $[-1,1]$. A generic example of such polynomial is only by a constant in $[-1 - \eta, 1 + \eta]$ when $\eta = \bigO{1/d^2}$ (\cref{slightly-beyond-sup}), and has Chebyshev coefficients bounded only by a constant, without decaying.
    Thus, the bound from Lemma 9 becomes $\bigO{d^5\norm{E}}$.
    Continuing the analysis into \cite{Paige76}, $E$ is the matrix whose $i$th column is the error incurred in the $i$th iteration; each column has norm $\eps\norm{A}\norm*{v_{j+1}} = \eps$ in our setting where we only incur error in matrix-vector products, since $\norm{A} = 1$ by normalizing and $\norm*{v_{j+1}} = 1$ because the algorithm normalizes it to unit norm, and we assume that scalar addition and multiplication can be performed exactly.
    We have no control over the direction of error, so we can at best bound $\norm{E}$ with the column-wise bound, giving $\norm{E} \leq \fnorm{E} \leq \sqrt{k}\eps$.
    So, our version of \cite[Equation 16]{mms18} gives $\norm{E} \leq \sqrt{d}\eps$, which gives us the asserted $\bigO{d^{5.5}\eps}$ bound.
}
To incur less error, we do not use Lanczos and analyze the Clenshaw recurrence directly.

This discussion will focus on standard Clenshaw instead of our modified version, since this is the focus of prior work; the same discussion holds for our odd/even recurrences, but with an additional factor of $d$ incurred.
The Clenshaw recurrence computes $p(x)$ through the iteration computing $q_d$ through to $q_0$:
\begin{align*}
    q_{d+1}, q_{d+2} &\coloneqq 0; \\
    q_k &\coloneqq 2x q_{k+1} - q_{k+2} + a_k; \\
    p(x) &= \tfrac{1}{2}(a_0 + q_0 - q_2).
\end{align*}
For example, in the randomized numerical linear algebra (RNLA) literature, this is often applied in the case where $a_d = 1$ and $a_k = 0$ otherwise, to evaluate a degree-$d$ Chebyshev polynomial $T_d(x)$.

Now, we consider the Clenshaw recurrence when each iteration incurs error $\eps(\abs{q_{k+1}} + \abs{q_{k+1}})$.
A naive argument gives a $\bigO{d^3}$ bound on the overhead, which implies that $\eps$ needs to be scaled down by that much to get the desired error bound.
For our modified recurrence, this is a $\bigO{d^4}$, giving a time of $\bigO{d^{16}\frac{\fnorm{A}^4}{\eps^4}\log^2\frac{1}{\delta}}$ to perform an iteration, and $d$ times that for the total runtime.

We can hope for better.
By the Markov brothers' inequality, a bounded polynomial has derivative bounded by $d^2$~\cite{schaeffer41}, and attained by the Chebyshev polynomial $T_d(x)$.
So, if we only incur error from an error in input, in that we perform the recurrence not with $x$ but some value in $(x-\eps, x+\eps)$, this error only cascades to a $\bigO{d^2\eps}$ worst-case error in the output.
This argument suggests that a $\bigO{d^2}$ overhead is the best we could hope for.

Our technical contribution is an analysis showing that the Clenshaw algorithm gives this optimal overhead, up to a logarithmic overhead.
We first show that the overhead can be upper bounded by the size of the largest Clenshaw iterate $\abs{q_k}$ (see~\cref{scalar-clenshaw-stability}), and then we bound the size of iterate $\abs{q_k}$ by $\bigO{d\log(d)\supnorm{p}}$ (see~\cref{thm:scalar-clenshaw-iterate-bound}).
The main lemma we prove states that for a bounded polynomial $p(x) = \sum_{\ell=0}^d a_\ell T_\ell(x)$, sums of the form $a_\ell + a_{\ell + 2} + \cdots$ are all bounded by $\bigO{\log(\ell)\supnorm{p}}$ (\cref{parity-two-tailsums}).
This statement follows from facts in Fourier analysis (in particular, this $\log(\ell)$ is the same $\log$ as the one occurs when bounding the $L^1$ norm of the Dirichlet kernel).

To our knowledge, our work is the first to give any bound of $o(d^3)$.
The standard literature either considers an additive error (where, unlike usual models like floating-point arithmetic, each multiplication incurs identical error regardless of magnitude)~\cite{clenshaw55,foxparker68,mh02} or eventually boils down to bounding $\abs{a_i}$ (since their main concern is dependence on $x$)~\cite{Oliver77,Oliver79}.
The modern work we are aware of shows a $\bigO{d^2}$ bound only for computing Chebyshev polynomials~\cite{bkm22}, sometimes used to give a $\bigO{d^2 \sum_{i} \abs{a_i}} = \bigO{d^3}$ bound for computing generic bounded polynomials~\cite{mms18}, since a degree-$d$ polynomial can be written as a linear combination of $T_k(x)$ with bounded coefficients.
None of these analysis are sufficient to get our $d^2\log(d)$ stability bound, since they ultimately depend on $\sum \abs{a_i}$, which can be as small as $\supnorm{p}$ and as large as $\sqrt{d}\supnorm{p}$; see the start of \cref{sec:basic-error-analysis-scalar} for more discussion.
Our improved stability analysis saves a $d^4$ factor in our main algorithm, from $d^{17}$ to $d^{13}$.

\paragraph{Improving the $\eps$-dependence.}
So far, our improvements give a $\bigOt{d^{13}\frac{\fnorm{A}^4}{\eps^4}\log^2\frac{1}{\delta}}$ running time, whereas we wish to achieve a $\bigO{1/\eps^2}$ dependence. 
Though so far we have only used a very limited subset of the sketching toolkit---namely, $\ell_2^2$ importance sampling---it's not clear how, for example, oblivious sketches~\cite{nn13} or the connection between importance sampling and leverage score sampling~\cite{cchw20} help us, since our choices of sketches are optimal up to log factors for the guarantees we desire.
To get the additional improvement, we need a new sketching technique.

A natural next step to improve per-iteration running time is to sparsify the matrix $TAS$, in order to make matrix-vector products more efficient.
If we can sparsify $TAS$ to $\bigO{1/\eps^2}$ non-zero entries, then we get the desired quadratic savings in per-iteration cost.
There is significant literature on sparsifying the entries of a matrix~\cite{am07,kundu2014note,braverman2021near}.
However, it does not suffice for us to use these as a black box.
For example, consider the sketch given by Drineas and Zouzias~\cite{drineas2011note}: for a matrix $M \in \mathbb{R}^{n\times n}$, zero out every entry smaller than $\frac{\eps}{2n}$, then sample entries proportional to their $\ell_2^2$ magnitude, and consider the corresponding unbiased estimator of $M$, denoted $\tilde{M}$.
The guarantee is that the operator norms are close, $\norm{M - \tilde{M}} \leq \eps$, and the sparsity is $\bigO{n\log(n)\frac{\fnorm{A}^2}{\eps^2}}$.
\ewin{Maybe cite LazySVD since it's similar: they use these sparsification sketches, right?}
This is precisely the guarantee we need and the sketch can be performed efficiently; however, this does not sparsify the matrix for us, since in our setting, our matrices $TAS$ have dimension $\fnorm{A}^2/\eps^2$, so the sparsity guarantee is only $\bigO{\frac{\fnorm{A}^4\log(n)}{\eps^4}} = \bigO{st\log(n)}$.
In other words, this sketch gives us no improvement on sparsity!

\paragraph{Bi-linear entry-wise sampling transform.}
Our second main technical contribution is to bypass this barrier by noticing that we don't need a guarantee as strong as a spectral norm bound.
Instead, we only need to achieve approximations of the form 
\begin{equation}
\label{eqn:bi-linear-tech-ov}
\norm*{ASMv - AS\tilde{M}x_k} \leq \eps,
\end{equation}
 for various different choices of $x_k$.
So, we only need our sketch $\tilde{M}$ to approximate $M$ in some directions with good probability, instead of approximating it in all directions.
We define a simple sketch, which we call the Bilinear Entry Sampling Transform (\ssketch\footnote{That is, we pronounce $\ssketch(A)$ as ``best of $A$''. We make no normative claim about our sketch vis-a-vis other sketches.}), that is an unbiased estimator for $A$ linearly (rather than the product $A^\dagger A$) and achieves the aforementioned guarantee (Equation \eqref{eqn:bi-linear-tech-ov}). 
This sketch is sampled from the same distribution as the one of Drineas and Zouzias from above, but without the zero-ing out small entries. In particular, we define
\begin{equation*}
    M^{(k)} = \frac{1}{p_{i,j} } A_{i,j} e_i e_j^\dagger \hspace{0.2in } \textrm{ with probability }  p_{i,j} = \frac{ \abs{A_{i,j} } }{\Norm{A}_F^2}.
\end{equation*} 
Then, to get the sketch with sparsity $r$, we take the average of $r$ independent copies of the above random matrix:
\begin{equation*}
    \ssketch(A) = \frac{1}{r} \sum_{k \in [r]} M^{(k)}.
\end{equation*}
This definition is not new: for example, one of the earliest papers on matrix sparsification for linear algebra algorithms briefly considers this sketch~\cite{am07}, and this sketch has been used implicitly in prior work on dequantizing QML algorithms~\cite{tang18b}.
However, to the best of our knowledge, our analysis of this sketch for preserving bi-linear forms and saving a factor of $1/\epsilon^2$ is novel.

We show that $\ssketch(A)$ satisfies the following guarantees:  taking $r = \Omega(\fnorm{M}^2/\eps^2)$, this sketch preserves the bilinear form $u^\dagger M v$ to $\eps\norm{u}\norm{v}$ error with probability $\geq 0.9$ (\cref{lem:basic-facts}).
Second, taking $r = \Omega(\fnorm{M}^2n)$, this sketch preserves the norm $\norm{Mv}$ to $0.1\norm{v}$ error with probability $\geq 0.9$.
So, by taking $r = \Omega(\fnorm{M}^2(n + \frac{1}{\eps^2}))$, we can get both properties.
However, with this choice of $r$, $\ssketch(M)$ does not preserve the spectral norm $\norm{M}$, even to constant error.\footnote{
    At this point, one might wonder whether one can simply zero out small entries to get these guarantees with the additional spectral norm bound.
    This is not the case; if we only zero out entries smaller than $\eps/(2n)$, this threshold is too small to improve the matrix Bernstein tail bound, whereas if we zero out entries smaller than, say, $1/(100n)$, the matrix Bernstein tail bound goes through, but the bilinear form is not $\eps$-preserved.
}
This sketch can be interpreted as a relaxation of the approximate matrix product and the sparsification sketches above, since we use it in regimes where $\norm{\ssketch(A)}$ is unbounded and $\norm{\ssketch(A)b}$ is not $\eps$-close to $\norm{Ab}$.
However, if we use it in the ``interior'' of a matrix product expression, as a proxy for approximate matrix product, it can be used successfully to improve the $n/\eps^2$ dependence of spectral norm entry sparsification to something like $n + 1/\eps^2$.
In our setting, this corresponds to an improvement of $1/\eps^4$ to $1/\eps^2$.

\paragraph{Applying \ssketch~to Clenshaw Iteration.} Returning to our odd Clenshaw iteration, we had our approximate iterate
\begin{align*}
    u_k &\approx AS[4(TAS)^\dagger(TAS)v_{k+1}- 2v_{k+1} - v_{k+2} + 2a_{2k+1} S^\dagger b].
\intertext{We can approximate this by taking $B = \ssketch(TAS)$ and $B_\dagger = \ssketch((TAS)^\dagger)$ with sparsity $r = \Theta(\fnorm{A}^4/\eps^2)$ to get that the bilinear forms like $[AS]_{i,*}(TAS)^\dagger((TAS)v_{k+1})$ are preserved. This allows us to successfully approximate with sparsified matrices,}
    u_k &\approx AS[4B_\dagger Bv_{k+1}- 2v_{k+1} - v_{k+2} + 2a_{2k+1} S^\dagger b], \\
    \text{so } v_k &\approx 4B_\dagger B v_{k+1}- 2v_{k+1} - v_{k+2} + 2a_{2k+1} S^\dagger b.
\end{align*}
This recurrence in $v_k$ will be our algorithm (\cref{algo:intro-odd-dquantizing-informal}): compute $S$ and $T$ in the pre-processing phase, then compute the iteration of $v_k$'s, pulling fresh copies of $B$ and $B_\dagger$ each time, and output it as the description of the output $u \approx p(A)b$.
What remains is the error analysis, which similar to the scalar case, requires bounding the size of the iterates $\norm{v_k}$.
We used the $\eps$-approximation of the bilinear form to show that the sparsifications $B$ and $B_\dagger$ successfully approximate $u_k$; we use the constant-error approximation of the norm to show that the $\norm{v_k}$'s are bounded.

\paragraph{Challenges in extending finite-precision Clenshaw iteration.}
We note here that the above error does not directly follow from the error of the sketches along with the finite-precision scalar Clenshaw iteration.
The error we incur in iteration $k$ is not $\eps\norm{u_{k+1}}$, but $\eps\norm{v_{k+1}}$, so our error analysis requires understanding both the original Clenshaw iteration along with the ``dual'' Clenshaw iteration $v_k$, which requires a separate error analysis.
This dual iteration is essentially the same Clenshaw iteration, but for computing $p(x)/x$, which bears resemblance to how prior algorithms for SVT depend on the Lipschitz constant of $p(x)/x$~\cite{cglltw19}.
That it is evaluating $p(x)/x$ makes sense, since the dual iteration is the same as the Clenshaw iteration but with a factor of $A$ removed.

\paragraph{Sums of Chebyshev coefficients.}
One final technical wrinkle remains, which is to bound the error accumulation of the matrix Clenshaw recurrences.
In the same way that the original scalar Clenshaw algorithm required bounding arithmetic progressions of Chebyshev coefficients with step size two, $a_\ell + a_{\ell + 2} + \cdots$ by the norm of the corresponding polynomial $p = \sum_{\ell = 0}^d a_\ell T_\ell(x)$, to prove stability for the even and odd versions, we need to bound arithmetic progressions with step size four.
Surprisingly, this is significantly more challenging.

We give a thorough explanation for this in \cref{why-sum-hard}, but in brief, some arithmetic progressions of Chebyshev coefficients arise naturally as linear combinations of polynomial evaluations.
For example, $\sum_{k \geq 0} a_{2k} = \tfrac12(p(-1) + p(1))$, so we can conclude that $\abs{\sum_{k \geq 0} a_{2k}} \leq \supnorm{p}$.
Through Fourier analysis, this can be generalized to progressions with different step sizes, but breaks for progressions with certain offsets.
The sum $\sum_{k \geq 0} a_{4k+1}$ is one such example, and this is a quantity we need to bound to give a $\bigOt{d^2}$ bound on the iterates of odd matrix Clenshaw.  A na\"ive bound of $\bigO{d}$ follows from bounding each coefficient separately, but this results in a significantly worse running time down the line. 

In \cref{lem:odd-flipped-coefs}, we show that we can bound the above sum $\sum_{k\geq 0}a_{4k+1}$ by $\bigO{\log^2(d)\supnorm{p}}$.
This shows that this sum has near-total cancellation: despite our guarantee on coefficient being merely that it is bounded by a constant,\footnote{
    In fact, for any degree $d$, there are polynomials of that degree which are bounded and yet $\sum_{\ell=0}^d \abs{a_\ell} = \Omega(d)$ \cite[Theorems 8.2 and 8.3]{trefethen19}.
} the sum of $\bigO{d}$ of these coefficients is only poly-logarithmic in magnitude.
The proof proceeds by considering many different polynomial evaluations $p(x_0), p(x_1),\ldots, p(x_d)$, and trying to write the arithmetic progression as a linear combination of these evaluations $\sum c_k p(x_k)$.
This can be thought of as a linear system in the $c_k$'s, which we can then prove has a solution, and then bound the $\ell_1$ norm of this solution.
To do this, we argue that we can bound its solution $A^{-1}b$ by the solution of a different system, $C^{-1}b$, where $C$ is an matrix that is entrywise larger than $A$.
Since this is not true generally, we use strong properties of the particular matrix $A$ in the linear system at hand to prove this.

\section{Discussion}
\label{sec:discussion}

\paragraph{Comparison to \cite{cchw20}.}

The work of Chepurko, Clarkson, Horesh, Lin, and Woodruff \cite{cchw20} gives algorithms for low-rank approximation sampling and ridge regression with significantly different guarantees from that of prior quantum-inspired algorithms, since these atypical guarantees are more amenable to existing sketching techniques.
These do not match QSVT in all regimes, but one could argue that for the purposes of the particular problems targeted, they match them ``in spirit'' for practical regimes.

The guarantees obtained by \cite{cchw20} differ from those obtained by QSVT since they do not try to obtain an approximation of $f(A)b$, for some function $f$. They instead create data-dependent sketches and solve the corresponding optimization problem in the sketched space. Since the sketches preserve the cost, the solution in the sketched space has comparable cost to the optimal solution ($f(A)b$) in the original space.
However, such a guarantee does not immediately translate to a bound on the distance between the sketched solution and the optimal solution.
While leverage score and ridge leverage score sampling are the right primitives when sketching a problem of optimizing an objective value, quantum algorithms extend more broadly to general functions $f$ and achieve different guarantees, about closeness to the output vector.

\paragraph{Related techniques in randomized numerical linear algebra.}
Our work draws upon several ideas from the randomized numerical linear algebra literature.
We refer the reader to the surveys of Mahoney~\cite{Mahoney11} and Woodruff~\cite{w14} for the relevant background.
The asymmetric approximate matrix product sketch we introduce is closely related to the AMP sketches considered in~\cite{magdon2011using,magen2011low,cohen2015dimensionality} (non-oblivious) and~\cite{cohen2016optimal} (oblivious).
Our result about asymmetric AMP has essentially been observed by Magen and Zouzias~\cite{magen2011low}, which samples from the distribution $\Set{ \norm*{A_{i,*}}\norm*{B_{i,*}}/ \sum_{j \in [n]}\norm*{A_{j,*}}\norm*{B_{j,*}} }$, which our distribution oversamples.
We could have used other AMP sketches as well, but we use importance sampling because (1) we can prepare these sketches in linear-time, (2) we can prepare them in time independent of dimension in the quantum-inspired setting, and (3) the sampling probabilities are optimal when $B = A^\dagger$, as shown by Drineas, Kannan, and Mahoney~\cite{dkm06b}.

Regarding the bi-linear entry-wise sampling transform, creating sketches by sampling entries proportional to their squared-magnitude are well-known in the literature but are typically used to bound the operator norm of the matrix rather than any particular bi-linear form.
Bounding the operator norm directly happens to be too restrictive, as we discussed in the technical overview.
Regardless, entry-wise sampling sketches were introduced by Achlioptas and McSherry~\cite{am07} to speed up low-rank approximation.
Arora, Hazan and Kale used them in the context of faster algorithms for SDP solving~\cite{arora2006fast}.
Since then, a series of works have provided sharper guarantees for such sketches~\cite{gittens2009error,drineas2011note,kundu2014note,kundu2017recovering,braverman2021near}.

In addition to the usual sketching and sampling toolkit, we also use ideas from iterative algorithms in the numerical linear algebra literature.
Iterative methods, such as power iteration, Golub-Kahan bidiagonalization, and Arnoldi iteration are ubiquitous in scientific computing and are used for linear programming, low-rank approximation, and numerous other fundamental linear algebra primitives~\cite{saad1981krylov,trefethen1997numerical}.
Our work is closely related to iterative algorithms that use the sparsification sketches and exploit singular values gaps~\cite{allen2016lazysvd,gs18,musco2015randomized,bcw22}, but this prior work incurs dimension-dependent factors, which are crucial to avoid in our setting. 

\paragraph{Speedups in quantum machine learning.}
In view of the broader goal of finding quantum advantage for machine learning tasks, we spend some time here to point out when our assumptions do not hold, and therefore, dequantization results do not apply.
First, quantum-inspired linear algebra crucially relies on the input data being \emph{classical}, meaning that, for example, we are given input data as a list of entries, rather than as a quantum state, which does not have its amplitudes easily accessible.
Specifically, the weakest setting in which our algorithms work is when we have \emph{oversampling and query access} to input (\cref{defn:phi-sq-access}).
This type of access has extensibility properties similar to that of the block-encoding and we often get this access whenever quantum states and generic block-encodings can be prepared efficiently~\cite[Section 3]{cglltw19}.
However, as observed by Cotler, Huang, and McClean~\cite{chm21}, there are simple problems for which having access to entries makes trivial a problem that is exponentially hard when only given vectors via their corresponding states.
A dequantized algorithm simply proves that, if its quantum counterpart does achieve a, say, exponential speedup, then oversampling and query access queries, which we need to run a dequantized algorithm, must be exponentially hard to perform on a classical computer.
This explains why quantum principal component analysis has a dequantization~\cite{tang18b} when one has classical access to input as well as a proof of exponential quantum advantage~\cite{hbcclmnbkpm22} when one does not.
More generally, an algorithm with a dequantization is still useful when run on ``quantum data''.

Even given classical data, algorithms can resist dequantization by using high-degree sparsity-based QSVT rather than low-rank-based QSVT, as techniques do not extend to this regime.
Further, sparsity-based QSVT is BQP-complete (indeed, quantum computation can be described as applying block-encodings of low-sparsity unitary matrices), so we would not expect this to be dequantized in full.
In particular, we note that the current proposals that resist dequantization and potentially obtain a super-polynomial quantum speedup include Zhao, Fitzsimons, and Fitzsimons on Gaussian process regression~\cite{zff15} and Lloyd, Garnerone, and Zanardi on topological data analysis~\cite{lgz16}.
Though these works avoid the dequantization barrier to large quantum speedups, their potential for practical quantum speedups remains to be seen, since the decision on whether the speedups are ``practical'' or ``useful'' will ultimately come down to the particular choice of dataset and hardware.

Finally, even if a quantum algorithm can be simulated by a classical algorithm, there are certain problems for which having $\ket{v}$ is better than having the succinct representation of $v$.
For example, estimating the Forrelation~\cite{aa18,ac17} of $\ket{v}$ with another vector $\ket{u}$ requires exponentially many queries to $v$ when it is given as a classical list of entries, even with the ability to produce importance samples.
If the desired task is to output the Forrelation of an output of low-rank QSVT, this could potentially give a large speedup despite our results on low-rank QSVT.
However, this is a case of artificially adding hardness: we are unaware of problems in machine learning where Forrelation-type quantities are desired.
Such a problem would be a good candidate for QML speedup.

\paragraph{Open problems.}

The main question in this area remains: what kinds of quantum linear algebra can we make convincing arguments of quantum speedup for?
As our main result highlights, there is currently a quartic gap in the $\fnorm{A}$ dependence, which researchers believe is large enough for near-term speedups~\cite{babbush2021focus}.
Are there ways to attain this speedup with more practical quantum algorithms?

Another question arising in our work is that of polynomial evaluation.
We showed that computing a polynomial $p(x)$ with $\eps$-noisy arithmetic operations can be done to $\bigO{d^2\eps}$ error.
Can we show a more fine-grained result: that, for $p(x)$ bounded and $L$-Lipschitz, can we compute $p(x)$ to $\bigO{L\eps}$ error?

With respect to the classical SVT algorithm, a natural question is to ask for an improved algorithm for estimating $\tr(p(A))$ to relative error, for $p(x)$ a low-degree bounded polynomial and $A$ symmetric.
Using our result along with a Hutch-style estimate requires paying dimension dependence, but prior work shows that it is possible to achieve a running time of $\bigO{\nnz(A) + \poly(\fnorm{A}, d, 1/\eps)}$ with a high polynomial cost, at least when $p(x)$ is an approximation to $e^x$~\cite{cglltw19}.
Can we get an improved algorithm for matrix polynomial traces?

Finally, with respect to the techniques we introduce, we could ask for settings in which these can be applied and improvements to our techniques.
Can our matrix sparsification techniques be applied to other settings?


\section{Preliminaries}

\ewin{Checked this section until 4.3!}

We use the notation $f \lesssim g$ to denote the ordering $f = \bigO{g}$ (and respectively for $\gtrsim$ and $\eqsim$), and $\bigOt{f}$ is shorthand for $\bigO{f\poly(\log f)}$.
We use $\log$ to refer to the natural logarithm.
We use the Iverson bracket, where $\iver{P}$ is one if the predicate $P$ is true and zero otherwise.
For example, $\sum_{i=1}^d \sum_{j = i}^d a_{ij} = \sum_{i = 1}^d \sum_{j=1}^d a_{ij} \iver{j \geq i}$.
Finally, we assume that arithmetic operations (e.g., addition and multiplication of real numbers) and function evaluation oracles (computing $f(x)$ from $x$) take unit time.

\subsection{Linear Algebra}

For a vector $v \in \mathbb{C}^n$, the standard Euclidean norm of $v$ is denoted $\norm{v} \coloneqq (\sum_{i=1}^n \abs{v_i}^2)^{1/2}$.
The number of non-zero entries of $v$ is denoted $\znorm{z}$.
For a matrix $A \in \mathbb{C}^{m\times n}$, the \emph{Frobenius norm} of $A$ is $\fnorm{A} \coloneqq (\sum_{i=1}^m\sum_{j=1}^{n} \abs{A_{i,j}}^2)^{1/2}$ and the \emph{spectral norm} of $A$ is $\norm{A} \coloneqq \sup_{x\in \C^{n}, \norm{x}=1} \norm{Ax}$.
The $i$-th row and $j$-th column of $A$ are denoted $A_{i,*}$ and $A_{*, j}$, respectively.
$A^\dagger$ denotes the \emph{conjugate transpose} of $A$, and $A^+$ denotes the \emph{Moore-Penrose pseudoinverse} of $A$.

A \emph{singular value decomposition} (SVD) of $A$ is a representation $A = UDV^{\dag}$, where for $N \coloneqq \min(m,n)$, $U\in \C^{m\times N}$ and $V\in \C^{n\times N}$ are isometries and $D\in \R^{N\times N}$ is diagonal with $ \sigma_i\coloneqq D_{i,i}$ and $\sigma_1 \geq \sigma_2 \geq \cdots \geq \sigma_N \geq 0$.
We can also write this decomposition as $A = \sum_{i=1}^{N} \sigma_i U_{*,i} V_{*,i}^\dagger$.

We denote the set of singular values of $A$ by $\spec(A) \coloneqq \{\sigma_k\}_{k \in [N]}$.
The \emph{rank} of $A$ is the number of nonzero singular values and the \emph{stable rank} of $A$ is $\fnorm{A}^2/\norm{A}^2$.

For a Hermitian matrix $A \in \C^{n\times n}$ and a function $f: \R \to \C$, $f(A) \in \C^{n\times n}$ denotes the matrix where $f$ is applied to the eigenvalues of $A$.
Below, in \cref{def:qsvt}, we define notions of matrix function that extend to non-Hermitian matrices.

\subsection{Polynomials and the Chebyshev Basis}

We consider polynomials with real coefficients, $p \in \R[x]$.
For a Hermitian matrix $A$, $p(A)$ refers to evaluating the polynomial with $x$ replacing $A$; this is equivalent to applying $p$ to the eigenvalues of $A$.
The right definition for applying $p$ to a general non-square matrix is subtle; as done in QSVT, we restrict to settings where the matrix formed by evaluating $p$ on the singular values of $A$ coincides with the evaluation of a corresponding polynomial in $A$.

\begin{definition}[Definition 6.1 of \cite{cglltw19}] \label{def:qsvt}
For a matrix $A \in \C^{m\times n}$ and degree-$d$ polynomial $p(x) \in \R[x]$ of parity-$d$ (i.e., even if $d$ is even and odd if $d$ is odd), we define the notation $p(A)$ in the following way:
\begin{enumerate}
    \item If $p$ is \emph{even}, meaning that we can express $p(x) = q(x^2)$ for some polynomial $q(x)$, then $$p(A) \coloneqq q(A^\dagger A)= p(\sqrt{A^\dagger A}).$$
    \item If $p$ is \emph{odd}, meaning that we can express $p(x) = x\cdot q(x^2)$ for some polynomial $q(x)$, then $$p(A) \coloneqq A\cdot q(\sqrt{A^\dagger A}).$$
\end{enumerate}
\end{definition}

For example, if $p(x) = x^2 + 1$, then $p(A) = A^\dagger A + I$, and if $p(x) = x^3 + x$, then $p(A) = AA^\dagger A + A$.
Looking at a singular value decomposition $A = \sum \sigma_iU_{*,i}V_{*,i}^\dagger$, $p(A) = \sum p(\sigma_i) U_{*,i}V_{*,i}^\dagger$ when $p$ is odd and $p(A) = \sum p(\sigma_i) V_{*,i}V_{*,i}^\dagger$ when $p$ is even, thus making this definition coincide with the singular value transformation as given in \cite[Definition 16]{gslw18}.

We work in the Chebyshev basis of polynomials throughout.
Let $T_\ell(x)$ and $U_\ell(x)$ denote the degree-$\ell$ Chebyshev polynomials of the first and second kind, respectively.
They can be defined on $[-1,1]$ via
\begin{align}
    T_\ell(\cos(\theta)) &= \cos(\ell\theta) \text{ and } \\
    U_\ell(\cos(\theta)) &= \sin((\ell+1)x)/\sin(x),
\end{align}
but we will give attention to their recursive definitions, since we use them for computation.
\begin{align}
    T_0(x) &= 1 & U_0(x) &= 1 \nonumber\\
    T_1(x) &= x & U_1(x) &= 2x \label{cheb-recursive-definition}\\
    T_k(x) &= 2x\cdot T_{k-1}(x) - T_{k-2}(x)
    & U_k(x) &= 2x\cdot U_{k-1}(x) - U_{k-2}(x) \nonumber
\end{align}
For a function $f: [-1,1] \to \mathbb{R}$, we denote $\supnorm{f} \coloneqq \sup_{x \in [-1,1]}\abs{f(x)}$.
In this norm, the Chebyshev polynomials have $\supnorm{T_k(x)} = 1$ and $\supnorm{U_k(x)} = n+1$.
More generally, for a function $f: S \to \mathbb{R}$ for $S \subset \mathbb{R}$, we denote $\norm{f}_S \coloneqq \sup_{x \in S} \abs{f(x)}$, so that $\supnorm{f} = \norm{f}_{[-1,1]}$.

We use the following well-known properties of Chebyshev polynomials from Mason and Handscomb~\cite{mh02}.
\begin{align}
    T_i(x) &= \tfrac12(U_i(x) - U_{i-2}(x)) \text{ for } i \geq 1 \label{eq:t-to-u} \\
    U_i(x) &= \sum_{j \geq 0} T_{i - 2j}(x)(1 + \iver{i - 2j \neq 0}) \label{eq:u-to-t} \\
    T_{jk}(x) &= T_j(T_k(x)) \label{eq:t-composition} \\
    U_{2k+1}(x) &= U_k(T_2(x))U_1(x) = U_k(T_2(x))2x \label{eq:u-composition} \\
    \tfrac{d}{dx} T_k(x) &= k U_{k-1}(x) \label{eq:t-differentiation}
\end{align}
Any Lipschitz continuous function\footnote{
    We call a function $f: [-1,1] \to \mathbb{R}$ Lipschitz continuous if there exists a constant $C$ such that $\abs{f(x) - f(y)} \leq C\abs{x - y}$ for $x, y \in [-1,1]$.
} $f: [-1,1] \to \mathbb{R}$ can be written as a (unique) linear combination of Chebyshev polynomials, $f(x) = \sum_\ell a_\ell T_\ell(x)$ (where we interpret $T_{\ell}(x) \equiv 0$ for negative $\ell$).
When $f$ is a degree-$d$ polynomial, then $a_\ell = 0$ for all $\ell > d$.
A common way to approximate a function is by truncating its Chebyshev series expansion; we denote this operation by $f_k(x) \coloneqq \sum_{\ell=0}^k a_\ell T_\ell(x)$, and we denote the remainder to be $\bar{f}_k(x) \coloneqq f(x) - f_k(x) = \sum_{\ell=k+1}^{\infty} a_\ell T_\ell(x)$.
Standard results in approximation give bounds on $\supnorm{f - f_k}$ for various smoothness assumptions on $f$.
We recommend the book by Trefethen on this topic~\cite{trefethen19}, and use results from it throughout.
We list here some basic lemmas for future use.

\begin{lemma}[{Coefficient bound, consequence of \cite[Eq. (3.12)]{trefethen19}}]
\label{lem:coefficient-bound}
    Let $f: [-1,1] \to \mathbb{R}$ be a Lipschitz continuous function.
    Then all its Chebyshev coefficients $a_k$ are bounded: $\abs{a_k} \leq 2\supnorm{f}$.
\end{lemma}

\begin{lemma}\label{slightly-beyond-sup}
    For a degree-$d$ polynomial $p$, and $\delta = \frac{1}{4d^2}$,
    \begin{align*}
        \norm{p}_{[-1-\delta, 1+\delta]} \leq e\norm{p}_{[-1,1]}.
    \end{align*}
\end{lemma}
\begin{proof}
    Without loss of generality, take $\supnorm{p} = 1$.
    By \cite[Proposition~2.4]{sv14} and basic properties of Chebyshev polynomials,
    \begin{align*}
        \norm{p(x)}_{[-1-\delta, 1+\delta]}
        \leq \norm{T_d(x)}_{[-1-\delta, 1+\delta]}
        = T_d(1+\delta).
    \end{align*}
    Further, by Proposition 2.5 in~\cite{sv14}, we can evaluate $T_d(1+\delta)$ via the formula
    \begin{align*}
        T_d(x) &= \frac12\Big(x + \sqrt{x^2-1}\Big)^d + \frac12\Big(x - \sqrt{x^2-1}\Big)^d \\
        T_d(1+\delta) &= \frac12\Big(1 + \delta + \sqrt{2\delta + \delta^2}\Big)^d + \frac12\Big(1 + \delta - \sqrt{2\delta + \delta^2}\Big)^d \\
        &\leq \exp\Big(d(\delta + \sqrt{2\delta + \delta^2})\Big) \\
        &\leq \exp\Big(\tfrac{1}{4d} + \sqrt{\tfrac{1}{2} + \tfrac{1}{16d^2}}\Big)
        \leq e. \qedhere
    \end{align*}
\end{proof}

\subsection{Sampling and Query Access}
\label{subsec:sq}

We now introduce the ``quantum-inspired'' access model, following the exposition from prior work~\cite{cglltw19}.
We refer the reader there for a more thorough investigation of this access model.
From a sketching perspective, this model encompasses ``the set of algorithms that can be performed in time independent of input dimension, using only $\ell_2^2$ sampling'', and is a decent classical analogue for the input given to a quantum machine learning algorithms operating on classical data.

\begin{definition}[Sampling and query access to a vector, {\cite[Definition 3.2]{cglltw19}}]
  \label{defn:sq-access}
  For a vector $v \in \C^n$, we have $\sq(v)$, \emph{sampling and query access} to $v$, if we can:
  \begin{enumerate}
    \item query for entries of $v$;
    \item obtain independent samples $s \in [n]$ where the probability that $s$ is some $i$ is $\abs*{v_i}^2/\norm{v}^2$;
    \item query for $\|v\|$.
  \end{enumerate}
\end{definition}

We will use our pre-processing time to construct data structures that give us sampling and query access to our input.
The samples from $\sq(v)$ are called ``$\ell_2^2$ importance samples'' in the randomized numerical linear algebra literature; we will call them samples from $v$.
Such samples are equivalent to measurements of the quantum state $|v\rangle \coloneqq \frac{1}{\|v\|}\sum v_i|i\rangle$ in the computational basis.
Sampling and query access is closed under taking linear combinations, once we introduce slack in the form of oversampling.
\begin{definition}[Oversampling and query access, {\cite[Definition 3.4]{cglltw19}}]
  \label{defn:phi-sq-access}
  For $v \in \C^n$ and $\phi \geq 1$, we have $\sq_\phi(v)$, \emph{$\phi$-oversampling and query access} to $v$, if we can query for entries of $v$ and we have $\sq(\tilde{v})$, where $\tilde{v} \in \C^n$ is a vector satisfying $\|\tilde{v}\|^2 = \phi\|v\|^2$ and $\abs{\tilde{v}_i}^2 \geq \abs{v_i}^2$ for all $i \in [n]$.
\end{definition}

The $\ell_2^2$ distribution over $\tilde{v}$ $\phi$-oversamples the distribution over $v$:
\begin{align*}
  \frac{\abs{\tilde{v}_i}^2}{\|\tilde{v}\|^2} = \frac{\abs{\tilde{v}_i}^2}{\phi\|v\|^2} \geq \frac{1}{\phi}\frac{\abs{v_i}^2}{\|v\|^2}.
\end{align*}
Intuitively speaking, as a consequence, estimators that use samples from $v$ can also use samples from $\tilde{v}$ at the expense of a factor $\phi$ increase in the number of samples used.
Formally, we can prove that oversampling access implies an approximate version of the usual sampling access:

\begin{lemma}[Oversampling to sampling, {\cite[Lemma 3.5]{cglltw19}}] \label{lem:b-sq-approx}
    For $u \in \mathbb{C}^n$, suppose we are given $\sq_\phi(u)$ and some $\delta \in (0,1]$.
    We can sample from $u$ with probability $\geq 1-\delta$ in $\bigO{\phi \log\frac{1}{\delta}}$ queries to $\sq_\phi(u)$.
    We can also estimate $\|u\|$ to $\nu$ multiplicative error for $\nu \in (0,1]$ with probability $\geq 1-\delta$ in $\bigO{\frac{\phi}{\nu^2}\log\frac{1}{\delta}}$ queries.
    Both of these algorithms take linear time in the number of queries.
\end{lemma}

Generally, compared to a quantum algorithm that can output (and measure) a desired vector $|v\rangle$, our algorithms will output $\sq_\phi(u)$ such that $\|u - v\|$ is small.
Our analysis will give a bound on $\phi$ to show that we can output samples from $|v\rangle$.
As for error, a bound $\norm{u - v} \leq \eps\norm{v}$ implies that measurements from $u$ and measurements from $v$ are $2\eps$-close in total variation distance {\cite[Lemma~4.1]{tang18a}}.
Now, we show that oversampling and query access of vectors is closed under taking small linear combinations.

\begin{lemma}[Linear combinations, {\cite[Lemma 3.6]{cglltw19}}] \label{lemma:sample-Mv}
  Given $\sq_{\varphi_t}(v^{(t)}) \in \C^n$ and $\lambda_t \in \C$ for all $t \in [\tau]$, we have $\sq_\phi(\sum_{t=1}^{\tau} \lambda_tv^{(t)})$ for $\phi = \tau\frac{\sum \varphi_t\|\lambda_tv^{(t)}\|^2}{\|\sum \lambda_tv^{(t)}\|^2}$.
  After paying the pre-processing cost of querying for each of the norms of the $\tilde{v}^{(t)}$'s, the cost of any query is equal to the cost of sampling from any of the $v^{(t)}$'s plus the cost of querying an entry from all of the $v^{(t)}$'s.
\end{lemma}

We can also define (over)sampling and query access for a matrix $A$ as having sampling and query access to all of the rows of $A$, along with the vector of row norms of $A$.

\begin{definition}[Oversampling and query access to a matrix, {\cite[Definition~3.7]{cglltw19}}]
  \label{defn:sampling-A}
For a matrix $A \in \C^{m\times n}$, we have $\sq(A)$ if we have $\sq(A_{i,*})$ for all $i \in [m]$ and $\sq(a)$ for $a \in \R^m$ the vector of row norms ($a_i \coloneqq \norm{A_{i,*}}$).

We have $\sq_\phi(A)$ if we can query entries of $A$ and we have $\sq(\tilde{A})$ for $\tilde{A} \in \C^{m\times n}$ satisfying $\fnorm{\tilde{A}}^2 = \phi\fnorm{A}^2$ and $\abs{\tilde{A}_{i,j}}^2 \geq \abs{A_{i,j}}^2$ for all $(i,j) \in [m]\times[n]$.
\end{definition}

Sampling and query access is relevant to QML because many settings where we can apply QML to classical data (that is, data given as a list of entries, with some allowed pre-processing or in some data structure) also admits efficient $\sq$ access to input.
So, we concern ourselves with the setting that we can perform $\sq$ queries in, say, $\bigO{1}$ time.

\begin{remark} \label{rmk:when-sq-access}
We can get $\sq$ access to input matrices and vectors in input-sparsity time.
Given $v \in \C^n$ in the standard RAM model, the alias method \cite{Vose1991} takes $\Theta(\nnz(v))$ pre-processing time to output a data structure that uses $\Theta(\nnz(v))$ space and can sample from $v$ in $\Theta(1)$ time.
In other words, we can get $\sq(v)$ with constant-time queries in $\bigO{\nnz(v)}$ time, and by extension, for a matrix $A \in \C^{m\times n}$, $\sq(A)$ with constant-time queries in $\bigO{\nnz(v)}$ time.\footnote{This holds in the word-RAM model, with an additional $\log$ overhead when considering bit complexity.}
\end{remark}

The algorithms presented here will take linear-time pre-processing to construct the above data structure, among other things.
However, they will still run in time independent of dimension without this pre-processing, supposing that we have efficient $\sq$ access to input.
Finally, we synthesize the prior results on linear combinations into the following corollary, which shows that a certain type of succinct description of a vector $u$ is sufficient to get approximate sampling and query access to it.

\begin{corollary} \label{Ax-sampling}
    Suppose we are given sampling and query access to a matrix $A \in \C^{m\times n}$ and a vector $b \in \C^n$, where we can respond to queries in $\bigO{1}$ time.
    Further suppose we have a vector $u \in C^n$ implicitly represented by $v \in \C^m$ and $\eta$, with $u = A^\dagger v + \eta b$.
    Then we can:
    \begin{enumerate}[label=(\roman*)]
        \item Compute entries of $u$ in $\bigO{\znorm{v}}$ time;
        \item Sample $i \in [n]$ with probability $\frac{\abs{u_i}^2}{\norm{u}^2}$ in $\bigO[\Big]{\znorm{v}(\znorm{v} \sum_k \norm{v_kA_{k, *}}^2 + \eta^2\norm{b}^2)\frac{1}{\norm{u}^2}\log\frac{1}{\delta}}$ time with probability $\geq 1-\delta$;
        \item Estimate $\norm{u}^2$ to $\nu$ relative error in $\bigO[\Big]{\znorm{v}(\znorm{v} \sum_k \norm{v_kA_{k, *}}^2 + \eta^2\norm{b}^2)\frac{1}{\nu^2\norm{u}^2}\log\frac{1}{\delta}}$ time with probability $\geq 1-\delta$.
    \end{enumerate}
\end{corollary}
\begin{proof}
By \cref{lemma:sample-Mv}, we have $\sq_\phi(A^\dagger v)$ for 
\begin{align*}
    \phi = \znorm{v} \frac{\sum_k \norm{v_kA_{k, *}}^2}{\norm{A^\dagger v}^2}
\end{align*}
and a query cost of $\bigO{\znorm{v}}$.
Applying \cref{lemma:sample-Mv} again, we have $\sq_\varphi(A^\dagger v + \eta b)$ for
\begin{align*}
    \varphi = 2\frac{\znorm{v} \sum_k \norm{v_kA_{k, *}}^2 + \eta^2\norm{b}^2}{\norm{A^\dagger v + \eta b}^2}.
\end{align*}
By \cref{lem:b-sq-approx}, we can draw one sample from $u$ with probability $\geq 1-\delta$ with $\bigO{\varphi\log\frac{1}{\delta}}$ queries to $\sq_\phi(A^\dagger v)$, each of which takes $\bigO{\znorm{v}}$ time.
Similarly, we can estimate $\norm{u}^2$ to $\nu$ multiplicative error with $\bigO{\frac{\varphi}{\nu^2}\log\frac{1}{\delta}}$ queries to $\sq_\phi(A^\dagger v)$.
\end{proof}

For intuition, the running times above are only large when $u = A^\dagger v + \eta b$ has significantly smaller norm than the magnitude of the summands $v_i (A_{i,*})^\dagger$ would suggest.
Usually, in our applications, we can intuitively think about this overhead being small when the desired output vector mostly lies in a subspace spanned by singular vectors with large singular values in our low-rank input.
Quantum algorithms also have the same kind of overhead.
Namely, the QSVT framework encodes this in the subnormalization constant $\alpha$ of block-encodings, and the overhead from the subnormalization appears during post-selection~\cite{gslw18}.
When this cancellation is not too large, the resulting overhead typically does not affect too badly the runtime of our applications.


\section{Extending the Sketching Toolkit}

In this section, we show how to extend the modern sketching toolkit (see e.g.~\cite{w14}) in two ways: (a) we provide a sub-sampling sketch that preserves bi-linear forms with only a inverse quadratic dependence on $\eps$ and (b) a non-oblivious, $\ell_2^2$ sampling based asymmetric approximate matrix product sketch. 

\subsection{The Bi-Linear Entry-wise Sampling Transform}
\label{subsec:best-description}

\begin{definition}[Bi-linear Entry-wise Sparsifying Transform] \label{def:ssketch}
    For a matrix $A \in \C^{m \times n}$, the \ssketch{} of $A$ with parameter $T$ is a matrix sampled as follows: for all $k \in [T]$,
    \begin{equation*}
        M^{(k)} = \frac{1}{p_{i,j} } A_{i,j} e_i e_j^\dagger \hspace{0.2in } \textrm{ with probability }  p_{i,j} = \frac{ \abs*{A_{i,j} }^2 }{\Norm{A}_F^2}.
    \end{equation*} 
    Then, 
    \begin{equation*}
    \ssketch_T(A) = \frac{1}{T} \sum_{k \in [T]} M^{(k)}.
    \end{equation*}
\end{definition}

\begin{lemma}[Basic properties of the Bi-Linear Entry-wise Sparsifying Transform]
\label{lem:basic-facts}
    For a matrix $A \in \C^{m \times n}$, let $M = \ssketch{} (A)$ with parameter $T$. Then, for $X \in \C^{m\times m}$,  $u \in \C^{m}$ and $v \in \C^{n}$, we have
    \begin{align}
        \nnz(M) &\leq T; \label{ssketch:nnz}\\
        \expecf{}{M } &= A; \label{ssketch:expect}\\
        \expecf{}{  M^\dagger X M  - A^\dagger XA } &= \frac{1}{T} \Paren{ \Tr(X)\|A\|_F^2I - A^\dagger XA }. \label{ssketch:moment}
    \end{align}
\end{lemma}
\begin{proof}
Observe, since $M$ is an average of $T$ sub-samples, each of which are $1$-sparse, $M$ has at most $T$ non-zero entries. Next,
\begin{equation*}
\expecf{}{M} = \frac{1}{T} \sum_{k \in T} \expec{}{M^{(k)}} = \sum_{ i \in [m]} \sum_{ j \in [n] } p_{i,j} \frac{ A_{i,j}}{p_{i,j}} e_i e_j^\dagger = A .
\end{equation*}
Similarly,
\begin{align*}
\expecf{}{M^\dagger X M } & =  \frac{1}{T^2} \expecf{}{  \Paren[\Big]{ \sum_{k \in [T]} M^{(k)} }^\dagger  X\Paren[\Big]{ \sum_{k \in [T]} M^{(k)} }  } \\
& = \frac{1}{T^2}  \expecf{}{  \Paren[\Big]{ \sum_{k , k' \in [T]} \Paren*{ M^{(k)} }^\dagger  X M^{(k')}}  }  \\
& = \frac{1}{T^2} \Paren[\Big]{   \Paren[\Big]{ \sum_{k \neq  k' \in [T]} \expec{}{  M^{(k)} } ^\dagger \cdot  X  \cdot \expec{}{ M^{(k')}}  }   +    \Paren[\Big]{ \sum_{k   \in [T]} \expec{}{ \Paren*{ M^{(k)} }^\dagger  X M^{(k)}}  } }  \\
& = \Paren[\Big]{ 1- \frac{1}{T} } A^\dagger X A + \frac{1}{T} \sum_{i \in [m], j \in [n] } p_{i,j }  \frac{ A_{i,j}^2 }{p_{i,j}^2 } e_j e_i^\dagger X e_i e_j^\dagger \\
& = \Paren[\Big]{ 1- \frac{1}{T} } A^\dagger X A + \frac{ \Norm{A}_F^2 }{T} \sum_{i \in [m], j \in [n] } X_{i,i} e_j  e_j^\dagger \\
& = \Paren[\Big]{ 1- \frac{1}{T} } A^\dagger X A + \frac{ \Norm{A}_F^2 \Tr(X) }{T}  I . \qedhere
\end{align*}
\end{proof}

We list a simple consequence of these bounds that we use later.
\begin{corollary} \label{cor:ssketch-simple}
For a matrix $A \in \C^{m\times n}$, let $M = \ssketch(A)$ with parameter $T$.
Then, for matrices $X \in \C^{\ell \times m}$ and $Y \in \C^{n\times d}$,
\begin{align*}
    \Pr\Big[\fnorm{XMY - XAY} \geq \frac{\fnorm{X}\fnorm{A}\fnorm{Y}}{\sqrt{\delta T}}\Big] \leq \delta.
\end{align*}
\end{corollary}

\subsection{\texorpdfstring{Approximate Matrix Products from $\ell^2_2$ Importance Sampling}{Approximate Matrix Products from Importance Sampling}}
\label{subsec:assym-amp}

\ewin{This section checked!}

In this subsection, we extend the well-known approximate matrix product to the setting where we have an $\ell_2^2$-sampling oracle~\cite{w14}.
The approximate matrix product guarantee is typically achieved in the oblivious sketching model~\cite{cohen2016optimal}, which we cannot extend to the quantum setting.
Earlier work~\cite{dkm06a} considers achieving this guarantee through row subsampling, which can be performed in this setting.

\begin{definition} \label{def:aamp}
    Given two matrices $A \in \mathbb{C}^{m \times n}$ and $B \in \mathbb{C}^{d \times n}$, along with a probability distribution $p \in \mathbb{R}_{\geq 0}^n$, we define the \emph{Asymmetric Approximate Matrix Product} of sketch size $s$, denoted $\aamp_s(A, B^\dagger, p)$, to be the $n\times s$ matrix whose columns are i.i.d.\ sampled according to the law
    \begin{align*}
        [\aamp_s(A, B^\dagger, p)]_{*,j} = \frac{e_k}{\sqrt{s\cdot p_k}} \text{ with probability } p_k
    \end{align*}
    For an $S = \aamp_s(A, B^\dagger, p)$, we will typically consider the expression $ASS^\dagger B^\dagger$, which can be written as the sum of independent rank-one matrices $\frac{1}{s\cdot p_k}A_{*,k}B_{*,k}$.
    Notice that $\E[ASS^\dagger B^\dagger] = AB^\dagger$.
\end{definition}

We show that $\aamp_s(A, B^\dagger, p)$ is a good approximate matrix product sketch for $AB^\dagger$ when the distribution $p$ is a mixture of the row sampling distributions for $A$ and $B$.

\begin{theorem}[Asymmetric Approximate Matrix Multiplication]
\label{thm:stable_amm}
    Given matrices $A \in \C^{m \times n}$ and $B \in \C^{n \times d}$, consider $S = \aamp_s(A, B^\dagger, p)$ for $p \in \mathbb{R}_{\geq 0}^n$ with $p_i \geq \frac{1}{2\phi}(\frac{\norm{A_{*,i}}^2}{\fnorm{A}^2} + \frac{\norm{B_{*,i}}^2}{\fnorm{B}^2})$ for some $\phi \geq 1$. Let $\sr = \frac{\norm{A}_F^2 }{ \norm{A}^2 } + \frac{\norm{B}_F^2 }{ \norm{B}^2 } $.
    Then, with probability at least $1-\delta > 0.75$,
    \begin{equation*}
        \norm*{ ASS^\dagger B^\dagger - AB^\dagger} \leq \sqrt{ \frac{2}{s}\log\Paren[\Big]{\frac{\sr}{\delta}}\phi\Paren{\norm{A}_F^2\norm{B}^2 + \norm{A}^2\norm{B}_F^2 } } + \frac{1}{s}\log\Paren[\Big]{\frac{\sr}{\delta}}\phi\fnorm{A}\fnorm{B}.
    \end{equation*}
\end{theorem}

We will use the following consequence of this theorem.
\begin{corollary} \label{cor:amm-usable}
    Given matrices $A \in \C^{m \times n}$ and $B \in \C^{n \times d}$, consider $S = \aamp_s(A, B^\dagger, p)$ for $p \in \mathbb{R}_{\geq 0}^n$ with $p_i \geq \frac{1}{2\phi}(\frac{\norm{A_{*,i}}^2}{\fnorm{A}^2} + \frac{\norm{B_{*,i}}^2}{\fnorm{B}^2})$ for some $\phi \geq 1$.
    For $\eps \in (0,1]$ and $\delta \in (0, 0.25]$, when $s = \frac{4\phi}{\eps^2}(\frac{\fnorm{A}^2}{\norm{A}^2} + \frac{\fnorm{B}^2}{\norm{B}^2})\log(\frac{1}{\delta}(\frac{\fnorm{A}^2}{\norm{A}^2} + \frac{\fnorm{B}^2}{\norm{B}^2}))$, then $\norm*{ASS^\dagger B^\dagger - AB^\dagger} \leq \eps\norm{A}\norm{B}$ with probability $\geq 1-\delta$.
\end{corollary}

The symmetric version of this result was previously stated in \cite{kv17,rudelson2007sampling}, and this asymmetric version was stated in \cite{magen2011low}.
However, the final theorem statement was weaker, so we reprove it here.
To obtain it, we prove the following key lemma: 

\begin{lemma}[Concentration of asymmetric random outer products]
\label{lem:asymmetric-outer-products}
Let $\{ (X_i, Y_i) \}_{i \in [s]}$ be $s$ independent copies of the tuple of random vectors $(X, Y)$, with $X \in \C^m$ and $Y \in \C^d$.
In particular, $(X, Y) = (a^{(i)}, b^{(i)})$ with probability $p_i$ for $i \in [n]$.
Let $M, L \geq 0$ be such that
\begin{align*}
    L&\geq \max_{i \in [n]} \norm*{a^{(i)}(b^{(i)})^\dagger}, \\
    M^2&\geq \max_{i \in [n]} \norm*{b^{(i)}}^2 \norm*{ \expec{}{X X^\dagger}} + \max_{i \in [n]} \norm*{a^{(i)}}^2 \norm*{ \expec{}{Y Y^\dagger}},\\
    \sr &\geq \frac{\max_{i \in [n]} \norm*{b^{(i)}}^2\expec{}{\fnorm{X}^2} + \max_{i \in [n]} \norm*{a^{(i)}}^2\expec{}{\fnorm{Y}^2}}{\max(\max_{i \in [n]} \norm*{b^{(i)}}^2 \norm*{ \expec{}{X X^\dagger}}, \max_{i \in [n]} \norm*{a^{(i)}}^2 \norm*{ \expec{}{Y Y^\dagger}})}.
\end{align*}
Then, for any $t \geq M/\sqrt{s} + 2L/(3s)$,
\begin{equation*}
    \Pr\bracks[\Big]{ \Norm[\Big]{ \frac{1}{s} \sum_{i \in [s]} X_iY_i^\dagger - \expec{}{XY^\dagger}} \geq t }
    \leq 4 \sr \exp\Paren[\Big]{ - \frac{st^2}{2(M^2 + Lt)} }.
\end{equation*}
\end{lemma}
\begin{proof}
For $i \in [s]$, let $Z_i  = \frac{1}{s}\Paren{ X_iY_i^\dagger - \expec{}{XY^\dagger} } $.
Then $\norm{ Z_i  } \leq \frac{2}{s}\norm{ X_iY_i^\dagger } \leq \frac{2L}{s}$. Next, we bound the variance,
\begin{equation*}
\begin{split}
    \sigma^2 & \coloneqq \max\Paren[\Bigg]{ \underbrace{  \norm[\Big]{ \sum_{i \in [n]} \expec{}{ Z_i Z_i^\dagger } } }_{\text{(\emph{i})}} , \underbrace{ \norm[\Big]{ \sum_{i \in [n]} \expec{}{ Z_i^\dagger Z_i } }}_{\text{(\emph{ii})}} }.
\end{split}
\end{equation*} 
We can observe that
\begin{align*}
    \sum_{i \in [s]} \expec{}{ Z_i Z_i^\dagger }
    &= \frac{1}{s}  \expecf{}{  \Paren*{ X_iY_i^\dagger - \expec{}{XY^\dagger } } \Paren*{ X_iY_i^\dagger - \expec{}{XY^\dagger } }^\dagger } \\
    &= \frac{1}{s} \expecf{}{ \norm{Y_i}^2 X_i X_i^\dagger - \expec{}{XY^\dagger}\expec{}{YX^\dagger}} \\
    &\preceq \frac{1}{s} (\max_{i \in [n]} \norm*{b^{(i)}}^2) \expec{}{XX^\dagger} =: V_1 \\
    \sum_{i \in [s]} \expec{}{ Z_i^\dagger  Z_i }
    &= \frac{1}{s} \expecf{}{ \Paren*{ X_iY_i^\dagger - \expec{}{XY^\dagger} }^\dagger \Paren*{ X_iY_i^\dagger - \expec{}{XY^\dagger} } } \\
    &= \frac{1}{s}  \expecf{}{ \norm{X_i}^2 Y_iY_i^\dagger - \expec{}{YX^\dagger}\expec{}{XY^\dagger} } \\
    &\preceq \frac{1}{s} (\max_{i \in [n]} \norm*{a^{(i)}}^2) \expec{}{YY^\dagger} =: V_2
\end{align*}
We can use this to bound term (\emph{i}):
\begin{align*}
    \norm[\Big]{ \sum_{i \in [s]} \expec{}{ Z_i Z_i^\dagger } }
    \leq \frac{1}{s} \norm{ \expec{}{ \norm{Y_i}^2  X_i X_i^\dagger } }
    \leq \frac{1}{s} (\max_{i \in [n]} \norm*{b^{(i)}}^2) \norm*{\expec{}{XX^\dagger}}.
\end{align*}
We bound term (\emph{ii}) as follows:
\begin{align*}
    \norm[\Big]{ \sum_{i \in [s]} \expec{}{ Z_i^\dagger  Z_i  } }
    \leq \frac{1}{s } \norm{ \expec{}{ \norm{X_i}^2  Y_i Y_i^\dagger } }
    \leq \frac{1}{s} (\max_{i \in [n]} \norm*{a^{(i)}}^2) \norm*{\expec{}{YY^\dagger}}.
\end{align*}
Altogether, we have shown that $\sigma^2 \leq M^2/s$.
Applying Matrix Bernstein (see Fact \ref{fact:matrix-bernstein}) with upper bounds of $V_1$ and $V_2$ and parameters $L \leftarrow \frac{2L}{s}$ and $v \leftarrow M^2/s$, we get
\begin{multline*}
    \Pr\Bigg[\Norm[\Big]{ \frac{1}{s} \sum_{i \in [s]} X_iY_i^\dagger - \expec{}{XY^\dagger}} \geq t \Bigg]
    = \Pr\Bigg[\norm[\Big]{ \sum_{i \in [s]} Z_i } \geq t\Bigg] \\
    \leq 4 \sr \exp\Paren[\Big]{ - \frac{t^2/2 }{M^2/s + 2Lt/(3s) } }
    \leq 4 \sr \exp\Paren[\Big]{ - \frac{st^2}{2(M^2 + Lt)} },
\end{multline*}
where
\begin{equation*}
    \sr = \frac{\tr(V_1) + \tr(V_2)}{\max(\norm{V_1}, \norm{V_2})}
    = \frac{\max_{i \in [n]} \norm*{b^{(i)}}^2\expec{}{\fnorm{X}^2} + \max_{i \in [n]} \norm*{a^{(i)}}^2\expec{}{\fnorm{Y}^2}}{\max(\max_{i \in [n]} \norm*{b^{(i)}}^2 \norm*{ \expec{}{X X^\dagger}}, \max_{i \in [n]} \norm*{a^{(i)}}^2 \norm*{ \expec{}{Y Y^\dagger}})}. \qedhere
\end{equation*}
\end{proof}

\begin{fact}[Intrinsic Matrix Bernstein, {\cite[Theorem~7.3.1]{tropp15}}]
\label{fact:matrix-bernstein}
Consider a finite sequence $ \Set{ Z_{k} }_{k \in [s]}$ of random complex matrices with the same size, and assume that $\expec{}{Z_k} =0$ and $\norm{Z_k} \leq L$.
Let $V_1$ and $V_2$ be semidefinite upper bounds for the corresponding matrix-valued variances:
\begin{align*}
    V_1 &\succeq \expec{}{ \sum_{k=1}^s Z_k Z_k^\dagger }; &
    V_2 &\succeq \expec{}{ \sum_{k=1}^s Z_k^\dagger Z_k }.
\intertext{Define an intrinsic dimension bound and a variance bound,}
    \sr &= \frac{\tr(V_1 + V_2)}{\max(\norm{V_1}, \norm{V_2})}; &
    v &= \max\{\norm{V_1}, \norm{V_2}\}.
\end{align*}
Then, for $t\geq \sqrt{v} + L/3$,
\begin{equation*}
\Pr\Bigg[ \norm[\Big]{ \sum_{i \in [k]} Z_i  } \geq t  \Bigg]  \leq 4\sr \exp\Paren[\Big]{ - \frac{t^2/2}{ v + Lt /3 } }.
\end{equation*}
\end{fact}

It is now straight-forward to prove \cref{thm:stable_amm} using the aforementioned lemma:

\begin{proof}[Proof of Theorem \ref{thm:stable_amm}]
We apply \cref{lem:asymmetric-outer-products} with $a^{(i)} = \sqrt{1/p_i} \cdot A_{*,i} $ and $b^{(i)} = \sqrt{1/p_i} \cdot B_{*,i}$.
As assumed the sampling distribution $p_i$ satisfies $p_i \geq \frac{1}{2\phi}(\frac{\norm{A_{*,i}}^2}{\fnorm{A}^2} + \frac{\norm{B_{*,i}}^2}{\fnorm{B}^2}) \geq \frac{\norm{A_{*,i}}{\norm{B_{*,i}}}}{\phi\fnorm{A}\fnorm{B}}$, so
\begin{align*}
    \norm*{a^{(i)}} &= \norm{A_{*,i}} /\sqrt{p_i}
    \leq \norm{A_{*,i}}\sqrt{\frac{2\phi\fnorm{A}^2}{\norm{A_{*,i}}^2}}
    = \sqrt{2\phi}\fnorm{A};\\
    \norm*{b^{(i)}} &= \norm{B_{*,i}} /\sqrt{p_i}
    \leq \norm{B_{*,i}}\sqrt{\frac{2\phi\fnorm{B}^2}{\norm{B_{*,i}}^2}}
    = \sqrt{2\phi}\fnorm{B}; \\
    \norm*{a^{(i)}(b^{(i)})^\dagger } &= \frac{\norm{A_{*,i}}\norm{B_{*,i}}}{p_i} \leq \phi\fnorm{A}\fnorm{B}.
\end{align*}
Further, 
\begin{align*}
    \norm*{\expec{}{XX^\dagger}}
    &= \norm[\Big]{\sum_{i \in [n]} A_{*,i}A_{*,i}^\dagger}
    = \norm*{AA^\dagger}
    = \norm*{A}^2; \\
    \norm*{\expec{}{YY^\dagger}}
    &= \norm[\Big]{\sum_{i \in [n]}B_{*,i}B_{*,i}}
    = \norm*{BB^\dagger}
    = \norm*{B}^2.
\end{align*}
Finally, 
\begin{align*}
    \frac{\max_{i \in [n]} \norm*{b^{(i)}}^2\expec{}{\fnorm{X}^2} + \max_{i \in [n]} \norm*{a^{(i)}}^2\expec{}{\fnorm{Y}^2}}{\max(\max_{i \in [n]} \norm*{b^{(i)}}^2 \norm*{ \expec{}{X X^\dagger}}, \max_{i \in [n]} \norm*{a^{(i)}}^2 \norm*{ \expec{}{Y Y^\dagger}})}
    &\leq \frac{\expec{}{\fnorm{X}^2}}{\norm*{ \expec{}{X X^\dagger}}} + \frac{\expec{}{\fnorm{Y}^2}}{\norm*{ \expec{}{Y Y^\dagger}}} \\
    &= \frac{\fnorm{A}^2}{\norm*{A}^2} + \frac{\fnorm{B}^2}{\norm{B}^2}.
\end{align*}
So, in \cref{lem:asymmetric-outer-products}, we can set $L = \phi\fnorm{A}\fnorm{B}$, $M^2 = 2\phi\fnorm{A}^2\norm{B}^2 + 2\phi\fnorm{B}^2\norm{A}^2$, and $\sr = \frac{\fnorm{A}^2}{\norm{A}^2} + \frac{\fnorm{B}^2}{\norm{B}^2}$ to get that, for all $t \geq M/\sqrt{s} + 2L/(3s)$,
\begin{equation*}
\begin{split}
    & \Pr\left[ \Norm[\Big]{ \frac{1}{s} \sum_{i \in [s]} X_iY_i^\dagger - \expec{}{XY^\dagger}} \geq t \right] \\
    &  \leq 4 \Paren{ \frac{ \fnorm{A}^2 }{ \norm{A} } + \frac{ \fnorm{B}^2  }{ \norm{B} } }  \exp\Paren[\Big]{ \frac{-st^2}{2\phi(\fnorm{A}^2\norm{B}^2 + \norm{A}^2\fnorm{B}^2) + \phi\fnorm{A}\fnorm{B}t} }.
\end{split}
\end{equation*}
To get the right-hand side of the above equation to be $\leq \delta$, it suffices to choose
\begin{equation*}
    t = \sqrt{ \frac{2}{s}\log\Paren[\Big]{\frac{\sr}{\delta}}\phi\Paren{\norm{A}_F^2\norm{B}^2 + \norm{A}^2\norm{B}_F^2 } } + \frac{1}{s}\log\Paren[\Big]{\frac{\sr}{\delta}}\phi\fnorm{A}\fnorm{B}.
\end{equation*}
This choice of $t$ is greater than $M/\sqrt{s} + 2L/(3s)$ when $\delta < 1/e$, so with this assumption, we can conclude that with probability $\geq 1-\delta$,
\begin{equation*}
\norm*{ ASS^\dagger B^\dagger - AB^\dagger} \leq \sqrt{ \frac{2}{s}\log\Paren[\Big]{\frac{\sr}{\delta}}\phi\Paren{\norm{A}_F^2\norm{B}^2 + \norm{A}^2\norm{B}_F^2 } } + \frac{1}{s}\log\Paren[\Big]{\frac{\sr}{\delta}}\phi\fnorm{A}\fnorm{B}. \qedhere
\end{equation*}
\end{proof}


\section{Sums of Chebyshev Coefficients}
\label{sec:chebsums}

To give improved stability bounds for the Clenshaw recurrence, we need to bound various sums of Chebyshev coefficients.
Since we aim to give bounds that hold for all degree-$d$ polynomials, we use no property of the function beyond that it has a unique Chebyshev expansion; of course, for any particular choice of function $f$, the bounds in this section can be improved by explicitly computing its Chebyshev coefficients, or in some cases, by using smoothness properties of the function \cite[Theorems~7.2 and 8.2]{trefethen19}.

Let $f: [-1,1] \to \mathbb{R}$ be a Lipschitz continuous function.
Then it can be expressed uniquely as a linear combination of Chebyshev polynomials $f(x) = \sum_{i=0}^\infty a_i T_i(x)$.
A broad topic of interest in approximation theory is bounds for linear combinations of these coefficients, $\sum a_ic_i$, in terms of $\supnorm{f}$; this was one motivation of Vladimir Markov in proving the Markov brothers' inequality~\cite[p575]{schaeffer41}.
Our goal for this section will be to investigate this question in the case where these sums are arithmetic progressions of step four.
This will be necessary for later stability analyses, and is one of the first non-trivial progressions to bound.
We begin with some straightforward assertions (see~\cite{trefethen19} for background).

\begin{fact} \label{parity-two-chebsums}
Let $f: [-1,1] \to \mathbb{R}$ be a Lipschitz continuous function.
Then its Chebyshev coefficients $\{a_\ell\}_\ell$ satisfy the following bounds.
\begin{equation*}
\begin{split}
    \abs*[\Big]{\sum_\ell a_\ell} &= \abs{f(1)} \leq \supnorm{f} \\
    \abs*[\Big]{\sum_\ell (-1)^\ell a_\ell} &= \abs{f(-1)} \leq \supnorm{f} \\
    \abs*[\Big]{\sum_\ell a_\ell \iver{\ell \text{ is even}}} &= \abs*[\Big]{\sum_{\ell} a_\ell \frac12(1 + (-1)^\ell)} \leq \supnorm{f} \\
    \abs*[\Big]{\sum_\ell a_\ell \iver{\ell \text{ is odd}}} &= \abs*[\Big]{\sum_\ell a_\ell \frac12(1 - (-1)^\ell)} \leq \supnorm{f}
\end{split}
\end{equation*}
\end{fact}

We use the following result on Lebesgue constants to bound truncations of the Chebyshev coefficient sums.
\begin{lemma}[{\cite[Theorem 15.3]{trefethen19}}] \label{projection-lebesgue}
    Let $f: [-1,1] \to \mathbb{R}$ be a Lipschitz continuous function, let $f_k(x) = \sum_{\ell=0}^k a_\ell T_\ell(x)$, and let the optimal degree-$k$ approximating polynomial to $f$ be denoted $f_k^*$.
    Then
    \begin{align*}
        \supnorm{f - f_k} &\leq \Big(4 + \frac{4}{\pi^2}\log(k+1)\Big)\supnorm{f - f_k^*} \\
        &\leq \Big(4 + \frac{4}{\pi^2}\log(k+1)\Big)\supnorm{f}.
    \end{align*}
    Similarly,
    \begin{align*}
        \supnorm{f_k} &\leq \supnorm{f-f_k} + \supnorm{f}
        \leq \Big(5 + \frac{4}{\pi^2}\log(k+1)\Big)\supnorm{f}.
    \end{align*}
\end{lemma}

This implies bounds on sums of coefficients.
\begin{fact} \label{parity-two-tailsums}
    Consider a function $f(x) = \sum_\ell a_\ell T_\ell(x)$.
    Then
    \begin{align*}
        \abs*[\Big]{\sum_{\ell=k}^\infty a_\ell \iver{\ell-k \text{ is even}}}
        &\leq \supnorm{f - f_{k-1}}
        \leq \Big(4 + \frac{4}{\pi^2}\log(k)\Big)\supnorm{f}, \\
        \abs*[\Big]{\sum_{\ell=k}^\infty a_\ell (-1)^\ell}
        &\leq \supnorm{f - f_{k-1}}
        \leq \Big(4 + \frac{4}{\pi^2}\log(k)\Big)\supnorm{f},
    \end{align*}
    where the inequalities follow from \cref{parity-two-chebsums} and \cref{projection-lebesgue}.
    When $k = 0$, then the sum is bounded by $\supnorm{f}$, as shown in \cref{parity-two-chebsums}.
\end{fact}

Now, we prove similar bounds in the case that $f(x)$ is an odd function. In particular, we want to obtain a bound on alternating signed sums of the Chebshyev coefficients and we incur a blowup that scales logarithmically in the degree. 

\begin{lemma} \label{lem:odd-flipped-coefs}
    Let $f: [-1,1] \to \mathbb{R}$ be an \emph{odd} Lipschitz continuous function with Chebyshev coefficients $\{a_\ell\}_\ell$, so that $a_k = 0$ for all even $k$.
    Then the Chebyshev coefficient sum is bounded as
    \begin{align*}
        \abs*[\Big]{\sum_{\ell=0}^d (-1)^\ell a_{2\ell+1}}
        &\leq  (\log(d) + 2)\max_{0 \leq k \leq 2d+1} \supnorm{f_k}\\
        &\leq (\log(d) + 2)\Big(5 + \frac{4}{\pi^2}\log(2d+2)\Big)\supnorm{f} \\
        &\leq \Big(16 + 4\log^2(d+1)\Big)\supnorm{f}.
    \end{align*}
\end{lemma}

We first state the following relatively straight-forward corollary:

\begin{corollary} \label{cor:four-step-sums}
    \cref{lem:odd-flipped-coefs} gives bounds on arithmetic progressions with step size four.
    Let $f: [-1,1] \to \mathbb{R}$ be a Lipschitz continuous function, and consider non-negative integers $c \leq d$.
    Then
    \begin{align*}
        \abs*[\Big]{\sum_{\ell=c}^d a_\ell \iver{\ell - c \equiv 0 \Mod{4}}}
        &\leq (32 + 8\log^2(d+1))\supnorm{f}.
    \end{align*}
\end{corollary}
\begin{proof}
    Define $f^\odd \coloneqq \tfrac12(f(x) - f(-x))$ and $f^\even \coloneqq \tfrac12(f(x) + f(-x))$ to be the odd and even parts of $f$ respectively.
    Triangle inequality implies that $\supnorm{f^\odd}, \supnorm{f^\even} \leq \supnorm{f}$.
    Suppose $c, d$ are odd.
    Then
    \begin{align*}
        &\qquad \abs*[\Big]{\sum_{\ell=c}^d a_\ell \iver{\ell - c \equiv 0 \Mod{4}}} \\
        &= \frac12\abs*[\Big]{\sum_{\ell=0}^{\lfloor(d-c)/2\rfloor} a_{c + 2\ell}(1 \pm (-1)^\ell)} \\
        &\leq \frac12\Big(\abs*[\Big]{\sum_{\ell=0}^{\lfloor(d-c)/2\rfloor} a_{c + 2\ell}} + \abs*[\Big]{\sum_{\ell=0}^{\lfloor(d-c)/2\rfloor} (-1)^\ell a_{c + 2\ell}}\Big) \\
        &= \frac12\Big(\abs*[\Big]{f_d^\odd(1) - f_{c-2}^\odd(1)} + \abs*[\Big]{\sum_{\ell = 0}^{(d-1)/2} (-1)^\ell a_{2\ell + 1} - \sum_{\ell = 0}^{(c-3)/2} (-1)^\ell a_{2\ell + 1}}\Big) \\
        &\leq \frac12\Big(\supnorm{f_{c-2}^\odd} + \supnorm{f_d^\odd} + 2(\log(d) + 2)\max_{0 \leq k \leq d}\supnorm{f_{k}^{\mathrm{odd}}}\Big) \\
        &\leq (32 + 8\log^2(d+1))\supnorm{f^\odd} \\
        &\leq (32 + 8\log^2(d+1))\supnorm{f}.
    \end{align*}
    The case when $c$ is even is easier: by \cref{eq:t-composition}, we know that
    \begin{align*}
        \supnorm[\Big]{\sum_\ell a_{2\ell} T_{\ell}(x)}
        = \supnorm[\Big]{\sum_\ell a_{2\ell} T_{\ell}(T_2(x))}
        = \supnorm[\Big]{\sum_\ell a_{2\ell} T_{2\ell}(x)}
        = \supnorm[\Big]{f^\even(x)}
        \leq \supnorm{f},
    \end{align*}
    so by \cref{parity-two-tailsums},
    \begin{align}
    \abs*[\Big]{\sum_{\ell\geq c} a_\ell \iver{\ell - c \equiv 0 \Mod{4}}}
        &= \abs*[\Big]{\sum_{\ell\geq c/2} a_{2\ell} \iver{\ell - c/2 \text{ is even}}} \nonumber\\
        &\leq \Big(4 + \frac{4}{\pi^2}\log(c/2-1)\Big)\supnorm[\Big]{\sum_\ell a_{2\ell} T_{\ell}(x)} \nonumber\\
        &\leq \Big(4 + \frac{4}{\pi^2}\log(c/2-1)\Big)\supnorm{f}.
    \end{align}
    From the above, we can bound the type of sums in the problem statement, paying an additional factor of two:
    \begin{align}
        \abs*[\Big]{\sum_{\ell=c}^d a_\ell \iver{\ell - c \equiv 0 \Mod{4}}}
        &\leq \abs*[\Big]{\sum_{\ell\geq c} a_\ell \iver{\ell - c \equiv 0 \Mod{4}}} + \abs*[\Big]{\sum_{\ell\geq d+1} a_\ell \iver{\ell - c \equiv 0 \Mod{4}}} \nonumber\\
        &\leq \Big(8 + \frac{4}{\pi^2}(\log(c/2-1) + \log(d/2+1))\Big)\supnorm{f},
    \end{align}
    giving the desired bound.
\end{proof}

We note that \cref{lem:odd-flipped-coefs} will be significantly harder to prove.
See \cref{why-sum-hard} for an intuitive explanation why. We begin with two structural lemmas on how the solution to a unitriangular linear system behaves, which might be of independent interest.

\begin{lemma}[An entry-wise positive solution]\label{a-inv-b-nonneg}
    Suppose that $A \in \mathbb{R}^{d\times d}$ is an upper unitriangular matrix such that, for all $i \leq j$, $A_{i,j} > 0$, $A_{i,j} > A_{i-1, j}$.
    Then $A^{-1}\vec{1}$ is a vector with positive entries.

    The same result holds when $A$ is a lower unitriangular matrix such that, for all $i \geq j$, $A_{i,j} > 0$, $A_{i,j} > A_{i+1, j}$.
\end{lemma}
\begin{proof}
    Let $x = A^{-1}\vec{1}$.
    Then $x_d = 1 \geq 0$.
    The result follows by induction:
    \begin{align*}
        x_{i} &= 1 - \sum_{j=i+1}^d A_{i,j}x_j \\
        &= \sum_{j=i+1}^d (A_{i+1,j} - A_{i,j})x_j + 1 - \sum_{j=i+1}^d A_{i+1,j}x_j \\
        &= \sum_{j=i+1}^d (A_{i+1,j} - A_{i,j})x_j + 1 - [Ax]_{i+1} \\
        &= \sum_{j=i+1}^d (A_{i+1,j} - A_{i,j})x_j \\
        &> 0.
    \end{align*}
    For lower unitriangular matrices, the same argument follows.
    The inverse satisfies $x_1 = 1$ and
    \begin{align*}
        x_{i} &= 1 - \sum_{j=1}^{i-1} A_{i,j}x_j \\
        &= \sum_{j=i+1}^d (A_{i-1,j} - A_{i,j})x_j + 1 - \sum_{j=i+1}^d A_{i-1,j}x_j
        > 0. \qedhere
    \end{align*}
\end{proof}

Next, we characterize how the solution to a unitriangular linear system behaves when we consider a partial ordering on the matrices. 

\begin{lemma} \label{bounding-the-recurrence}
    Let $A$ be a non-negative upper unitriangular matrix such that $A_{i,j} > A_{i-1,j}$ and $A_{i,j} > A_{i,j+1}$ for all $i \leq j$.
    Let $B$ be a matrix with the same properties, such that $A \geq B$ entrywise.
    By \cref{a-inv-b-nonneg}, $x^{(A)} = A^{-1}\vec{1}$ and $x^{(B)} = B^{-1}\vec{1}$ are non-negative.
    It further holds that $\sum_{i=1}^d [A^{-1}\vec{1}]_i \leq \sum_{i=1}^d [B^{-1}\vec{1}]_i$.
\end{lemma}
\begin{proof}
We consider the line between $A$ and $B$, $A(t) = A(1-t) + Bt$ for $t \in [0,1]$.
Let $x(t) = A^{-1}\vec{1}$; we will prove that $\vec{1}^\dagger x(t)$ is monotonically increasing in $t$.
The gradient of $x(t)$ has a simple form~\cite{tao2013matrix}:
\begin{align*}
    A(t) x(t) &= \vec{1} \\
    \partial [A(t) x(t)] &= \partial_t[\vec{1}] \\
    (B-A) x(t) + A(t) \partial_t x(t) &= 0 \\
    \partial_t x(t) &= A^{-1}(t)(A-B) x(t).
\end{align*}
So,
\begin{align*}
    \vec{1}^\dagger \partial_t x(t) &= \vec{1}^\dagger A^{-1}(t)(A-B) A^{-1}(t)\vec{1} \\
    &= [([A(t)]^{-1})^\dagger\vec{1}]^\dagger (A-B) [A^{-1}(t)\vec{1}].
\end{align*}
Since $A$ and $B$ satisfy the entry constraints, so do every matrix along the line.
Consequently, the column constraints in \cref{a-inv-b-nonneg} are satisfied for both $A$ and $A^\dagger$, so both $([A(t)]^{-1})^\dagger\vec{1}$ and $A^{-1}(t)\vec{1}$ are positive vectors.
Since $A \geq B$ entrywise, this means that $\vec{1}^\dagger \partial_t x(t)$ is positive, as desired.
\end{proof}

\begin{proof}[Proof of \cref{lem:odd-flipped-coefs}]
We first observe that the following sorts of sums are bounded.
Let $x_k := \cos(\frac{\pi}{2}(1 - \frac{1}{2k+1}))$.
Then, using that $T_\ell(\cos(x)) = \cos(\ell x)$,
\begin{align*}
    f_{2k+1}(x_k) &= \sum_{\ell=0}^{2k+1} a_\ell T_{\ell}(x_k) \\
    &= \sum_{\ell=0}^k a_{2\ell+1} T_{2\ell+1}(x_k) \\
    &= \sum_{\ell=0}^k a_{2\ell+1} \cos\Big(\frac{\pi}{2}\Big(2\ell + 1 - \frac{2\ell+1}{2k+1}\Big)\Big) \\
    &= \sum_{\ell=0}^k (-1)^\ell a_{2\ell+1}\sin\Big(\frac{\pi}{2}\frac{2\ell+1}{2k+1}\Big).
\end{align*}
We have just shown that
\begin{align}
    \abs*[\Big]{\sum_{\ell=0}^k a_{2\ell+1}(-1)^\ell\sin\Big(\frac{\pi}{2}\frac{2\ell+1}{2k+1}\Big)} \leq \supnorm{f_{2k+1}}.
\end{align}
We now claim that there exist non-negative $c_k$ for $k \in \{0,1,\ldots,d\}$ such that
\begin{align} \label{odd-coef-lincomb}
    \sum_{\ell=0}^d (-1)^\ell a_{2\ell+1}
    &= \sum_{k=0}^d c_k f_{2k+1}(x_k).
\end{align}
The $f_{2k+1}(x_k)$'s can be bounded using \cref{projection-lebesgue}.
The rest of the proof will consist of showing that the $c_k$'s exist, and then bounding them.

To do this, we consider the coefficient of each $a_{2\ell+1}$ separately; let $A^{(k)} \in [0,1]^{d+1}$ (index starting at zero) be the vector of coefficients associated with $p_{2k+1}(x_k)$:
\begin{align}
    A^{(k)}_\ell = \sin\Big(\frac{\pi}{2}\frac{2\ell+1}{2k+1}\Big) \text{ for } 0 \leq \ell \leq k,\,0 \text{ otherwise}.
\end{align}
Note that the $A^{(k)}_\ell$ is always non-negative and increasing with $\ell$ up to $A_k^{(k)} = 1$.
Then \cref{odd-coef-lincomb} holds if and only if
\begin{align*}
    c_0 A^{(0)} + \cdots + c_d A^{(d)} = \vec{1},
\end{align*}
or in other words, the equation $Ac = \vec{1}$ is satisfied, where $A$ is the matrix with columns $A^{(k)}$ and $c$ is the vector of $c_\ell$'s.
Since $A$ is upper triangular (in fact, with unit diagonal), this can be solved via backwards substitution: $c_d = 1$, then $c_{d-1}$ can be deduced from $c_d$, and so on.
More formally, the $s$th row gives the following constraint that can be rewritten as a recurrence.
\begin{gather}
    \sum_{t=s}^d \sin\Big(\frac{\pi}{2}\frac{2s+1}{2t+1}\Big)c_t = 1 \\
    c_s = 1 - \sum_{t=s+1}^d \sin\Big(\frac{\pi}{2}\frac{2s+1}{2t+1}\Big)c_t \label{odd-cheb-recurrence}
\end{gather}
Because the entries of $A$ increase in $\ell$, the $c_\ell$'s are all positive.

Invoking \cref{a-inv-b-nonneg} with the matrix A establishes that such $c_s$ exist; our goal now is to bound them.
Doing so is not as straightforward as it might appear: since the recurrence \cref{odd-cheb-recurrence} \emph{subtracts} by $c_t$'s, an upper bound on $c_t$ for $t \in [s+1,d]$ does not give an upper bound on $c_s$; it gives a lower bound.
So, an induction argument to show bounds for $c_s$'s fails.
Further, we were unable to find any closed form for this recurrence.
However, since all we need to know is the sum of the $c_s$'s, we show that we \emph{can} bound this via a generic upper bound on the recurrence.

Here, we apply \cref{bounding-the-recurrence}  to $A$ as previously defined, and the bounding matrix is (for $i \leq j)$
\begin{align*}
    B_{i,j} &= \frac ij \leq \frac{2i+1}{2j+1}
    \leq \sin\Big(\frac{\pi}{2}\frac{2i+1}{2j+1}\Big)
    = A_{i,j},
\end{align*}
using that $\sin(\frac{\pi}{2}x) \geq x$ for $x \in [0,1]$.
Let $\hat{c} = B^{-1}\vec{1}$.
Then $\hat{c}_i = \frac{1}{i+1}$ for $i \neq d$ and $\hat{c}_d = 1$.
\begin{align*}
    [B\hat{c}]_i
    = \sum_{j=i}^d B_{i,j}\hat{c}_j
    = \sum_{j=i}^{d-1} \frac ij \frac{1}{j+1} + \frac{i}{d}
    = i\sum_{j=i}^{d-1} \Big(\frac{1}{j} - \frac{1}{j+1}\Big) + \frac{i}{d}
    = i\Big(\frac{1}{i} - \frac{1}{d}) + \frac{i}{d}
    = 1
\end{align*}
By \cref{bounding-the-recurrence}, $\sum_i c_i \leq \sum_i \hat{c}_i \leq \log(d) + 2$.
So, altogether, we have
\begin{align*}
    \abs*[\Big]{\sum_{\ell=0}^d (-1)^\ell a_{2\ell+1}}
    &= \abs*[\Big]{\sum_{k=0}^d c_k f_{2k+1}(x_k)} \\
    &\leq \sum_{k=0}^d c_k \supnorm{f_{2k+1}} \\
    &\leq \Big(\sum_{k=0}^d c_k\Big)\max_{0 \leq k \leq d}\supnorm{f_{2k+1}} \\
    &= \Big(\sum_{k=0}^d c_k\Big)\max_{0 \leq k \leq 2d+1}\supnorm{f_k} \\
    &\leq (\log(d)+2)\max_{0 \leq k \leq 2d+1}\supnorm{f_k}. \qedhere
\end{align*}
\end{proof}

\begin{remark} \label{why-sum-hard}
    A curious reader will (rightly) wonder whether this proof requires this level of difficulty.
    Intuition from the similar Fourier analysis setting suggests that arithmetic progressions of any step size at any offset are easily bounded.
    We can lift to the Fourier setting by considering, for an $f: [-1,1] \to \mathbb{R}$, a corresponding $2\pi$-periodic $g: [0,2\pi] \to \mathbb{R}$ such that
    \begin{align*}
        g(\theta) \coloneqq f(\cos(\theta)) = \sum_{k=0}^\infty a_kT_k(\cos(\theta)) = \sum_{k=0}^\infty a_k \cos(k\theta) = \sum_{k=0}^\infty a_k \frac{e^{ik\theta} + e^{-ik\theta}}{2}.
    \end{align*}
    This function has the property that $\abs{g(\theta)} \leq \supnorm{f}$ and $\widehat{g}(k) = a_{\abs{k}}/2$ (except $\widehat{g}(0) = a_0$).
    Consequently,
    \begin{multline*}
        \frac{1}{t}\sum_{j=0}^{t-1} f\Big(\cos(\frac{2\pi j}{t})\Big)
        = \frac{1}{t}\sum_{j=0}^{t-1} g\Big(\frac{2\pi j}{t}\Big)
        = \frac{1}{t}\sum_{j=0}^{t-1} \sum_{k=-\infty}^\infty \widehat{g}(k) e^{2\pi ijk/t}
        = \sum_{k=-\infty}^\infty \widehat{g}(k) \sum_{j=0}^{t-1}\frac{1}{t}e^{2\pi ijk/t} \\
        = \sum_{k=-\infty}^\infty \widehat{g}(k) \iver{k \text{ is divisible by } t}
        = \sum_{k=-\infty}^\infty \widehat{g}(kt),
    \end{multline*}
    so we can bound arithmetic progressions $\abs{\sum_k \widehat{g}(kt)} \leq \supnorm{f}$, and this generalizes to other offsets, to bound $\abs{\sum_k \widehat{g}(kt + o)}$ for some $o \in [t-1]$.
    Notably, though, this approach does not say anything about sums like $\sum_k a_{4k+1}$.
    The corresponding progression of Fourier coefficients doesn't give it, for example, since we pick up unwanted terms from the negative Fourier coefficients.\footnote{These sums are related to the Chebyshev coefficients one gets from interpolating a function at Chebyshev points \cite[Theorem~4.2]{trefethen19}.}
    \begin{align*}
        \sum_k \widehat{g}(4k + 1)
        &= \parens{\widehat{g}(1) + \widehat{g}(5) + \widehat{g}(9) + \cdots} + \parens{\widehat{g}(-3) + \widehat{g}(-7) + \widehat{g}(-11) + \cdots} \\
        &= \frac12\parens{a_1 + a_5 + a_9 + \cdots} + \frac12\parens{a_3 + a_7 + a_{11} + \cdots} = \sum_{k \geq 0} a_{2k+1}.
    \end{align*}
    In fact, by inspection of the distribution\footnote{This is the functional to integrate against to compute the sum, $\tfrac{2}{\pi}\int_{-1}^1 f(x)D(x)/\sqrt{1-x^2} = \sum a_{4k+1}$. The distribution is not a function, but can be thought of as the limit object of $D_n(x) = \sum_{k = 0}^n T_{4k+1}(x)$ as $n \to \infty$, analogous to Dirichlet kernels and the Dirac delta distribution.} $D(x) = \sum_{k =0}^\infty T_{4k+1}(x)$, it appears that this arithmetic progression cannot be written as a linear combination of evaluations of $f(x)$.
    Since the shape of the distribution appears to have $1/x$ behavior near $x = 0$, we conjecture that our analysis losing a log factor is, in some respect, necessary.
\end{remark}

\begin{conjecture}
    For any step size $t > 1$ and offset $o \in [t-1]$ such that $o \neq t/2$, there exists a function $f: [-1,1] \to \mathbb{R}$ such that $\supnorm{f} = 1$ but $\abs{\sum_{k=0}^n a_{tk + o}} = \Omega(\log(n))$.
\end{conjecture}


\section{Properties of the Clenshaw Recursion} \label{sec:clenshaw}

\subsection{Deriving the Clenshaw Recursions}

Suppose we are given as input a degree-$d$ polynomial as a linear combination of Chebyshev polynomials:
\begin{align}
    p(x) = \sum_{k=0}^d a_k T_k(x).
\end{align}
Then this can be computed with the \emph{Clenshaw algorithm}, which is the following recurrence.
\begin{align} \label{scalar-clenshaw-recursion}
    q_{d+1} &= q_{d+2} = 0 \nonumber \\
    q_k &= 2x q_{k+1} - q_{k+2} + a_k \tag{Clenshaw}\\
    \tilde{p}&= \tfrac{1}{2}(a_0 + q_0 - q_2)\nonumber
\end{align}
\begin{lemma} \label{clenshaw-correctness}
    The recursion in \cref{scalar-clenshaw-recursion} computes $p(x)$.
    That is, in exact arithmetic, $\tilde{p} = p(x)$.
    In particular,
    \begin{align} \label{iterate-closed-expression}
        q_k = \sum_{i=k}^d a_i U_{i-k}(x).
    \end{align}
\end{lemma}
\begin{proof}
We show \cref{iterate-closed-expression} by induction.
\begin{align*}
    q_k &= 2x q_{k+1} - q_{k+2} + a_k \\
    &= 2x\Big(\sum_{i=k+1}^d a_i U_{i-k-1}(x)\Big) - \Big(\sum_{i=k+2}^d a_i U_{i-k-2}(x)\Big) + a_k \\
    &= a_k + 2x a_{k+1}U_0(x) + \sum_{i=k+2}^d a_i (2x U_{i-k-1}(x) - U_{i-k-2}(x)) \\
    &= \sum_{i=k}^d a_i U_{i-k}(x).
\end{align*}
Consequently, we have
\begin{align*}
    \frac{1}{2}(a_0 + u_0 - u_2)
    &= \frac{1}{2}\Big(a_0 + \sum_{i=0}^d a_i U_{i}(x) - \sum_{i=2}^d a_i U_{i-2}(x)\Big) \\
    &= a_0 + a_1x + \sum_{i=2}^d \frac{a_i}{2}(U_{i}(x) - U_{i-2}(x))
    = \sum_{i=0}^d a_iT_i(x). \qedhere
\end{align*}
\end{proof}

\begin{remark}
Though the aforementioned discussion is specialized to the scalar setting, it extends to the the matrix setting almost entirely syntactically: consider a Hermitian $A \in \mathbb{C}^{n\times n}$ and $b \in \mathbb{C}^n$ with $\|A\|, \|b\| \leq 1$.
Then $p(A)b$ can be computed in the following way:
\begin{equation}
\label{eqn:matrix-clenshaw}
\begin{split}
    u_{d+1} &= \vec{0}; \\
    u_d &= a_d b; \\
    u_k &= 2A u_{k+1} - u_{k+2} + a_k b; \\
    u := p(A)b &= \frac{1}{2}(a_0b + u_0 - u_2).
\end{split}
\end{equation}
The proof that this truly computes $p(A)b$ is the same as the proof of correctness for Clenshaw's algorithm shown above.
We will also be generalizing to non-Hermitian $A \in \mathbb{C}^{m\times n}$, in which case the only additional wrinkle is that in the recurrence we will need to choose either $A$ or $A^\dagger$ such that dimensions are consistent.
Provided that the polynomial being computed is even or odd, no issues will arise.
Consequently, Clenshaw-like recurrences will give matrix polynomials where $x^k$ are replaced with $A^\dagger A A^\dagger \cdots A b$, which corresponds to the definition of singular value transformation from \cref{def:qsvt}.
\end{remark}

\subsection{Evaluating Even and Odd Polynomials}

We will be considering evaluating odd and even polynomials. We again focus on the scalar setting and note that this extends to the matrix setting in the obvious way. 
The previous recurrence \cref{scalar-clenshaw-recursion} can work in this setting, but it'll be helpful for our analysis if the recursion multiplies by $x^2$ each time, instead of $x$~\cite[Chapter 2, Problem 7]{mh02}.
So, in the case where the degree-$(2d+1)$ polynomial $p(x)$ is \emph{odd} (so $a_{2k} = 0$ for every $k$), it can be computed with the iteration
\begin{align} \label{odd-clenshaw}
    q_{d+1} &= q_{d+2} = 0; \nonumber\\
    q_k &= 2 T_2(x) q_{k+1} - q_{k+2} + a_{2k+1} U_1(x); \tag{Odd Clenshaw}\\
    \tilde{p} &= \tfrac12(q_0 - q_1). \nonumber
\end{align}
When $p(x)$ is a degree-$(2d)$ \emph{even} polynomial (so $a_{2k+1} = 0$ for every $k$), it can be computed via the same recurrence, replacing $a_{2k+1}U_1(x)$ with $a_{2k}$.
However, we will use an alternative form that's more convenient for us (since we can reuse the analysis of the odd case).
\begin{align} 
    \tilde{a}_{2k} &\coloneqq a_{2k} - a_{2k+2} + a_{2k+4} - \cdots \pm a_{2d}; \label{even-a-definition}\\
    q_{d+1} &= q_{d+2} = 0; \nonumber\\
    q_k &= 2T_2(x)q_{k+1} - q_{k+2} + \tilde{a}_{2k+2}U_1(x)^2; \tag{Even Clenshaw} \label{even-clenshaw}\\
    \tilde{p} &= \tilde{a}_0 + \tfrac12(q_0 - q_1). \nonumber
\end{align}
These recurrences correctly compute $p$ follows from a similar analysis to the standard Clenshaw algorithm, formalized below.
\begin{lemma} \label{parity-clenshaw-correctness}
    The recursions in \cref{odd-clenshaw} and \cref{even-clenshaw} correctly compute $p(x)$ for even and odd polynomials, respectively.
    That is, in exact arithmetic, $\tilde{p} = p(x)$.
    In particular,
    \begin{align} \label{parity-iterate-closed-expression}
        q_k = \sum_{i=k}^d a_i U_{i-k}(x).
    \end{align}
\end{lemma}

\begin{proof}
We can prove these statements by applying \cref{iterate-closed-expression}.
In the odd case, \cref{odd-clenshaw} is identical to \cref{scalar-clenshaw-recursion} except that $x$ is replaced by $T_2(x)$ and $a_k$ is replaced by $a_{2k+1}U_1(x)$, so by making the corresponding changes in the iterate, we get that
\begin{align}
    q_k &= \sum_{i=k}^d a_{2i+1}U_1(x)U_{i-k}(T_2(x))
    = \sum_{i=k}^d a_{2i+1} U_{2(i-k) + 1}(x) \tag*{by \cref{eq:u-composition}}\\
    \tilde{p} &= \tfrac12(q_0 - q_1) = \sum_{i=0}^d \frac{a_{2i+1}}{2}\Big(U_{2i+1}(x) - U_{2i-1}(x)\Big)
    = p(x). \tag*{by \cref{eq:t-to-u}}
\end{align}
Similarly, in the even case, \cref{even-clenshaw} is identical to \cref{scalar-clenshaw-recursion} except that $x$ is replaced by $T_2(x)$ and $a_k$ is replaced by $4\tilde{a}_{2k}x^2$ (see Definition \ref{even-a-definition}), so that
\begin{align}
    q_k &= \sum_{i=k}^d \tilde{a}_{2i+2} U_1(x)^2U_{i-k}(T_2(x)) \nonumber\\
    &= \sum_{i=k}^d \tilde{a}_{2i+2} U_1(x)U_{2(i-k)+1}(x) \tag*{by \cref{eq:u-composition}}\\
    &= \sum_{i=k}^d \tilde{a}_{2i+2} (U_{2(i-k)}(x) + U_{2(i-k+1)}(x)) \tag*{by \cref{cheb-recursive-definition}}\\
    &= \sum_{i=k}^{d+1} \tilde{a}_{2i+2} U_{2(i-k)}(x) + \sum_{i=k+1}^{d+1} \tilde{a}_{2i}U_{2(i-k)}(x) \tag*{noticing that $\tilde{a}_{2d+2} = 0$} \\
    &= \tilde{a}_{2k+2} + \sum_{i=k+1}^{d+1} (\tilde{a}_{2i} + \tilde{a}_{2i+2}) U_{2(i-k)}(x) \nonumber\\
    &= \tilde{a}_{2k+2} + \sum_{i=k+1}^{d} a_{2i} U_{2(i-k)}(x) \nonumber\\
    &= -\tilde{a}_{2k} + \sum_{i=k}^{d} a_{2i} U_{2(i-k)}(x). \nonumber
\end{align}
Finally, observe
\begin{equation*}
    \tilde{a}_{0} + \frac{1}{2}(q_0 - q_1) = \tilde{a}_{0} + \frac{1}{2}(a_0 - \tilde{a}_{0} + \tilde{a}_2) + \sum_{i=1}^d \frac{a_{2i}}{2}(U_{2i}(x) - U_{2i-2}(x)) = p(x). \qedhere
\end{equation*}
\end{proof}

\begin{remark} \label{rmk:scalar-stabilities}
    We can further compute what happens to all these recursions with some additive $\eps^{(k)}$ error in iteration $k$.
    This follows just by adding $\eps^{(k)}$ to the constant term in the recursion and chasing the resulting changes through the analysis of Clenshaw, as is done for \cref{parity-clenshaw-correctness}.
    \begin{align*}
        &\text{for standard Clenshaw:} &
        \tilde{q}_k &= \sum_{i = k}^d (a_i + \eps^{(i)}) U_{i-k}(x) \\
        &\text{for odd Clenshaw:} &
        \tilde{q}_k &= \sum_{i = k}^d (a_{2i+1}U_{2(i-k)+1}(x) + \eps^{(i)} U_{i-k}(T_2(x))) \\
        &\text{for even Clenshaw:} &
        \tilde{q}_k &= -\tilde{a}_{2k} + \sum_{i = k}^d (a_{2i}U_{2(i-k)}(x) + \eps^{(i)} U_{i-k}(T_2(x)))
    \intertext{Propagating this error to the full result gives the following results:}
        &\text{for standard Clenshaw:} &
        \tilde{p} - p(x) &= \frac12 + \sum_{i = 1}^d \eps^{(i)} T_{i}(x) \\
        &\text{for odd Clenshaw:} &
        \tilde{p} - p(x) &= \frac12 \eps^{(0)} + \frac12\sum_{i = 1}^d \eps^{(i)}(U_i(T_2(x)) - U_{i-1}(T_2(x))) \\
        &\text{for even Clenshaw:} &
        \tilde{p} - p(x) &= \frac12 \eps^{(0)} + \frac12\sum_{i = 1}^d \eps^{(i)}(U_i(T_2(x)) - U_{i-1}(T_2(x)))
    \end{align*}
    Because $\supnorm{T_i(x)} = 1$ but $\supnorm{U_i(T_2(x)) - U_{i-1}(T_2(x))} = 2i+1$, this suggests that these parity-specific recurrences are less stable than the standard recursion.
    However, they will be more amenable to our sketching techniques.
\end{remark}

\section{Stability of the Scalar Clenshaw Recursion}
\label{sec:basic-error-analysis-scalar}

Before we move to the matrix setting, we warmup with a stability analysis of the scalar Clenshaw recurrence.
Suppose we perform \cref{scalar-clenshaw-recursion} to compute a degree-$d$ polynomial $p$, except every addition, subtraction, and multiplication incurs $\eps$ relative error.
Typically, this has been analyzed in the finite precision setting, where the errors are caused by truncation.
These standard analyses show that this finite precision recursion gives $p(x)$ to $d^2(\sum \abs{a_i})\eps = \bigO{d^3\supnorm{p}\eps}$ error.
This bound $\sum \abs{a_i}$ is not easily improved in settings where $p$ is a polynomial approximation of a smooth function, since standard methods only give bounds of the form $\abs{a_k} = \Theta((1-\log(1/\eps)/d)^{-k})$, \cite[Theorem 8.1]{trefethen19}, giving only constant bounds for the coefficients.

Such a bound on $\sum \abs{a_k}$ is not tight, however.
A use of Parseval's formula~\cite[Theorem~5.3]{mh02} improves on this by a factor of $\sqrt{d}$: 
\begin{align}
    \sum \abs{a_i} \leq \sqrt{d}\sqrt{\sum a_i^2} = \bigO{\sqrt{d}\supnorm{p}}.
\end{align}
Testing $\supnorm{\sum_{\ell = 1}^d s_\ell \frac{1}{\sqrt{\ell}} T_\ell(x)}$ for random signs $s_\ell \in \{\pm 1\}$ suggests that this bound is tight, meaning coefficient-wise bounds can only prove an error overhead of $\bigThetat{d^{2.5}\supnorm{p}}$ for the Clenshaw recurrence.

We improve on prior stability analyses to show that the Clenshaw recurrence for Chebyshev polynomials only incurs an error overhead of $d^2\log(d)\supnorm{p}$.
This is tight up to a logarithmic factor.
This, for example, could be used to improve the bound in \cite[Lemma~9]{mms18} from $k^3$ to $k^2\log(k)$ (where in that paper, $k$ denotes degree).
As we do for the upcoming matrix setting, we proceed by performing an error analysis on the recursion with a stability parameter $\mu$, and then showing that for any bounded polynomial, $1/\mu$ can be chosen to be $\bigO{d^2 \log d}$.

\subsection{Analyzing Error Propagation}

The following is a simple analysis of Clenshaw, with some rough resemblance to an analysis of Oliver \cite{Oliver79}.

\begin{theorem}[Stability analysis for scalar Clenshaw] \label{scalar-clenshaw-stability}
    Consider a degree-$d$ polynomial $p: [-1,1] \to \mathbb{R}$ with Chebyshev coefficients $p(x) = \sum_{k=0}^d a_k T_k(x)$.
    Let $\oplus, \ominus, \odot: \mathbb{C} \times \mathbb{C} \to \mathbb{C}$ be binary operations representing addition, subtraction, and multiplication to $\mu\eps$ relative error, for $0<\eps<1$:
    \begin{align*}
        \abs{(x \oplus y) - (x + y)} &\leq \mu\eps(\abs{x} + \abs{y}); \\
        \abs{(x \ominus y) - (x - y)} &\leq \mu\eps(\abs{x} + \abs{y}); \\
        \abs{x \odot y - x\cdot y} &\leq \mu\eps\abs{x}\abs{y} = \mu\eps\abs{xy}.
    \end{align*}
    Given an $x \in [-1,1]$, consider performing the Clenshaw recursion with these noisy operations:
    \begin{align} \label{finite-precision-clenshaw}
        \tilde{q}_{d+1} &= \tilde{q}_{d+2} = 0; \nonumber \\
        \tilde{q}_k &= (2 \odot x) \odot \tilde{q}_{k+1} \ominus (\tilde{q}_{k+2} \ominus a_k); \tag{Finite-Precision Clenshaw}\\
        \tilde{\tilde{p}}&= \tfrac{1}{2} \odot ((a_0 \oplus q_0) \ominus q_2).\nonumber
    \end{align}
    Then \cref{finite-precision-clenshaw} outputs $p(x)$ up to $50\eps\supnorm{p}$ error\footnote{
        We did not attempt to optimize the constants for this analysis.
    }, provided that $\mu > 0$ satisfies the following three criterion.
    \begin{enumerate}[label=(\alph*)]
        \item $\mu\eps \leq \frac{1}{50 (d+2)^2}$;
        \item $\mu \sum_{i=0}^ d \abs{a_i} \leq \supnorm{p}$;
        \item $\mu\abs{q_k} = \mu\abs*{\sum_{i=k}^d a_iU_{i-k}(x)} \leq \frac1d\supnorm{p}$ for all $k \in \{0,\ldots,d\}$.
    \end{enumerate}
\end{theorem}

This analysis shows that arithmetic operations incurring $\mu \eps$ error result in computing $p(x)$ to $\eps$ error.
In particular, the stability of the scalar Clenshaw recurrence comes down to understanding how small we can take $\mu$.
Note that if we ignored coefficient sign, $\abs*{\sum_{i=k}^d a_i U_{i-k}(x)} \leq \abs*{\sum_{i=k}^d \abs{a_i} U_{i-k}(x)} = \sum_{i=k}^d (i-k+1)\abs{a_i}$, this would require setting $\mu = \Theta(1/d^3)$.
We show in \cref{subsec:clenshaw-iterates} that we can set $\mu = \Theta((d^2\log(d))^{-1})$ for all $x \in [-1,1]$ and polynomials $p$.
\begin{lemma} \label{scalar-mu-choice}
    In \cref{scalar-clenshaw-stability}, it suffices to take $\mu = \Theta((d^2\log(d))^{-1})$.
\end{lemma}

\begin{proof}[Proof of \cref{scalar-clenshaw-stability}]
We will expand out these finite precision arithmetic to get error intervals for each iteration.
\begin{equation}
\tilde{q}_{d+1} = \tilde{q}_{d+2} = 0,
\end{equation}
and 
\begin{equation*}
\begin{split}
\tilde{q}_k &= (2 \odot x) \odot \tilde{q}_{k+1} \ominus (\tilde{q}_{k+2} \ominus a_k) \\
    &= (2x \pm 2\mu\eps\abs{x}) \odot \tilde{q}_{k+1} \ominus (\tilde{q}_{k+2} - a_k \pm \mu\eps(\abs{\tilde{q}_{k+2}} + \abs{a_k})) \\
    &= ((2x\tilde{q}_{k+1} \pm 2\mu\eps \abs{x}\tilde{q}_{k+1}) \pm \mu\eps\abs{(2x \pm 2\mu\eps\abs{x})\tilde{q}_{k+1}}) \ominus (\tilde{q}_{k+2} - a_k \pm \mu\eps(\abs{\tilde{q}_{k+2}} + \abs{a_k})) \\
    &\in (2x\tilde{q}_{k+1} \pm (2\mu\eps + \mu^2\eps^2)2\abs{x\tilde{q}_{k+1}}) \ominus (\tilde{q}_{k+2} - a_k \pm \mu\eps(\abs{\tilde{q}_{k+2}} + \abs{a_k})) \\
    &\in (2x\tilde{q}_{k+1} \pm 6\mu\eps\abs{x\tilde{q}_{k+1}}) \ominus (\tilde{q}_{k+2} - a_k \pm \mu\eps(\abs{\tilde{q}_{k+2}} + \abs{a_k})) \\
    &= 2x\tilde{q}_{k+1} - \tilde{q}_{k+2} + a_k \pm \mu\eps(6\abs{x\tilde{q}_{k+1}} + \abs{\tilde{q}_{k+2}} + \abs{a_k}) \\
    &\qquad + \mu\eps\abs{2x\tilde{q}_{k+1} \pm 6\mu\eps\abs{x\tilde{q}_{k+1}}} + \mu\eps\abs{\tilde{q}_{k+2} - a_k \pm \mu\eps(\abs{\tilde{q}_{k+2}} + \abs{a_k})} \\
    &\in 2x\tilde{q}_{k+1} - \tilde{q}_{k+2} + a_k \pm \mu\eps(14\abs{x\tilde{q}_{k+1}} + 3\abs{\tilde{q}_{k+2}} + 3\abs{a_k}),
\end{split}
\end{equation*}
and,
\begin{equation*}
\begin{split}
\tilde{\tilde{p}}&= \tfrac{1}{2} \odot ((a_0 \oplus q_0) \ominus q_2) \\
    &= \tfrac{1}{2} \odot ((a_0 + q_0 \pm \mu\eps(\abs{a_0} + \abs{q_0})) \ominus q_2) \\
    &= \tfrac{1}{2} \odot ((a_0 + q_0 - q_2 \pm \mu\eps(\abs{a_0} + \abs{q_0})) \pm \mu\eps(\abs{a_0 + q_0 \pm \mu\eps(\abs{a_0} + \abs{q_0})} + \abs{q_2})) \\
    &\in \tfrac{1}{2} \odot (a_0 + q_0 - q_2 \pm \mu\eps(3\abs{a_0} + 3\abs{q_0} + \abs{q_2})) \\
    &= \tfrac12(a_0 + q_0 - q_2 \pm \mu\eps(3\abs{a_0} + 3\abs{q_0} + \abs{q_2})) \pm \mu\eps\tfrac12\abs{a_0 + q_0 - q_2 \pm \mu\eps(3\abs{a_0} + 3\abs{q_0} + \abs{q_2})} \\
    &\in \tfrac12(a_0 + q_0 - q_2) \pm \tfrac12\mu\eps(7\abs{a_0} + 7\abs{q_0} + 3\abs{q_2}).
\end{split}
\end{equation*}
To summarize, we have
\begin{align}
    \tilde{q}_{d+1} &= \tilde{q}_{d+2} = 0, \nonumber \\
    \tilde{q}_k &= 2x\tilde{q}_{k+1} - \tilde{q}_{k+2} + a_k + \delta_k,
    &\text{where } \abs{\delta_k} &\leq \mu\eps(14\abs{x\tilde{q}_{k+1}} + 3\abs{\tilde{q}_{k+2}} + 3\abs{a_k}) \label{eqn:additive-error-in-qk}\\
    \tilde{\tilde{p}}&= \tfrac12(a_0 + q_0 - q_2) + \delta,
    &\text{where } \abs{\delta} &\leq \tfrac12\mu\eps(7\abs{a_0} + 7\abs{q_0} + 3\abs{q_2})
\end{align}
By \cref{clenshaw-correctness}, this recurrence satisfies
\begin{align} 
    \tilde{q}_k &= \sum_{i=k}^d U_{i-k}(x)(a_i + \delta_i) \nonumber\\
    q_k - \tilde{q}_k &= \sum_{i=k}^d U_{i-k}(x) \delta_i \nonumber\\
    q - \tilde{q} &= \delta + \frac{1}{2}\Big(\sum_{i=0}^d U_{i}(x) \delta_i - \sum_{i=2}^d U_{i-2}(x) \delta_i\Big)\nonumber\\
    &= \delta + \frac{1}{2}\delta_0 + \sum_{i=1}^d T_{i}(x) \delta_i \nonumber\\
    \abs{q - \tilde{q}} &\leq \abs{\delta} + \frac{1}{2}\abs{\delta_0} + \sum_{i=1}^d \abs{T_i(x)\delta_i}
    \leq \abs{\delta} + \sum_{i=0}^d \abs{\delta_i}.
\intertext{This analysis so far has been fully standard. Let's continue bounding.}
    &\leq \mu\eps\parens[\Big]{\tfrac72\abs{a_0} + \tfrac72\abs{q_0} + \tfrac32 \abs{q_2} + \sum_{i=0}^d(14\abs{x\tilde{q}_{i+1}} + 3\abs{\tilde{q}_{i+2}} + 3\abs{a_i})} \nonumber\\
    &\leq \mu\eps\sum_{i=0}^d(20\abs{\tilde{q}_i} + 10\abs{a_i}).
    \label{eqn:clenshaw-bound-0}
\end{align}
Now, we will bound all of the $\delta_k$'s.
Combining previous facts, we have
\begin{align*}
    \abs{\tilde{q}_k}
    &= \abs{\sum_{i=k}^d U_{i-k}(x)(a_i + \delta_i)} \\
    &\leq \abs[\Big]{\sum_{i=k}^d U_{i-k}(x) a_i} + \sum_{i=k}^d \abs[\Big]{U_{i-k}(x)\delta_i} \\
    &\leq \abs[\Big]{\sum_{i=k}^d U_{i-k}(x) a_i} + \sum_{i=k}^d (i-k+1)\abs{\delta_i} \\
    &\leq \abs[\Big]{\sum_{i=k}^d U_{i-k}(x) a_i} + \mu\eps\sum_{i=k}^d (i-k+1)(14\abs{\tilde{q}_{i+1}} + 3\abs{\tilde{q}_{i+2}} + 3\abs{a_i}) \\
    &\leq \parens[\Big]{\frac{1}{\mu d} + 3\mu \eps\frac{d-k+1}{\mu}}\supnorm{p} + \mu\eps\sum_{i=k}^d (i-k+1)(14\abs{\tilde{q}_{i+1}} + 3\abs{\tilde{q}_{i+2}}) \\
    &\leq \frac{1.5}{\mu d}\supnorm{p} + \mu\eps\sum_{i=k}^d (i-k+1)(14\abs{\tilde{q}_{i+1}} + 3\abs{\tilde{q}_{i+2}}).
\end{align*}
Note that $\abs{\tilde{q}_k} \leq c_k$, where
\begin{align}
    c_d &= 0; \nonumber\\
    c_k &= \frac{1.5}{\mu d}\supnorm{p} + \mu\eps\sum_{i=k}^d (i-k+1)(14c_{i+1} + 3c_{i+2}). \nonumber
    \intertext{Solving this recurrence, we have that $c_k \leq \frac{2}{\mu d}\supnorm{p}$, since by strong induction,}
    c_k &\leq \Big(\frac{1.5}{\mu d} + \mu\eps\sum_{i=k}^d (i-k+1)17\frac{2}{\mu d}\Big)\supnorm{p} \nonumber\\
    &= \Big(\frac{1.5}{\mu d} + 17\mu\eps\frac{1}{\mu d}(d-k+1)(d-k+2)\Big)\supnorm{p} 
    \leq \frac{2}{\mu d}\supnorm{p}. \nonumber
\end{align}
Returning to \Cref{eqn:clenshaw-bound-0}:
\begin{align}
    \abs{q - \tilde{q}}
    &\leq \mu\eps\sum_{i=0}^d(20\abs{\tilde{q}_i} + 10\abs{a_i})
    \leq \mu\eps\sum_{i=0}^d(20c_i + 10\abs{a_i}) \nonumber \\
    &\leq 40\eps\supnorm{p} + 10\mu\eps\sum_{i=0}^d\abs{a_i}
    \leq 50\eps\supnorm{p}. \qedhere
\end{align}
\end{proof}

\subsection{Bounding the Iterates of the Clenshaw Recurrence}
\label{subsec:clenshaw-iterates}

The goal of this section is to prove \cref{scalar-mu-choice}.
In particular, we wish to show that for $\mu = \Theta((d^2\log(d))^{-1})$, the following criteria hold:
\begin{enumerate}[label=(\alph*)]
    \item $\mu\eps \leq \frac{1}{50 (d+2)^2}$;
    \item $\mu \sum_{i=0}^ d \abs{a_i} \leq \supnorm{p}$;
    \item $\mu\abs{q_k} = \mu\abs*{\sum_{i=k}^d a_iU_{i-k}(x)} \leq \frac1d\supnorm{p}$ for all $k \in \{0,\ldots,d\}$.
\end{enumerate}
For this choice of $\mu$, (a) is clearly satisfied, and since $\abs{a_i} \leq 2\supnorm{p}$ (\cref{lem:coefficient-bound}), $\mu \sum_{i=0}^d \abs{a_i} \leq 2(d+1)\supnorm{p} \leq \supnorm{p}$, so (b) is satisfied.
In fact, both of these criterion are satisfied for $\mu = \bigOmega{1/d}$, provided $\eps$ is sufficiently small.

Showing (c) requires bounding $\supnorm{\sum_{\ell=k}^d a_{\ell} U_{\ell - k}(x)}$ for all $k \in [d]$.
These expressions are also the iterates of the Clenshaw algorithm (\cref{clenshaw-correctness}), so we are in fact trying to show that in the process of our algorithm we never produce a value that's much larger than the final value.
From testing computationally, we believe that the following holds true.

\begin{conjecture} \label{conj:iterates-bound}
Let $p(x)$ be a degree-$d$ polynomial with Chebyshev expansion $p(x) = \sum_{\ell=0}^d a_\ell T_\ell(x)$.
Then, for all $k$ from $0$ to $d$,
\begin{align*}
    \supnorm[\Big]{\sum_{\ell=k}^d a_{\ell} U_{\ell - k}(x)} \leq (d-k+1)\supnorm{p},
\end{align*}
maximized for the Chebyshev polynomial $p(x) = T_d(x)$.
\end{conjecture}
\cref{conj:iterates-bound} would imply that it suffices to take $\mu = \Theta(1/d^2)$.
We prove it up to a log factor.
\begin{theorem} \label{thm:scalar-clenshaw-iterate-bound}
    For a degree-$d$ polynomial $p(x) = \sum_{\ell=0}^d a_\ell T_\ell(x)$, consider the degree-$(d-k)$ polynomial $q_k(x) = \sum_{\ell=k}^d a_{\ell} U_{\ell - k}(x)$.
    Then
    \begin{align*}
        \supnorm{q_k} \leq (d-k+1)\Big(16 + \frac{16}{\pi^2}\log(d)\Big)\supnorm{p}.
    \end{align*}
\end{theorem}
\begin{proof}
We proceed by carefully bounding the Chebyshev coefficients of $q_k$, which turn out to be arithmetic progressions of the $a_k$'s which we bounded in \cref{sec:chebsums}.
\begin{align*}
    q_k(x) &= \sum_i a_iU_{i-k}(x) \\
    &= \sum_i \sum_{j \geq 0}a_i T_{i - k - 2j}(x)(1 + \iver{i - k - 2j \neq 0}) \\
    &= \sum_i \sum_{j \geq 0}a_{i + k + 2j} T_{i}(x)(1 + \iver{i \neq 0}) \\
    &= \sum_i T_{i}(x)(1 + \iver{i \neq 0}) \sum_{j \geq 0}a_{i + k + 2j} \\
    \abs{q_k(x)} &\leq \sum_i \iver{i \geq 0}(1 + \iver{i \neq 0}) \abs*[\Big]{\sum_{j \geq 0}a_{i + k + 2j}} \\
    &\leq 2\sum_{i=0}^{d-k} \abs*[\Big]{\sum_{j \geq 0}a_{i + k + 2j}} \\
    &\leq 4\sum_{i=0}^{d-k} \Big(4 + \frac{4}{\pi^2}\log(i + k - 1)\Big)\supnorm{p} \tag*{by \cref{parity-two-tailsums}}\\
    &\leq (d-k+1)\Big(16 + \frac{16}{\pi^2}\log(d)\Big)\supnorm{p}. \qedhere
\end{align*}
\end{proof}

\begin{remark}
    We spent some time trying to prove \cref{conj:iterates-bound}, since its form is tantalizingly close to that of the Markov brothers' inequality~\cite{schaeffer41} $\supnorm{\frac{d}{dx}p(x)} = \supnorm{\sum_{\ell=0}^d a_\ell \ell U_{\ell - 1}(x)} \leq d^2\supnorm{p(x)}$, except with the linear differential operator $\frac{d}{dx}: T_\ell \mapsto \ell U_{\ell - 1}$ replaced with the linear operator $T_\ell \mapsto U_{\ell - k}$.
    However, calculations suggest that the variational characterization of $\max_{\supnorm{p} = 1}\abs*{\frac{d}{dx}p(x)}$ underlying proofs of the Markov brothers' inequality~\cite{shadrin04} does not hold here, and from our shallow understanding of these proofs, it seems that they strongly use properties of the derivative.
\end{remark}


\section{Computing Matrix Polynomials}

\ewin{Note to self: you can't improve $\supnorm{p}$ to $\norm{p}_{\spec(A)}$ without incurring a condition number dependence: $TAS$ approximates $A$'s singular values, but only $\sigma_i^2$ to $\eps$ error.}

Our goal is to prove the following theorem:

\begin{theorem} \label{main-svt-alg}
    Suppose we are given sampling and query access to $A \in \mathbb{C}^{m\times n}$ and $b \in \mathbb{C}^n$ with $\norm{A} \leq 1$; an even or odd degree-$d$ polynomial $p$ with $p(0) = 0$, given as its Chebyshev coefficients;
    and a sufficiently small accuracy parameter $\eps > 0$.
    Then we can output a description of a vector $y$ (in $\C^m$ if $p$ is odd, in $\C^n$ if $p$ is even) such that $\|y - p(A)b\| \leq \eps\supnorm{p}\norm{b}$ with probability $\geq 0.9$ in time
    \begin{align*}
        \bigO*{\min\braces{\nnz(A), \frac{\|A\|_F^4}{\eps^4}d^{12}\log^8(d)\log^2\frac{\fnorm{A}}{\norm{A}}} + \frac{\|A\|_F^4}{\eps^2}d^{11}\log^4(d)\log\frac{\fnorm{A}}{\norm{A}}}.
    \end{align*}
    We can access the output description in the following way:
    \begin{enumerate}[label=(\roman*)]
        \item Compute entries of $y$ in $\bigO[\Big]{ \frac{\fnorm{A}^2}{\eps^2}d^6\log^4(d)\log\frac{\fnorm{A}}{\norm{A}} }$ time;
        \item Sample $i \in [n]$ with probability $\frac{\abs{y_i}^2}{\norm{y}^2}$ in $\bigO[\Big]{ \frac{\supnorm{p}^2\fnorm{A}^4\norm{b}^2}{\eps^2\norm{y}^2}d^8\log^8(d) \log\frac{\fnorm{A}}{\norm{A}}}$ time with probability $\geq 0.9$;
        \item Estimate $\norm{y}^2$ to $\nu$ relative error in $\bigO[\Big]{ \frac{\supnorm{p}^2\fnorm{A}^4\norm{b}^2}{\nu^2\eps^2\norm{y}^2}d^8\log^8(d)\log\frac{\fnorm{A}}{\norm{A}}}$ time with probability $\geq 0.9$.
    \end{enumerate}
\end{theorem}

We prove the even and odd cases separately.
We can conclude the above theorem as a consequence of \cref{odd-svt-alg,cor:odd-iterate-mu-bound,even-svt-alg,cor:even-iterate-mu-bound}.

\begin{remark} \label{rmk:overlaps-description}
We can also compute estimate $\angles{u|y}$ to $\eps\norm{u}\norm{b}$ error without worsening the $1/\eps^2$ dependence.
This does not follow from the above; rather, we can observe that, from the description of our output, $y = (AS)v + \eta b$, it suffices to estimate $u^\dagger (AS)v$ and $u^\dagger b$.
Because of properties of the sketches and the later analysis, $\fnorm{AS} \lesssim \fnorm{A}$ and $\norm{v} \lesssim d\log^2(d)\norm{b}$, so by \cite[Lemma~4.12 and Remark~4.13]{cglltw19} we can estimate these to the desired error with $\bigO{d^2\log^4(d)\fnorm{A}^2\frac{1}{\eps^2}\log\frac{1}{\delta}}$ queries to $u, v, A$ and samples to $AS$.
All of these queries can be done in $\bigO{1}$ time.
\end{remark}

\subsection{Computing Odd Matrix Polynomials}
In this section, we prove the following theorem.
The statement involves a parameter $\mu$ which depends on the polynomial being evaluated; this parameter is between $1$ and $(d\log d)^{-2}$, depending on how well-conditioned the polynomial is.

\begin{theorem} \label{odd-svt-alg}
    Suppose we are given sampling and query access to $A \in \mathbb{C}^{m\times n}$ and $b \in \mathbb{C}^n$ with $\norm{A} \leq 1$; a $(2d+1)$-degree odd polynomial $p$, written in its Chebyshev coefficients as
    \begin{align*}
        p(x) = \sum_{i=0}^{d} a_{2i+1} T_{2i+1}(x);
    \end{align*}
    an accuracy parameter $\eps > 0$; a failure probability parameter $\delta > 0$; and a stability parameter $\mu > 0$.
    Then we can output a vector $x \in \C^n$ such that $\|A x - p(A)b\| \leq \eps\supnorm{p}\norm{b}$ with probability $\geq 1-\delta$ in time
    \begin{align*}
        \bigO*{\min\braces{\nnz(A), \frac{d^4\|A\|_F^4}{(\mu\eps)^4}\log^2\parens[\big]{\frac{\fnorm{A}}{\delta\norm{A}}}} + \frac{d^7\|A\|_F^4}{(\mu\eps)^2\delta}\log\parens[\big]{\frac{\fnorm{A}}{\delta\norm{A}}}},
    \end{align*}
    assuming $\mu\eps < \min(\frac{1}{4}d\norm{A}, \frac{1}{100d})$ and $\mu$ satisfies the following bounds:
    \begin{enumerate}[label=(\alph*),ref=\thetheorem (\alph*)]
        \item $\mu\sum_{i=0}^d \abs{a_{2i+1}} \leq \supnorm{p}$; \label{odd-svtcond-ai}
        \item $\mu\supnorm{\sum_{i=k}^d a_{2i+1}U_{i-k}(T_2(x))} \leq \frac{1}{d}\supnorm{p}$ for all $0 \leq k \leq d$; \label{odd-svtcond-trunc}
    \end{enumerate}
    The output description has the additional properties
    \begin{gather*}
        \sum_j \norm{A_{*,j}}^2\abs{x_j}^2 \lesssim \frac{\eps^2\supnorm{p}^2\norm{b}^2}{d^4\log\frac{\fnorm{A}}{\delta\norm{A}}}
        \qquad \znorm{x} \lesssim \frac{d^2\fnorm{A}^2}{(\mu\eps)^2}\log\frac{\fnorm{A}}{\delta\norm{A}},
    \end{gather*}
    so that by \cref{Ax-sampling}, for the output vector $y \coloneqq Ax$, we can:
    \begin{enumerate}[label=(\roman*)]
        \item Compute entries of $y$ in $\bigO{\znorm{x}} = \bigO[\Big]{ \frac{d^2\fnorm{A}^2}{(\mu\eps)^2}\log\tfrac{\fnorm{A}}{\delta\norm{A}} }$ time;
        \item Sample $i \in [n]$ with probability $\frac{\abs{y_i}^2}{\norm{y}^2}$ in $\bigO[\Big]{ \frac{\supnorm{p}^2\fnorm{A}^4\norm{b}^2}{\mu^4\eps^2\norm{y}^2} \log\frac{\fnorm{A}}{\delta\norm{A}}\log\frac{1}{\delta}}$ time with probability $\geq 1-\delta$;
        \item Estimate $\norm{y}^2$ to $\nu$ relative error in $\bigO[\Big]{ \frac{\supnorm{p}^2\fnorm{A}^4\norm{b}^2}{\nu^2\mu^4\eps^2\norm{y}^2}\log\frac{\fnorm{A}}{\delta\norm{A}} \log\frac{1}{\delta}}$ time with probability $\geq 1-\delta$.
    \end{enumerate}
\end{theorem}

\begin{mdframed}
  \begin{algorithm}[Odd singular value transformation]
    \label{algo:odd-dquantizing}\mbox{}
    \begin{description}
    \item[Input (pre-processing):] A matrix $A \in \mathbb{C}^{m \times n}$, vector $b \in \mathbb{C}^{n}$, and parameters $\eps,\,\delta,\,\mu>0$.
    \item[Pre-processing sketches:]
    Let $s, t = \Theta\Paren{\frac{d^2\fnorm{A}^2}{(\mu\eps)^2}\log(\frac{\fnorm{A}}{\delta\norm{A}})}$.
    This phase will succeed with probability $\geq 1-\delta$.
    \begin{enumerate}[label=P\arabic*., ref=A\thealgorithm.P\arabic*]
        \item If $\sq(A^\dagger)$ and $\sq(b)$ are not given, compute data structures to simulate them in $\bigO{1}$ time; \label{alg:sq-ab}
        \item Sample $S \in \C^{n \times s}$ to be $\aamp_s\parens[\Big]{ A, b, \{\tfrac12(\frac{\|A_{*,j}\|^2}{\fnorm{A}^2} + \frac{\abs{b_j}^2}{\norm{b}^2})\}_{j \in [n]}}$ (\cref{def:aamp}); \label{alg:get-s}
        \item Sample $T^\dagger \in \C^{m \times t}$ to be $\aamp_t\parens[\Big]{S^\dagger A^\dagger, AS, \{\frac{\norm{[AS]_{i,*}}}{\fnorm{AS}^2}\}_{i \in [m]}}$ (\cref{def:aamp}); \label{alg:get-t}
        \item Compute a data structure that can respond to $\sq(TAS)$ queries in $\bigO{1}$ time; \label{alg:sq-tas}
    \end{enumerate}
    \item[Input:] A degree $2d+1$ polynomial $p(x) = \sum_{i = 0}^{d} a_{2i+1} T_{2i+1}(x)$ given as its coefficients $a_{2i+1}$.
    \item[Clenshaw iteration:]
    Let $r = \Theta(d^4\fnorm{A}^2(s+t)\frac{1}{\delta}) = \Theta(d^6\frac{\fnorm{A}^4}{(\mu\eps)^2\delta}\log(\frac{\fnorm{A}}{\delta\norm{A}}))$.
    This phase will succeed with probability $\geq 1-\delta$.
    Starting with $v_{d+1} = v_{d+2} = \vec{0}^s$ and going until $v_0$,
    \begin{enumerate}[label=I\arabic*., ref=A\thealgorithm.I\arabic*]
        \item Let $B^{(k)} = \ssketch_r(TAS)$ and $B_\dagger^{(k)} = \ssketch_r((TAS)^\dagger)$ (\cref{def:ssketch}); \label{alg:bests}
        \item Compute $v_k = 2(2B_\dagger^{(k)} B^{(k)} - I)v_{k+1} - v_{k+2} + a_{2k+1} S^\dagger b$. \label{alg:itr}
    \end{enumerate}
    \item[Output:] 
    Output $x = \tfrac12 S(v_0 - v_1)$ satisfying $\Norm{ A x - p(A) b } \leq \eps\supnorm{p}\norm{b}$.
    \end{description}
  \end{algorithm}
\end{mdframed}

The criterion for what $\mu$ need to be are somewhat non-trivial; the important requirement is \cref{odd-svtcond-trunc}, which states that $1/\mu$ is a bound on the norm of various polynomials.
These polynomials turn out to be iterates of \cref{odd-clenshaw} when computing $p(x)/x$, so roughly speaking, the algorithm we present depends on the numerical stability of evaluating $p(x) / x$.
This is necessary because we primarily work in the ``dual'' space, maintaining our Clenshaw iterate $u_k$ as $A v_{k}$ where $v_k$ is a sparse vector.
This bears a resemblance to the dependence in \cite[Theorem 5.1]{cglltw19} on the Lipschitz smoothness of $f(x) / x$.
For any bounded polynomial, we can always take $\mu$ to be $\bigOmega{(d\log(d))^{-2}}$.

\begin{corollary}[Corollary of \cref{odd-iterate-mu-bound}] \label{cor:odd-iterate-mu-bound}
    In \cref{odd-svt-alg}, we can always take $1/\mu \lesssim d^2\log^2(d)$ for $d > 1$.
\end{corollary}

This bound is achieved up to log factors by $p(x) = T_{2k+1}(x)$.
We are now ready to dive into the proof of~\cref{odd-svt-alg}.
Without loss of generality, we assume $\supnorm{p} = 1$.

\paragraph{Pre-processing sketches, running time.}
Given a matrix $A \in C^{m\times n}$ and a vector $b \in \C^n$, the pre-processing phase of \cref{algo:odd-dquantizing} can be performed in $\bigO{\nnz(A) + \nnz(b)}$ time.
First, we build a data structure to respond to $\sq(A^\dagger)$ and $\sq(b)$ queries in $\bigO{1}$ time using the alias data structure described in \cref{rmk:when-sq-access} (\ref{alg:sq-ab}).
Then, we use these accesses to construct an $\aamp$ sketch $S$ for $Ab$, where we produce the samples for the sketch by sampling from $b$ and sampling from the row norms of $A^\dagger$, each with probability $\half$ (\ref{alg:get-s}).
Since we need $s$ samples, this takes $\bigO{s}$ queries to the data structure.
The sketch $S$ is defined such that $AS$ is a subset of the columns of $A$, with each column rescaled according to the probability it was sampled.
So, with another pass through $A$, we can construct a data structure for $\sq(AS)$ in $\bigO{1}$ time using \cref{rmk:when-sq-access} and use this to construct an $\aamp$ sketch $T^\dagger$ for $(AS)^\dagger AS$ (\ref{alg:get-t}).
The samples for the sketch are drawn from the row norms of $AS$.
The final matrix $TAS$ is a rescaled submatrix of $A$; another $\bigO{\nnz(A)}$ pass through $A$ suffices to construct a data structure to respond to $\sq(TAS)$ queries in $\bigO{1}$ time (\ref{alg:sq-tas}).
The running time is $\bigO{\nnz(A) + \nnz(b) + s + t}$ or, alternatively, three passes through $A$ and one pass with $b$, using $\bigO{st}$ space.
We can assume that $s, t \leq \nnz(A)$, though: if $s \geq n$, then we can take $S = I$, and if $t \geq m$, then we can take $T = I$, and this will satisfy the same guarantees; further, without loss, $\nnz(A) \geq \max(m, n)$, since otherwise there is an empty row or column that we can ignore.

If we are given $A, b$ such that $\sq(A^\dagger)$ and $\sq(b)$ queries can be performed in $\bigO{Q}$ time, then the pre-processing phase of \cref{algo:odd-dquantizing} can be performed in $\bigO{Qst}$ time.
The main difference from the description above is that we use that, given $\sq(A^\dagger)$, we can simulate queries to $\sq(AS)$ with only $\bigO{s}$ overhead.
To produce a sample $i \in [m]$ with probability $\norm{[AS]_{i,*}}^2/\fnorm{AS}^2$, sample a column index $j \in [s]$ with probability $\norm{[AS]_{*,j}}^2/\fnorm{AS}^2$ and then query $\sq(A)$ to get a row index $i \in [m]$ with probability $\abs{[AS]_{i,j}}^2/\norm{[AS]_{*,j}}^2$.

\begin{remark}
    If we are not given $b$ until after the pre-processing phase, or if we are only given $b$ as a list of entries without $\sq(b)$, then we can take the $S$ sketch to just be sampling from the row norms of $A^\dagger$.
    This will decrease the success probability of the following phase (specifically, because of the guarantee in \cref{eq:ab-amp}) to $0.99$.
\end{remark}

\paragraph{Pre-processing sketches, correctness.}
We list the guarantees of the sketch that we will use in the error analysis, and point to where they come from in \cref{subsec:assym-amp}.
Recall that, with $S \in \C^{n \times s}$ taken to be an $\aamp$ sketch of $X \in \C^{m\times n}, Y \in \C^{n \times d}$, by \cref{cor:amm-usable}, with probability $\geq 1-\delta$, $\norm*{XSS^\dagger Y^\dagger - XY^\dagger} \leq \eps\norm{X}\norm{Y}$ provided $s = \Omega(\frac{1}{\eps^2}(\frac{\fnorm{X}^2}{\norm{X}^2} + \frac{\fnorm{Y}^2}{\norm{Y}^2})\log(\frac{1}{\delta}(\frac{\fnorm{X}^2}{\norm{X}^2} + \frac{\fnorm{Y}^2}{\norm{Y}^2})))$ and $\eps \leq 1$.
We will use this with $s, t = \Theta\Paren*{\frac{d^2\fnorm{A}^2}{(\mu\eps)^2}\log(\frac{\fnorm{A}}{\delta\norm{A}})}$, where $\frac{\mu\eps}{d} \leq \frac14\norm{A}$.
The guarantees of \cref{cor:amm-usable} individually fail with probability $\bigO{\delta}$, so we will rescale to say that they all hold with probability $\geq 1-\delta$.
The following bounds hold for the sketch $S$.
\begin{align}
    \norm*{[AS]_{*,j}}^2 &\leq 2\fnorm{A}^2 / s \text{ for all } j \in [s]
    &\text{by \cref{def:aamp}} \tag{$\|[AS]_{*,j}\|$ bd} \label{eq:as-bound}\\
    \fnorm{AS}^2 &\leq 2\fnorm{A}^2
    &\text{by \cref{eq:as-bound}} \tag{$\|AS\|_F$ bd} \label{eq:fas-bound}\\
    \norm*{S^\dagger b}^2 &\leq 2\norm{b}^2
    &\text{by \cref{def:aamp}} \tag{$\|S^\dagger b\|$ bd} \label{eq:sb-bound}\\
    \norm*{A b - AS S^\dagger b} &\leq \tfrac{\mu\eps}{d}\norm{b}
    &\text{by \cref{cor:amm-usable}} \tag{$Ab$ AMP} \label{eq:ab-amp}\\
    \norm*{A A^\dagger - AS (AS)^\dagger} &\leq \tfrac{\mu\eps}{d}\norm{A}
    &\text{by \cref{cor:amm-usable}} \tag{$AA^\dagger$ AMP} \label{eq:aa-amp}\\
    \norm*{AS}^2 = \norm*{AS(AS)^\dagger} &\leq (1+\tfrac{\mu\eps}{d\norm{A}})\norm{A}^2
    &\text{by \cref{eq:aa-amp}} \tag{$\|AS\|$ bd} \label{eq:as-op-bound}
\end{align}
The following bounds hold for the sketch $T$.
\begin{align*}
    \fnorm{TAS}^2 &= \fnorm{AS}^2
    &\text{by \cref{def:aamp}} \\
    &\leq 2\fnorm{A}^2
    &\text{by \cref{eq:fas-bound}} \tag{$\|TAS\|_F$ bd} \label{eq:ftas-bound}\\
    \norm*{(AS)^\dagger AS - (TAS)^\dagger TAS} &\leq \tfrac{\mu\eps}{d}\norm{AS}
    &\text{by \cref{cor:amm-usable}} \\
    &\leq 2\tfrac{\mu\eps}{d}\norm{A}
    &\text{by \cref{eq:as-op-bound}} \tag{$(AS)^\dagger AS$ AMP} \label{eq:asas-amp} \\
    \norm*{TAS}^2 = \norm*{(TAS)^\dagger TAS} &\leq (1+\tfrac{\mu\eps}{d\norm{AS}})\norm{AS}^2
    &\text{by \cref{eq:asas-amp}} \\
    &\leq (1+2\tfrac{\mu\eps}{d\norm{A}})\norm{A}^2
    &\text{by \cref{eq:as-op-bound}}
    \tag{$\|TAS\|$ bd} \label{eq:tas-bound}
\end{align*}

\paragraph{One Clenshaw iteration.}
We are trying to perform the odd Clenshaw recurrence defined in \cref{odd-clenshaw}.
The matrix analogue of this is
\begin{align} 
    u_k &= 2(2AA^\dagger- \Id)u_{k+1} - u_{k+2} + 2a_{2k+1} Ab, \tag{Odd Matrix Clenshaw}\label{odd-matrix-clenshaw}\\
    u &= \tfrac12(u_0 - u_1). \nonumber
\end{align}
We now show how to compute the next iterate $u_k$ given $b$ and the previous two iterates as $v_{k+1}, v_{k+2} \in \C^s$ where $u_{k+1} = AS v_{k+1}$ and $u_{k+2} = AS v_{k+2}$, for all $k \geq 0$.
The analysis begins by showing that $u_k$ is $\eps\norm{v_{k+1}}$-close to the output of an exact, zero-error iteration.
Next, we show that $\norm{v_k}$ is $\bigO{d}$-close to its zero-error iteration value.
Finally, we show that these errors don't accumulate too much towards the final outcome.

Starting from \cref{odd-matrix-clenshaw}, we take the following series of approximations by applying intermediate sketches:
\begin{align}
4A A^\dagger u_{k+1} -\,& 2u_{k+1} - u_{k+2} + a_{2k+1} A b \nonumber\\
&\approx_1 4 AS (AS)^\dagger u_{k+1} - 2u_{k+1} - u_{k+2} + a_{2k+1} A b \nonumber\\
    &= AS \Paren{ 4(AS)^\dagger (AS) v_{k+1} - 2v_{k+1} - v_{k+2}} + a_{2k+1} Ab \nonumber\\
    &\approx_2 AS \Paren{ 4(TAS)^\dagger (TAS) v_{k+1} -2v_{k+1} - v_{k+2} } + a_{2k+1}A b \nonumber\\
    &\approx_3 AS \Paren{ 4(TAS)^\dagger (TAS) v_{k+1} -2v_{k+1} - v_{k+2} + a_{2k+1}S^\dagger b} \label{eq:itr-without-best}\\
    & \approx_4 AS (4(TAS)^\dagger B^{(k)} v_{k+1} -2v_{k+1} - v_{k+2} + a_{2k+1}S^\dagger b)\nonumber\\
    &\approx_5 AS (4B_\dagger^{(k)}B^{(k)} v_{k+1} -2v_{k+1} - v_{k+2} + a_{2k+1}S^\dagger b). \label{eq:itr-derivation}
\end{align}
The final expression, \cref{eq:itr-derivation}, is what the algorithm computes, taking
\begin{equation}
    v_k \coloneqq (4B_\dagger^{(k)}B^{(k)} -2I)v_{k+1} - v_{k+2} + a_{2k+1}S^\dagger b. \label{eq:v-derivation}
\end{equation}
Here, $B^{(k)}$ and $B_\dagger^{(k)}$ are taken to be $\ssketch_r(TAS)$ and $\ssketch_r((TAS)^\dagger)$ for $r = \Theta(d^4\fnorm{A}^2(s+t)\frac{1}{\delta}) = \Theta(d^6\frac{\fnorm{A}^4}{(\mu\eps)^2\delta}\log(\frac{\fnorm{A}}{\delta\norm{A}}))$.
The runtime of each iteration is $\bigO{r}$, since the cost of producing the sketches $B^{(k)}$ and $B_\dagger^{(k)}$ using $\sq(TAS)$ is $\bigO{r}$ (\ref{alg:bests}), and actually computing the iteration costs $\bigO{r + s + t} = \bigO{r}$ (\ref{alg:itr}).

\begin{remark} \label{rmk:no-best}
We could have stopped sketching at \cref{eq:itr-without-best}, without using \ssketch, and instead took $v_k \coloneqq 4(TAS)^\dagger(TAS) v_{k+1} -2v_{k+1} - v_{k+2} + a_{2k+1}S^\dagger b$.
The time to compute each iteration increases to $\bigO{st}$, and the final running time is
\begin{equation*}
    \bigO[\Big]{\frac{d^5\|A\|_F^4}{\mu^4\eps^4}\log^2\frac{\fnorm{A}}{\delta\norm{A}}}
\end{equation*}
to achieve the guarantees of \cref{odd-svt-alg} with probability $\geq 1-\delta$.
This running time is worse by a factor of $d^2/\eps^2$, but scales logarithmically in failure probability.
\end{remark}

As for approximation error, let $\eps_1, \eps_2, \eps_3, \eps_4$ and $\eps_5$ be the errors introduced in the approximation steps for \cref{eq:itr-derivation}.
Using the previously established bounds on $S$ and $T$,
\begin{align*}
    \eps_1 &\leq 4\|AA^\dagger - AS(AS)^\dagger \| \|u_{k+1}\| \leq 4\tfrac{\mu\eps}{d}\norm{A}\|u_{k+1}\| \tag*{by \cref{eq:aa-amp}}\\
    \eps_2 &\leq 4\|AS\|\|((AS)(AS)^\dagger - (TAS)(TAS)^\dagger)v_{k+1}\| \leq 16\tfrac{\mu\eps}{d}\norm{A}^2\|v_{k+1}\| \tag*{by \cref{eq:asas-amp}} \\
    \eps_3 &\leq \abs{a_{2k+1}}\|A b - AS S^\dagger b\| \leq \abs{a_{2k+1}}\tfrac{\mu\eps}{d}\norm{b} \tag*{by \cref{eq:ab-amp}}
\end{align*}
The bounds on $\eps_4$ and $\eps_5$ follow from the bounds in \cref{subsec:best-description} applied to $TAS$.
With probability $\geq 1 - \delta/d$, the following hold:
\begin{align*}
\eps_4 &\leq 4\|AS (TAS)^\dagger(TAS - B^{(k)})v_{k+1}\| \\
    & \leq 4\fnorm{AS(TAS)^\dagger}\fnorm{TAS}\norm{v_{k+1}}\sqrt{d/(r\delta)}
    \tag*{by \cref{cor:ssketch-simple}}\\
    & \leq \fnorm{AS(TAS)^\dagger}\fnorm{TAS}\norm{v_{k+1}} \frac{\mu\eps}{12\fnorm{A}^2d^{5/2}} \\
    & \leq d^{-5/2}\mu\eps\norm{A}\norm{v_{k+1}} \tag*{by \cref{eq:as-op-bound,eq:ftas-bound}} \\
\eps_5 &\leq 4\norm*{AS ((TAS)^\dagger - B_\dagger^{(k)})B^{(k)} v_{k+1}} \\
    & \leq 4\fnorm{AS}\fnorm{TAS}\norm*{B^{(k)} v_{k+1}}\sqrt{d/(r \delta)}
    \tag*{by \cref{cor:ssketch-simple}}\\
    & \leq 4\sqrt{d/(r \delta)}\fnorm{AS}\fnorm{TAS}\parens[\Big]{\norm*{TAS v_{k+1}} + \sqrt{d/(r \delta)}\fnorm{I_t}\fnorm{TAS}\norm{v_{k+1}}}
    \tag*{by \cref{cor:ssketch-simple}}\\
    & \leq \tfrac13 d^{-5/2}\mu\eps\parens{\norm*{TAS v_{k+1}} + d^{-3/2}\norm{v_{k+1}}}
    \tag*{by \cref{eq:ftas-bound,eq:tas-bound}}\\
    & \leq d^{-5/2}\mu\eps\norm{v_{k+1}}. \tag*{by $\norm{A} \leq 1$}
\end{align*}
In summary, we can view the iterate of \ref{alg:itr} as computing
\begin{align} 
    \widetilde{u}_k &= 2(2AA^\dagger- \Id)\widetilde{u}_{k+1} - \widetilde{u}_{k+2} + 2a_{2k+1} Ab + \eps^{(k)}, \label{eq:approx-odd-matrix-clenshaw}
\end{align}
where $\eps^{(k)} \in \C^m$ is the error of the approximation in the iterate \cref{eq:itr-derivation}.
We have showed that
\begin{align*}
    \norm*{\eps^{(k)}} \leq \eps_1 + \eps_2 + \eps_3 + \eps_4 + \eps_5 \lesssim \tfrac{\mu\eps}{d}\Big(\norm{v_{k+1}} + \abs{a_{2k+1}}\norm{b}\Big).
\end{align*}
Upon applying a union bound, we see that this bound on $\eps^{(k)}$ holds for every $k$ from $0$ to $d-1$ with probability $\geq 1-3\delta$.

\paragraph{Error accumulation across iterations.}
Now, we analyze how the error from one iteration affects to the final output.
Using the formulation of the iterate from \cref{eq:approx-odd-matrix-clenshaw}, we notice that this is the standard Clenshaw iteration \cref{scalar-clenshaw-recursion} with $x$ replaced with $T_2(A^\dagger) = 2AA^\dagger - I$ and $a_k$ replaced with $2a_{2k+1}Ab + \eps^{(k)}$.
Following \cref{rmk:scalar-stabilities} and \cref{parity-clenshaw-correctness}, we conclude that the output of \cref{algo:odd-dquantizing} satisfies
\begin{align*}
    \widetilde{u}_k &= \sum_{i=k}^d U_{i-k}(T_2(A^\dagger))(2a_{2i+1}Ab + \eps^{(k)}); \\
    \widetilde{u}
    &\coloneqq \frac12(\widetilde{u}_0 - \widetilde{u}_1) \\
    &= \sum_{i=0}^d \tfrac12(U_{i}(T_2(A^\dagger)) - U_{i-1}(T_2(A^\dagger)))(2a_{2i+1}Ab + \eps^{(k)}) \\
    &= \sum_{i=0}^d a_{2i+1}T_{2i+1}(A)b + \sum_{i=0}^d \tfrac12(U_{i}(T_2(A^\dagger)) - U_{i-1}(T_2(A^\dagger)))\eps^{(k)}.
\end{align*}
In other words, after completing the iteration, we have a vector $\widetilde{u}$ such that
\begin{align} 
    \norm{\widetilde{u} - p(A)b}
    &\leq \norm[\Big]{\sum_{i=0}^d \tfrac12(U_{i}(T_2(A^\dagger)) - U_{i-1}(T_2(A^\dagger)))\eps^{(k)}} \nonumber\\
    &\leq \sum_{i=0}^d (2i+1)\norm*{\eps^{(k)}} \nonumber\\
    &\lesssim \mu\eps\sum_{k=0}^d(\|v_{k+1}\| + \abs{a_{2k+1}}\norm{b}) \nonumber\\
    &\leq \eps\norm{b} + \mu\eps\sum_{k=1}^d\|v_k\|.
    \label{eq:odd-clenshaw-error-sum}
\end{align}
The last step follows from \cref{odd-svtcond-ai}.
So, it suffices to bound the $v_k$'s.
Recalling from \cref{eq:v-derivation}, the recursions defining them is
\begin{align*}
    v_k &= 4B_\dagger^{(k)}B^{(k)} v_{k+1} - 2v_{k+1} - v_{k+2} + a_{2k+1}S^\dagger b \\
    &= 2(2 (TAS)^\dagger(TAS) - I) v_{k+1} - v_{k+2} + a_{2k+1} S^\dagger b + 4(B_\dagger^{(k)}B^{(k)} - (TAS)^\dagger(TAS)) v_{k+1}.
    \intertext{This is \cref{odd-matrix-clenshaw} on the matrix $(TAS)^\dagger$ with an additional error term.
    Following \cref{rmk:scalar-stabilities}, this solves to}
    v_k &= \sum_{i=k}^d U_{i-k}(T_2(TAS))\parens[\Big]{a_{2i+1} S^\dagger b + 4(B_\dagger^{(i)}B^{(i)} - (TAS)^\dagger(TAS)) v_{i+1}}.
\end{align*}
Since $B^{(k)}$ and $B_\dagger^{(k)}$ are all drawn independently, $\E[B_\dagger^{(k)}B^{(k)} - (TAS)^\dagger(TAS)]$ is the zero matrix, where the expectation is over the randomness of $B^{(k)}$ and $B_\dagger^{(k)}$.
We use the following bound on the variance of the error.
In this computation, we use \cref{ssketch:moment}, which states that, for $B = \ssketch(A)$ with parameter $r$ and $X$ positive semi-definite, $\E[B^\dagger X B] \leq A^\dagger X A + \frac{1}{r}\tr(X)\fnorm{A}^2I$.
\begin{align}
    &\expecf{k}{\norm{(B_\dagger^{(k)}B^{(k)} - (TAS)^\dagger(TAS)) v_{k+1}}^2} \nonumber \\
    &= \expecf{k}{\norm{B_\dagger^{(k)}B^{(k)} v_{k+1}}^2} - \Norm{(TAS)^\dagger(TAS) v_{k+1}}^2 \nonumber \\
    &= \expecf{k}{(B^{(k)}v_{k+1})^\dagger (B_\dagger^{(k)})^\dagger B_\dagger^{(k)} (B^{(k)} v_{k+1})} - \Norm{(TAS)^\dagger(TAS) v_{k+1}}^2 \nonumber \\
    &\leq \expecf{k}{(B^{(k)}v_{k+1})^\dagger((TAS)(TAS)^\dagger + \frac{1}{r}\Tr(\Id_s)\fnorm{TAS}^2 \Id_t)(B^{(k)} v_{k+1})} - \Norm{(TAS)^\dagger(TAS) v_{k+1}}^2 \nonumber \\
    &\leq \E_k\Bigl[v_{k+1}^\dagger(TAS)^\dagger((TAS)(TAS)^\dagger + \frac{s}{r}\fnorm{TAS}^2 \Id_t)(TAS)v_{k+1} \nonumber \\
    &\qquad + v_{k+1}^\dagger \frac{1}{r}\Tr((TAS)(TAS)^\dagger + \frac{s}{r}\fnorm{TAS}^2\Id_t)\fnorm{TAS}^2 v_{k+1}\Bigr] - \Norm{(TAS)^\dagger(TAS) v_{k+1}}^2 \nonumber \\
    &= \frac{s}{r}\fnorm{TAS}^2\norm{TAS v_{k+1}}^2 + \parens[\Big]{\frac{1}{r} + \frac{st}{r^2}}\fnorm{TAS}^4\norm{v_{k+1}}^2 \nonumber \\
    &\leq 4\parens[\Big]{\frac{s\fnorm{A}^2\norm{A}^2}{r} + \frac{\fnorm{A}^4}{r} + \frac{st\fnorm{A}^4}{r^2}}\norm{v_{k+1}}^2 \tag*{by \cref{eq:ftas-bound,eq:tas-bound}} \\
    &\leq \frac{\delta}{1000d^4}\norm{A}^2\norm{v_{k+1}}^2,
    \label{eq:bbv-variance} 
\end{align}
where the last line uses $r = \Theta(d^4\fnorm{A}^2(s+t)\frac{1}{\delta})$ (and is the bottleneck for the choice of $r$).
\ewin{I think this is not the best choice of $r$ with respect to $\norm{A} \leq 1$.}
Let $\E_{[k, d]}$ denote taking the expectation over $B^{(i)}$ and $B_\dagger^{(i)}$ for $i$ between $k$ and $d$ (treating $T, S$ as fixed).
\begin{equation*}
\bar{v}_k \coloneqq \expecf{[k,d]}{v_k } = \sum_{i=k}^d U_{i-k}(T_2(TAS))a_{2i+1} S^\dagger b .
\end{equation*}
We first bound the recurrence in expectation, then we bound the second moment.
\begin{align*}
    \norm{\bar{v}_k} &\leq \norm[\Big]{\sum_{i=k}^d U_{i-k}(T_2(TAS))a_{2i+1}}\norm*{S^\dagger b} \tag*{by sub-multiplicativity of $\norm{\cdot}$}\\
    &= \norm[\Big]{\sum_{i=k}^d U_{i-k}(T_2(x))a_{2i+1}}_{\spec(TAS)}\norm*{S^\dagger b} \\
    &\leq \norm[\Big]{\sum_{i=k}^d U_{i-k}(T_2(x))a_{2i+1}}_{[-1-2\frac{\mu\eps}{d}, 1+2\frac{\mu\eps}{d}]}\norm*{S^\dagger b} \tag*{by \cref{eq:tas-bound}}\\
    &\leq e\supnorm[\Big]{\sum_{i=k}^d U_{i-k}(T_2(x))a_{2i+1}}\norm*{S^\dagger b} \tag*{by \cref{slightly-beyond-sup}, $\mu\eps \leq \frac{1}{100d}$} \\
    &\leq 4\supnorm[\Big]{\sum_{i=k}^d U_{i-k}(T_2(x))a_{2i+1}}\norm{b} \tag*{by \cref{eq:sb-bound}} \\
    &\leq 4\frac{1}{\mu d}\norm{b} \tag*{by \cref{odd-svtcond-trunc}}
\end{align*}
We now compute the second moment of $v_k$.
\begin{align}
    \expecf{[k,d]}{ \Norm{v_k - \bar{v}_k}^2}
    &= \expecf{[k,d]}{\norm[\Big]{\sum_{i=k}^d U_{i-k}(T_2(TAS))4(B_\dagger^{(i)}B^{(i)} - (TAS)^\dagger(TAS)) v_{i+1}}^2} \nonumber
\intertext{Because the $B^{(i)}$'s are independent, the variance of the sum is the sum of the variances, so}
    &= \sum_{i=k}^d \expecf{[k,d]}{\Norm{U_{i-k}(T_2(TAS))4(B_\dagger^{(i)}B^{(i)} - (TAS)^\dagger(TAS)) v_{i+1}}^2 } \nonumber \\
    &\leq 16\sum_{i=k}^d \norm[\Big]{U_{i-k}(T_2(TAS))}^2 \expecf{[k,d]}{\norm{(B_\dagger^{(i)}B^{(i)} - (TAS)^\dagger(TAS)) v_{i+1}}^2 } \nonumber \\
    &\leq 16\sum_{i=k}^d e^2d^2 \expecf{[k,d]}{\norm{(B_\dagger^{(i)}B^{(i)} - (TAS)^\dagger(TAS)) v_{i+1}}^2 } \nonumber \\
    &\leq 16e^2\sum_{i=k}^d \frac{d^2\delta}{1000d^4} \norm{A}^2 \expecf{[i+1,d]}{\norm{v_{i+1}}^2} \nonumber \\
    &\leq \frac{\delta}{2d^2}\norm{A}^2\sum_{i=k}^d \expecf{[i+1,d]}{\norm{v_{i+1}}^2} \nonumber \\
    &= \frac{\delta}{2d^2}\norm{A}^2\sum_{i=k}^d \parens*{\expecf{[i+1,d]}{\norm{v_{i+1} - \bar{v}_{i+1}}^2} + \norm{\bar{v}_{i+1}}^2} \nonumber \\
    &\leq \frac{\delta}{2d^2}\norm{A}^2\sum_{i=k}^d \parens*{\expecf{[i+1,d]}{\norm{v_{i+1} - \bar{v}_{i+1}}^2} + \frac{16}{\mu^2 d^2}\norm{b}^2} \tag*{by \cref{eq:bbv-variance}}
\end{align}
To bound this recurrence, we define the following recurrence $c_k$ to satisfy $\E_{[k,d]}[\norm{v_k - \bar{v}_k}^2] \leq c_k$:
\begin{align*}
    c_k &= \gamma\sum_{i=k}^d (c_{i+1}+ \Gamma) \qquad \gamma = \frac{\delta}{2d^2}\norm{A}^2, \, \Gamma = \frac{16}{\mu^2 d^2}\norm{b}^2.
\intertext{For this recurrence, $c_k \leq d \gamma \Gamma$ for all $k$ between $0$ and $d$ provided that $d \gamma \leq \frac{1}{2}$.}
    c_k &\leq \frac{8\delta}{\mu^2 d^3}\norm{A}^2\norm{b}^2
\end{align*}
We have shown that $\E[\norm{v_k - \bar{v}_k}^2] \leq \frac{8\delta}{d}\parens*{\frac{\norm{A}\norm{b}}{\mu d}}^2$.
By Markov's inequality, with probability $\geq 1 - \delta/100$, we have that for all $k$, $\norm{v_k} \lesssim \frac{\norm{b}}{\mu d}$.
Returning to the final error bound \cref{eq:odd-clenshaw-error-sum},
\begin{align}
    \norm{\widetilde{u} - p(A)b}
    &\lesssim \eps\norm{b} + \mu\eps\sum_{k=1}^d\norm{v_k}
    \lesssim \eps\norm{b}.
\end{align}

\paragraph{Output description properties.}
After the iteration concludes, we can compute $u$ by computing $x = \tfrac{1}{2}S(v_0 - v_1)$ in linear $\bigO{s}$ time.
Then, $u = \tfrac12(u_0 - u_1) = \tfrac12 AS(v_0 - v_1) = Ax$.
Note that though $x \in \C^n$, its sparsity is at most the sparsity of $x$, which is bounded by $s$.

Further, using the prior bounds on $v_0$ and $v_1$, we have that
\begin{align*}
    \sum_{j=1}^n \norm{A_{*,j}}^2\abs{x_i}^2
    &= \sum_{j=1}^s \norm{[SA]_{*,j}}^2\abs{\tfrac12(v_0 - v_1)_i}^2 \\
    &\leq \sum_{j=1}^s \tfrac2s\fnorm{A}^2 \abs{\tfrac12(v_0 - v_1)_i}^2 \\
    &\leq \tfrac2s\fnorm{A}^2\norm{\tfrac12(v_0 - v_1)}^2 \\
    &\leq \tfrac2s\fnorm{A}^2(\norm{v_0}^2 + \norm{v_1}^2) \\
    &\lesssim \fnorm{A}^2\norm{b}^2/(\sqrt{s}\mu d)^2 \\
    &\lesssim \eps^2\norm{b}^2/(d^4\log(\tfrac{\fnorm{A}}{\delta\norm{A}})).
\end{align*}

\subsection{Generalizing to Even Polynomials} \label{subsec:evensvt}

We also obtain an analogous result for even polynomials.
For the most part, changes are superficial; the red text indicates differences from \cref{odd-svt-alg}.
The major difference is that the representation of the output is $Ax + \eta b$ instead of $Ax$, which results from constant terms being allowed when $p$ is even.
We state the theorem for estimating $p(A^\dagger)b$ because it makes the similarities with the odd setting more apparent.

\begin{theorem} \label{even-svt-alg}
    Suppose we are given sampling and query access to $A \in \mathbb{C}^{m\times n}$ and $b \in \diff{\mathbb{C}^m}$ with $\norm{A} \leq 1$; a \diff{$(2d)$-degree even} polynomial, written in its Chebyshev coefficients as
    \begin{align*}
        \diff{p(x) = \sum_{i=0}^{d} a_{2i} T_{2i}(x);}
    \end{align*}
    an accuracy parameter $\eps > 0$; a failure probability parameter $\delta > 0$; and a stability parameter $\mu > 0$.
    Then we can output a vector $x \in \C^n$ \diff{and $\eta \in \C$} such that $\|A x \diff{+ \eta b} - p(\diff{A^\dagger})b\| \leq \eps\supnorm{p}\norm{b}$ with probability $\geq 1-\delta$ in time
    \begin{align*}
        \bigO*{\min\braces{\nnz(A), \frac{d^4\|A\|_F^4}{(\mu\eps)^4}\log^2\parens[\big]{\frac{\fnorm{A}}{\delta\norm{A}}}} + \frac{d^7\|A\|_F^4}{(\mu\eps)^2\delta}\log\parens[\big]{\frac{\fnorm{A}}{\delta\norm{A}}}},
    \end{align*}
    assuming $\mu\eps < \min(\frac{1}{4}d\norm{A}, \frac{1}{100d})$ and $\mu$ satisfies the following bounds.
    \diff{Below, $\tilde{a}_{2k} \coloneqq a_{2k} - a_{2k+2} + \cdots \pm a_{2d}$.}
    \begin{enumerate}[label=(\alph*),ref=\thetheorem (\alph*)]
        \item \diff{$\mu\sum_{i=1}^d \abs{\tilde{a}_{2i}} \leq \supnorm{p}$ and $d\mu^2\sum_{i=1}^d \abs{\tilde{a}_{2i}}^2 \leq \supnorm{p}^2$}; \label{even-svtcond-ai}
        \item \diff{$\mu\supnorm{\sum_{i=k}^d 4\tilde{a}_{2i+2}x\cdot U_{i-k}(T_2(x))} \leq \frac{1}{d}\supnorm{p}$} for all $0 \leq k \leq d$. \label{even-svtcond-trunc}
    \end{enumerate}
    The output description has the additional properties
    \begin{gather*}
        \sum_j \norm{A_{*,j}}^2\abs{x_j}^2 \lesssim \frac{\eps^2\supnorm{p}^2\norm{b}^2}{d^4\log(\frac{\fnorm{A}}{\delta\norm{A}})}
        \qquad \znorm{x} \lesssim \frac{d^2\fnorm{A}^2}{(\mu\eps)^2}\log(\frac{\fnorm{A}}{\delta\norm{A}}),
    \end{gather*}
    so that by \cref{Ax-sampling}, for the output vector $y \coloneqq A x \diff{+ \eta b}$, we can:
    \begin{enumerate}[label=(\roman*)]
        \item Compute entries of $y$ in $\bigO{\znorm{x}} = \bigO[\Big]{ \frac{d^2\fnorm{A}^2}{(\mu\eps)^2}\log(\frac{\fnorm{A}}{\delta\norm{A}}) }$ time;
        \item Sample $i \in [n]$ with probability $\frac{\abs{y_i}^2}{\norm{y}^2}$ in $\bigO[\Big]{\parens[\big]{\frac{\supnorm{p}^2\fnorm{A}^2}{(\mu d)^2} \diff{+ p(0)^2}}\frac{d^2\fnorm{A}^2\norm{b}^2}{\mu^2\eps^2\norm{y}^2}\log(\frac{\fnorm{A}}{\delta\norm{A}})\log\frac{1}{\delta}}$ time with probability $\geq 1-\delta$;
        \item Estimate $\norm{y}^2$ to $\nu$ relative error in $\bigO[\Big]{\parens[\big]{\frac{\supnorm{p}^2\fnorm{A}^2}{(\mu d)^2} \diff{+ p(0)^2}}\frac{d^2\fnorm{A}^2\norm{b}^2}{\nu^2\mu^2\eps^2\norm{y}^2}\log(\frac{\fnorm{A}}{\delta\norm{A}})\log\frac{1}{\delta}}$ time with probability $\geq 1-\delta$.
    \end{enumerate}
\end{theorem}

\begin{corollary}[Corollary of \cref{even-iterate-mu-bound}] \label{cor:even-iterate-mu-bound}
    In \cref{even-svt-alg}, we can always take $1/\mu \lesssim d^2\diff{\log(d)}$ for $d > 1$.
\end{corollary}

Recall the odd and even recurrences defined in \cref{odd-clenshaw,even-clenshaw}.
\begin{align} 
    u_k &= 2(2AA^\dagger- \Id)u_{k+1} - u_{k+2} + 2a_{2k+1} Ab, \tag{Odd Matrix Clenshaw}\\
    p(A)b = u &= \tfrac12(u_0 - u_1). \nonumber
\end{align}
The matrix analogue of the even recurrence is identical except that the final term is $4\tilde{a}_{2k+2}A A^\dagger b$ instead of $2a_{2k+1}Ab$.
\begin{align}
    \tilde{a}_{2k} &\coloneqq a_{2k} - a_{2k+2} + a_{2k+4} - \cdots \pm a_{2d} \nonumber\\
    u_k &= 2(2A A^\dagger -1)u_{k+1} - u_{k+2} \diff{+ 4\tilde{a}_{2k+2} A A^\dagger b}, \tag{Even Matrix Clenshaw}\label{even-matrix-clenshaw}\\
    p(A^\dagger)b = u &= \diff{\tilde{a}_0b +} \tfrac12(u_0 - u_1). \nonumber
\end{align}
So, a roughly identical analysis works upon making the appropriate changes.
As before, we assume $\supnorm{p} = 1$ without loss of generality.

\begin{mdframed}
  \begin{algorithm}[Even singular value transformation]
    \label{algo:even-dquantizing}\mbox{}
    \begin{description}
    \item[Input (pre-processing):] A matrix $A \in \mathbb{C}^{m \times n}$, vector $b \in \diff{\mathbb{C}^{m}}$, and parameters $\eps,\,\delta,\,\mu>0$.
    \item[Pre-processing sketches:]
    Let $s, t = \Theta\Paren{\frac{d^2\fnorm{A}^2}{(\mu\eps)^2}\log(\frac{\fnorm{A}}{\delta\norm{A}})}$.
    This phase will succeed with probability $\geq 1-\delta$.
    \begin{enumerate}[label=P\arabic*., ref=A\thealgorithm.P\arabic*]
        \item If $\sq(A^\dagger)$ and $\sq(b)$ are not given, compute data structures to simulate them in $\bigO{1}$ time;
        \item Sample $S \in \C^{n \times s}$ to be $\aamp_s\parens[\Big]{A, A^\dagger, \diff{\{\frac{\norm*{A_{*,j}}^2}{\fnorm{A}^2}\}_{j \in [n]}}}$ (\cref{def:aamp});
        \item Sample $T^\dagger \in \C^{m \times t}$ to be $\aamp_t\parens[\Big]{S^\dagger A^\dagger, AS, \diff{\{\tfrac12(\frac{\norm*{[AS]_{i,*}}^2}{\fnorm{AS}^2} + \frac{\abs*{b_i}^2}{\norm{b}^2})\}_{i \in [m]}}}$ (\cref{def:aamp});
        \item Compute a data structure that can respond to $\sq(TAS)$ queries in $\bigO{1}$ time;
    \end{enumerate}
    \item[Input:] A degree \diff{$2d$} polynomial \diff{$p(x) = \sum_{i = 0}^{d} a_{2i} T_{2i}(x)$} given as its coefficients \diff{$a_{2i}$}. \\
    \diff{Compute all $\tilde{a}_{2k} = a_{2k} - a_{2k+2} + \cdots \pm a_{2d}$.}
    \item[Clenshaw iteration:]
    Let $r = \Theta(d^4\fnorm{A}^2(s+t)\frac{1}{\delta}) = \Theta(d^6\frac{\fnorm{A}^4}{(\mu\eps)^2\delta}\log(\frac{\fnorm{A}}{\delta\norm{A}}))$.
    This phase will succeed with probability $\geq 1-\delta$.
    Starting with $v_{d+1} = v_{d+2} = \vec{0}^s$ and going until $v_0$,
    \begin{enumerate}[label=I\arabic*., ref=A\thealgorithm.I\arabic*]
        \item Let $B^{(k)} = \ssketch_r(TAS)$ and $B_\dagger^{(k)} = \ssketch_r((TAS)^\dagger)$ (\cref{def:ssketch});
        \item Compute $v_k = 2(2B_\dagger^{(k)} B^{(k)} - I)v_{k+1} - v_{k+2} + \diff{4\tilde{a}_{2k+2} B_\dagger^{(k)} T b}$.
    \end{enumerate}
    \item[Output:] 
    Output $x = \tfrac12 S(v_0 - v_1)$ \diff{and $\eta = \tilde{a}_0$} satisfying $\Norm{ A x \diff{+ \eta b} - p(\diff{A^\dagger}) b } \leq \eps\supnorm{p}\norm{b}$.
    \end{description}
  \end{algorithm}
\end{mdframed}

\paragraph{Pre-processing sketches, correctness.}

Though the sketches are chosen to be slightly different because of the different parity, all of the sketching bounds used for the odd SVT analysis hold here, up to rescaling $s, t$ by constant factors.
This includes \cref{eq:as-bound,,eq:fas-bound,,eq:aa-amp,,eq:as-op-bound,,eq:ftas-bound,,eq:asas-amp,,eq:tas-bound}.
What remains (\cref{eq:sb-bound,,eq:ab-amp}) have analogues that follow from the same argument:
\begin{align}
    \norm{Tb}^2 &\leq 2\norm{b}^2 \tag{$\|Tb\|$ bd} \label{eq:tb-bound} \\
    \norm*{(AS)^\dagger b - (TAS)^\dagger Tb} &\leq \frac{\mu\eps}{d} \norm{b} \tag{$(AS)^\dagger b$ AMP} \label{eq:asb-amp}
\end{align}
All this holds with probability $\geq 1-\delta$.

\paragraph{One (even) Clenshaw iteration.}

As with the odd case, we perform the recurrence on $u_k$ by updating $v_k$ such that $u_k = (AS) v_k$.
The error analysis proceeds by bounding
\begin{align}
4A A^\dagger u_{k+1} -\,& 2u_{k+1} - u_{k+2} + 4\tilde{a}_{2k+2} A A^\dagger b \nonumber\\
&\approx_1 4 AS(AS)^\dagger u_{k+1} - 2u_{k+1} - u_{k+2} + 4\tilde{a}_{2k+2} A A^\dagger b \nonumber\\
    &= AS \Paren{ 4(AS)^\dagger(AS) v_{k+1} - 2v_{k+1} - v_{k+2}} + 4\tilde{a}_{2k+2}AA^\dagger b \nonumber\\
    &\approx_2 AS \Paren{ 4(TAS)^\dagger (TAS) v_{k+1} -2v_{k+1} - v_{k+2} } + 4\tilde{a}_{2k+2}A A^\dagger b \nonumber\\
    &\approx_3 AS \Paren{ 4(TAS)^\dagger (TAS) v_{k+1} -2v_{k+1} - v_{k+2} + 4\tilde{a}_{2k+2}(AS)^\dagger b } \nonumber\\
    &\approx_4 AS \Paren{ 4(TAS)^\dagger B^{(k)} v_{k+1} -2v_{k+1} - v_{k+2} + 4\tilde{a}_{2k+2}(AS)^\dagger b } \nonumber\\
    &\approx_5 AS \Paren{ 4B_\dagger^{(k)} B^{(k)} v_{k+1} -2v_{k+1} - v_{k+2} + 4\tilde{a}_{2k+2}(AS)^\dagger b } \nonumber\\
    &\approx_6 AS \Paren{ 4(TAS)^\dagger (TAS) v_{k+1} -2v_{k+1} - v_{k+2} + 4\tilde{a}_{2k+2}(TAS)^\dagger Tb } \nonumber\\
    &\approx_7 AS \Paren{ 4B_\dagger^{(k)} B^{(k)} v_{k+1} -2v_{k+1} - v_{k+2} + 4\tilde{a}_{2k+2} B_\dagger^{(k)} Tb }. \label{eq:even-derivation}
\end{align}
So, our update is
\begin{align} \label{eq:even-v-derivation}
    v_k = 4B_\dagger^{(k)} B^{(k)} v_{k+1} -2v_{k+1} - v_{k+2} + 4\tilde{a}_{2k+2} B_\dagger^{(k)} Tb.
\end{align}
As before, we can label the error incurred by each approximation in \cref{eq:even-derivation} with $\eps_1, \ldots, \eps_7$.
The approximations in $\eps_1, \eps_2, \eps_4, \eps_5$ do not involve the constant term and so can be bounded identically to the odd case.
\begin{align*}
    \eps_1 &\leq 4\|AA^\dagger - AS(AS)^\dagger \| \|u_{k+1}\| \leq 4\tfrac{\mu\eps}{d}\norm{A}\|u_{k+1}\| \\
    \eps_2 &\leq 4\|AS\|\|((AS)(AS)^\dagger - (TAS)(TAS)^\dagger)v_{k+1}\| \leq 16\tfrac{\mu\eps}{d}\norm{A}^2\|v_{k+1}\| \\
    \eps_4 &\leq 4\|AS (TAS)^\dagger(TAS - B^{(k)})v_{k+1}\| 
    \leq d^{-5/2}\mu\eps\norm{A}\norm{v_{k+1}} \\
    \eps_5 &\leq 4\norm*{AS ((TAS)^\dagger - B_\dagger^{(k)})B^{(k)} v_{k+1}}
    \leq d^{-5/2}\mu\eps\norm{v_{k+1}}.
\intertext{The approximation in $\eps_3$ goes through with a slight modification.}
    \eps_3 &\leq \abs{4\tilde{a}_{2k+2}}\|AA^\dagger b - AS (AS)^\dagger b\| \leq 4\abs{\tilde{a}_{2k+2}}\tfrac{\mu\eps}{d}\norm{A}\norm{b} \tag*{by \cref{eq:aa-amp}}
\intertext{The approximations $\eps_6$ and $\eps_7$ follow from similar arguments.}
    \eps_6 &\leq \abs{4\tilde{a}_{2k+2}} \norm{AS}\norm*{(AS)^\dagger b - (TAS)^\dagger Tb} \\
    &\leq \abs{8\tilde{a}_{2k+2}} \tfrac{\mu\eps}{d}\norm{b} \tag*{by \cref{eq:as-bound,eq:asb-amp}} \\
    \eps_7 &\leq \abs{4\tilde{a}_{2k+2}} \norm*{AS((TAS)^\dagger - B_\dagger^{(k)}) Tb} \\
    &\leq \abs{4\tilde{a}_{2k+2}} \fnorm{AS}\fnorm{TAS}\norm{Tb}\sqrt{d/(r\delta)} \tag*{by \cref{cor:ssketch-simple}} \\
    &\leq \abs{\tilde{a}_{2k+2}} d^{-5/2}\mu\eps\norm{b} \tag*{by \cref{eq:tas-bound,eq:tb-bound}}
\end{align*}
In summary, we can view the iterate of \ref{alg:itr} as computing
\begin{align} 
    \widetilde{u}_k &= 2(2AA^\dagger- \Id)\widetilde{u}_{k+1} - \widetilde{u}_{k+2} + 4\tilde{a}_{2k+2} AA^\dagger b + \eps^{(k)}, \label{eq:approx-even-matrix-clenshaw}
\end{align}
where $\eps^{(k)} \in \C^m$ is the error of the approximation in the iterate \cref{eq:itr-derivation}, and
\begin{align*}
    \norm*{\eps^{(k)}} &\leq \eps_1 + \eps_2 + \eps_3 + \eps_4 + \eps_5 + \eps_6 + \eps_7 \\
    &\lesssim \tfrac{\mu\eps}{d}\Big(\norm{v_{k+1}} + \abs{\tilde{a}_{2k+2}}\norm{b}\Big).
\end{align*}
Upon applying a union bound, we see that this bound on $\eps^{(k)}$ holds for every $k$ from $0$ to $d-1$ with probability $\geq 1-4\delta$.

\paragraph{Error accumulation across iterations.}
The error accumulates in the same way as in the odd setting.
Using the formulation of the iterate from \cref{eq:approx-even-matrix-clenshaw}, we notice that this is the standard Clenshaw iteration \cref{scalar-clenshaw-recursion} with $x$ replaced with $T_2(A^\dagger) = 2AA^\dagger - I$ and $a_k$ replaced with $4\tilde{a}_{2k+2}AA^\dagger b + \eps^{(k)}$.
Following \cref{rmk:scalar-stabilities} and \cref{parity-clenshaw-correctness}, we conclude that the output of \cref{algo:even-dquantizing} satisfies
\begin{align*}
    \widetilde{u}_k &= \sum_{i=k}^d U_{i-k}(T_2(A^\dagger))(4\tilde{a}_{2i+2}AA^\dagger b + \eps^{(i)}); \\
    \widetilde{u}
    &\coloneqq \tilde{a}_0 + \frac12(\widetilde{u}_0 - \widetilde{u}_1) \\
    &= \tilde{a}_0 + \sum_{i=0}^d \tfrac12(U_{i}(T_2(A^\dagger)) - U_{i-1}(T_2(A^\dagger)))(4\tilde{a}_{2i+2}AA^\dagger b + \eps^{(k)}) \\
    &= \sum_{i=0}^d a_{2i}T_{2(i-k)}(A^\dagger)b + \sum_{i=0}^d \tfrac12(U_{i}(T_2(A^\dagger)) - U_{i-1}(T_2(A^\dagger)))\eps^{(k)}.
\end{align*}
In other words, after completing the iteration, we have a vector $\widetilde{u}$ such that
\begin{align} 
    \norm*{\widetilde{u} - p(A^\dagger)b}
    &\leq \norm[\Big]{\sum_{i=0}^d \tfrac12(U_{i}(T_2(A^\dagger)) - U_{i-1}(T_2(A^\dagger)))\eps^{(k)}} \nonumber\\
    &\leq \sum_{i=0}^d (2i+1)\norm*{\eps^{(k)}} \nonumber\\
    &\lesssim \mu\eps\sum_{k=0}^d(\|v_{k+1}\| + \abs{\tilde{a}_{2k+2}}\norm{b}) \nonumber\\
    &\leq \eps\norm{b} + \mu\eps\sum_{k=1}^d\|v_k\|.
    \label{eq:even-clenshaw-error-sum}
\end{align}
In the last line, we use the assumption \cref{even-svtcond-ai}.
It suffices to bound the $v_k$'s.
Recalling from \cref{eq:even-v-derivation}, the recursions defining them is
\begin{align*}
    v_k &= 4B_\dagger^{(k)} B^{(k)} v_{k+1} -2v_{k+1} - v_{k+2} + 4\tilde{a}_{2k+2} B_\dagger^{(k)} Tb \\
    &= 2(2 (TAS)^\dagger(TAS) - I) v_{k+1} - v_{k+2} + 4\tilde{a}_{2k+2} B_\dagger^{(k)} Tb + 4(B_\dagger^{(k)}B^{(k)} - (TAS)^\dagger(TAS)) v_{k+1}.
\intertext{Following \cref{rmk:scalar-stabilities}, this solves to}
    v_k &= \sum_{i=k}^d U_{i-k}(T_2(TAS))\parens[\Big]{4\tilde{a}_{2i+2} B_\dagger^{(i)} Tb + 4(B_\dagger^{(i)}B^{(i)} - (TAS)^\dagger(TAS)) v_{i+1}}.
\end{align*}
From here, the identical analysis applies.
\begin{align*}
    \bar{v}_k \coloneqq \expecf{[k,d]}{v_k } &= \sum_{i=k}^d U_{i-k}(T_2(TAS))4\tilde{a}_{2i+2} (TAS)^\dagger Tb \\
    \norm*{\bar{v}_k} &\leq \norm[\Big]{\sum_{i=k}^d U_{i-k}(T_2(TAS))4\tilde{a}_{2i+2} (TAS)^\dagger}\norm{Tb} \leq 4\tfrac{1}{\mu d}\norm{b}.
\end{align*}
We now compute the second moment of $v_k$.
The main difference from the odd setting is that there is an additional term, $4\tilde{a}_{2k+2}(B_\dagger^{(k)} - (TAS)^\dagger)Tb$, where 
\begin{align*}
    \E_k\norm{(B_\dagger^{(k)} - (TAS)^\dagger)Tb}^2
    &\leq \frac{1}{r}\tr(I_s)\fnorm{TAS}^2\norm{Tb}^2 \tag*{by \cref{ssketch:moment}}\\
    &\leq \frac{\delta}{1000d^4}\norm{b}^2.\tag*{by \cref{eq:tb-bound}, $r = \Theta(d^4\fnorm{A}^2(s+t)\frac1\delta)$}
\end{align*}
We use this and the derivation in \cref{eq:bbv-variance} to conclude
\begin{align}
    &\quad \expecf{[k,d]}{ \Norm{v_k - \bar{v}_k}^2} \nonumber\\
    &= \expecf{[k,d]}{\norm[\Big]{\sum_{i=k}^d U_{i-k}(T_2(TAS))4(\tilde{a}_{2i+2} (B_\dagger^{(i)} - (TAS)^\dagger) Tb + (B_\dagger^{(i)}B^{(i)} - (TAS)^\dagger(TAS)) v_{i+1})}^2} \nonumber \\
    &= \sum_{i=k}^d \expecf{[k,d]}{\Norm{U_{i-k}(T_2(TAS))4(\tilde{a}_{2i+2} (B_\dagger^{(i)} - (TAS)^\dagger) Tb + (B_\dagger^{(i)}B^{(i)} - (TAS)^\dagger(TAS)) v_{i+1})}^2} \nonumber \\
    &\leq 16\sum_{i=k}^d \norm[\Big]{U_{i-k}(T_2(TAS))}^2 \expecf{[k,d]}{\norm{\tilde{a}_{2i+2} (B_\dagger^{(i)} - (TAS)^\dagger) Tb + (B_\dagger^{(i)}B^{(i)} - (TAS)^\dagger(TAS)) v_{i+1}}^2} \nonumber \\
    &\leq 32\sum_{i=k}^d e^2d^2 \expecf{[k,d]}{\norm{\tilde{a}_{2i+2} (B_\dagger^{(i)} - (TAS)^\dagger) Tb}^2 + \norm{(B_\dagger^{(i)}B^{(i)} - (TAS)^\dagger(TAS)) v_{i+1}}^2} \nonumber \\
    &\leq 32\sum_{i=k}^d e^2d^2 \expecf{[k,d]}{\frac{\delta\abs*{\tilde{a}_{2i+2}}^2}{1000d^4}\norm{b}^2 + \frac{\delta}{1000d^4}\norm{A}^2\norm{v_{k+1}}^2 } \nonumber \\
    &\leq \frac{\delta}{4d^2}\norm{b}^2\sum_{i=k}^d \abs*{\tilde{a}_{2i+2}}^2 + \frac{\delta\norm{A}^2}{4d^2}\sum_{i=k}^d \expecf{[k,d]}{\norm{v_{k+1}}^2 } \nonumber \\
    &\leq \frac{\delta}{4\mu^2d^3}\norm{b}^2 + \frac{\delta\norm{A}^2}{4d^2}\sum_{i=k}^d \expecf{[k,d]}{\norm{v_{k+1}}^2 } \tag*{by \cref{even-svtcond-ai}} \\
    &= \frac{\delta}{4\mu^2d^3}\norm{b}^2 + \frac{\delta\norm{A}^2}{4d^2}\sum_{i=k}^d \parens*{\expecf{[i+1,d]}{\norm{v_{i+1} - \bar{v}_{i+1}}^2} + \norm{\bar{v}_{i+1}}^2} \nonumber \\
    &\leq \frac{\delta}{4\mu^2d^3}\norm{b}^2 + \frac{\delta\norm{A}^2}{4d^2}\sum_{i=k}^d \parens*{\expecf{[i+1,d]}{\norm{v_{i+1} - \bar{v}_{i+1}}^2} + \frac{16}{\mu^2 d^2}\norm{b}^2} \tag*{by \cref{eq:bbv-variance}} \\
    &\leq \frac{\delta\norm{A}^2}{4d^2}\sum_{i=k}^d \parens*{\expecf{[i+1,d]}{\norm{v_{i+1} - \bar{v}_{i+1}}^2} + \frac{17}{\mu^2 d^2}\norm{b}^2} \nonumber
\intertext{By the same recurrence argument, this is }
    &\leq \frac{9\delta}{\mu^2 d^3}\norm{A}^2\norm{b}^2.
\end{align}
We have shown that $\E[\norm{v_k - \bar{v}_k}^2] \leq \frac{9\delta}{d}\parens*{\frac{\norm{A}\norm{b}}{\mu d}}^2$.
By Markov's inequality, with probability $\geq 1 - \delta/100$, we have that for all $k$, $\norm{v_k} \lesssim \frac{\norm{b}}{\mu d}$.
Returning to the final error bound \cref{eq:even-clenshaw-error-sum},
\begin{align}
    \norm*{\widetilde{u} - p(A^\dagger)b}
    &\lesssim \eps\norm{b} + \mu\eps\sum_{k=1}^d\norm{v_k}
    \lesssim \eps\norm{b}.
\end{align}

\paragraph{Output description properties.}
The argument from the odd case shows that
\begin{equation*}
    \sum_j \norm*{A_{*,j}}^2\abs*{x_j}^2
    \lesssim \frac{\eps^2\norm{b}^2}{d^4\log(\frac{\fnorm{A}}{\delta\norm{A}})}
\end{equation*}
and $\znorm{x} \leq s$.
By \cref{Ax-sampling} we get the desired bounds.
Notice that $\tilde{a}_0 = a_0 - a_2 + a_4 - \cdots \pm a_{2d} = p(0)$.

\ewin{Reminder to eventually add stuff that gets $O(\nnz(A) + n \cdots)$}

\subsection{Bounding Iterates of Singular Value Transformation}
We give bounds on the value of $\mu$ that suffices for a generic polynomial.
Though bounds may be improvable for specific functions, they are tight up to $\log$ factors for Chebyshev polynomials $T_k(x)$, and improve by a factor of $d$ over naive coefficient-wise bounds.

\begin{proposition} \label{odd-iterate-mu-bound}
    Let $p(x)$ be an odd polynomial with degree $2d+1$, with a Chebyshev series expansion of $p(x) = \sum_{i=0}^d a_{2i+1}T_{2i+1}(x)$.
    Then $\sum_{i=0}^d \abs*{a_{2i+1}} \leq 2d\supnorm{p}$ and, for all integers $k \leq d$,
    \begin{align*}
        \supnorm[\Big]{\sum_{i=k}^d a_{2i+1}U_{i-k}(T_2(x))} \lesssim (d-k+1)(1+\log^2(d+1))\supnorm{p}.
    \end{align*}
\end{proposition}

\begin{proof}
Without loss of generality, we take $\supnorm{p} = 1$.
By \cref{lem:coefficient-bound}, $\abs*{a_{2i+1}} \leq 2$, giving the first conclusion.
Towards the second conclusion, we first note that
\begin{align*}
    \supnorm[\Big]{\sum_{i=k}^d a_{2i+1} U_{i-k}(T_2(x))}
    = \supnorm[\Big]{\sum_{i=k}^d a_{2i+1} U_{i-k}(x)},
\end{align*}
since $T_2$ maps $[-1, 1]$ to $[-1, 1]$.
Then, we use the strategy from \cref{thm:scalar-clenshaw-iterate-bound}, writing out $U_{i-k}$ with \cref{eq:u-to-t} and then bounding the resulting coefficients.
Note that, by convention, $a_i$ and $U_i$ are zero for any integer $i \in \mathbb{Z}$ for which they are not defined.
\begin{align*}
    \sum_{i=k}^d a_{2i+1} U_{i-k}(x)
    &= \sum_i a_{2i+1} U_{i-k}(x) \\
    &= \sum_i a_{2i+1} \sum_{j \geq 0} T_{i-k-2j}(x)(1 + \iver{i - k - 2j \neq 0}) \\ 
    &= \sum_{j \geq 0} \sum_i a_{2i+1} T_{i-k-2j}(x)(1 + \iver{i - k - 2j \neq 0}) \\ 
    &= \sum_{j \geq 0} \sum_i a_{2(i + k + 2j) + 1}  T_{i}(x)(1 + \iver{i \neq 0}) \\
    &= \sum_i T_{i}(x) (1 + \iver{i \neq 0}) \sum_{j \geq 0}  a_{2(i + k + 2j) + 1} \\
    \supnorm[\Big]{\sum_{i=k}^d a_{2i+1} U_{i-k}(x)}
    &\leq \sum_i (1 + \iver{i \neq 0}) \supnorm{T_i(x)} \abs*[\Big]{\sum_{j \geq 0}  a_{2(i + k + 2j)+1}} \\
    &= \sum_{i = 0}^{d-k} (1 + \iver{i \neq 0}) \abs*[\Big]{\sum_{j \geq 0}  a_{2(i + k) + 1 + 4j}} \\
    &\leq \sum_{i = 0}^{d-k} (1 + \iver{i \neq 0}) (32 + 8\log^2(2d+2)) \tag*{by \cref{cor:four-step-sums}}\\
    &= (2(d-k)+1)(32 + 8\log^2(2d+2)) \\
    &\lesssim (d-k+1)(1+\log^2(d+1))
\end{align*}
\end{proof}

\begin{proposition} \label{even-iterate-mu-bound}
    Let $p(x)$ be an even polynomial with degree $2d$, written as $p(x) = \sum_{i=0}^d a_{2i}T_{2i}(x)$.
    Let $\tilde{a}_{2k} \coloneqq a_{2k} - a_{2k+2} + \cdots \pm a_{2d}$.
    Then
    \begin{align*}
        \sum_{i=0}^d \abs*{\tilde{a}_{2i}} &\leq 4(d+1)(1+\log(d+1))\supnorm{p}, \\
        \text{and }\sum_{i=0}^d \abs*{\tilde{a}_{2i}}^2 &\leq 32(d+1)(1+\log^2(d+1))\supnorm{p},
    \end{align*}
    and, for all integers $k \leq d$,
    \begin{align*}
        \supnorm[\Big]{\sum_{i=k}^d 4\tilde{a}_{2i+2}x\cdot U_{i-k}(T_2(x))} \lesssim (d-k+1)(1+\log(d+1))\supnorm{p}.
    \end{align*}
\end{proposition}
\begin{proof}
Without loss of generality, we take $\supnorm{p} = 1$.
Consider the polynomial $q(x) = \sum_{i=0}^d b_i T_i(x)$ where $b_i \coloneqq a_{2i}$.
By \cref{eq:t-composition}, $p(x) = q(T_2(x))$, and because $T_2$ maps $[-1, 1]$ to $[-1, 1]$, $\supnorm{q} = \supnorm{q} = 1$.
Then by \cref{parity-two-tailsums},
\begin{align*}
    \abs*{\tilde{a}_{2k}}
    = \abs*{b_k - b_{k+1} + \cdots \pm b_d}
    \leq \parens[\Big]{4 + \frac{4}{\pi^2}\log(\max(k,1))}\supnorm{q}
    = 4 + \frac{4}{\pi^2}\log(\max(k,1)).
\end{align*}
From this follows the first two conclusions.
For the final conclusion, by \cref{eq:u-composition},
\begin{align*}
    \sum_{i=k}^d 4 \tilde{a}_{2i+2} x U_{i-k}(T_2(x))
    = \sum_{i=k}^d 2 \tilde{a}_{2i+2} U_{2(i-k)+1}(x)
    = \sum_i 2 \tilde{a}_{2i+2} U_{2(i-k)+1}(x).
\end{align*}
From here, we proceed as in the proof of \cref{thm:scalar-clenshaw-iterate-bound}.
\begin{align*}
    \sum_i \tilde{a}_{2i+2} U_{2(i-k)+1}(x)
    &= \sum_i \sum_{r \geq 0} (-1)^r b_{i+r+1} 2\sum_s T_{2s+1}(x) \iver{s \leq i-k}\\
    &= 2\sum_s T_{2s+1}(x) \sum_i \sum_{r\geq 0} \iver{s \leq i-k} (-1)^r b_{i+r+1} \\
    &= 2\sum_s T_{2s+1}(x) \sum_t b_t \sum_i \sum_r \iver{r \geq 0} \iver{t = i+r+1}\iver{s \leq i-k} (-1)^r \\
    &= 2\sum_s T_{2s+1}(x) \sum_t b_t \sum_r \iver{r \geq 0} \iver{s \leq t-r-1-k} (-1)^r \\
    &= 2\sum_s T_{2s+1}(x) \sum_t b_t \sum_{r=0}^{t-s-k-1} (-1)^r \\
    &= 2\sum_s T_{2s+1}(x) \sum_t b_t \iver{t-s-k-1 \in 2\mathbb{Z}_{\geq 0}} \\
    &= 2\sum_s T_{2s+1}(x) \sum_{t \geq 0} b_{2t+s+k+1} \\
    \supnorm[\Big]{\sum_i \tilde{a}_{2i+2} U_{2(i-k)+1}(x)}
    &\leq 2\sum_{s = 0}^{d-k} \abs[\Big]{\sum_{t \geq 0} b_{2t+s+k+1}} \\
    &\leq 2\sum_{s = 0}^{d-k} \parens[\Big]{4 + \frac{4}{\pi^2}\log(\max(s+k, 1))} \tag*{by \cref{parity-two-tailsums}}\\
    &\lesssim (d-k+1)(1+\log(d+1))
\end{align*}
\end{proof}


\section{Dequantizing QML Algorithms}

In this section, we focus on providing classical algorithms for regression, low-rank approximation and Hamiltonian simulation.
We show that plugging in the appropriate polynomials into our algorithm corresponding to Theorem~\ref{main-svt-alg} results in the fastest algorithms for these problems.  

\subsection{Recommendation Systems}
\label{subsec:rec-sys}

Given a matrix $A$, the quantum recommendation system problem, introduced by Kerenidis and Prakash~\cite{kp17}, is to output a sample from a row of a low-rank approximation of $A$.
The quantum algorithm for this approaches this in a different way from the standard classical algorithms, by taking a polynomial approximation of a threshold function and using quantum linear algebra techniques to apply it to $A$.
The threshold function is as follows:

\begin{lemma}[Polynomial approximations of the rectangle function {\cite[Lemma 29]{gslw18}}] \label{rec-polynomial}
    Let $\delta', \eps' \in (0, \frac12)$ and $t \in [-1,1]$.
    There exists an even polynomial $p \in \R[x]$ of degree $\bigO{\log(\frac{1}{\eps'})/\delta'}$, such that $\abs*{p(x)} \leq 1$ for all $x \in [-1, 1]$, and
    \begin{align*}
        p(x) \in 
        \begin{cases}
            [0, \eps'] & \text{for all } x \in [-1, -t-\delta'] \cup [t+\delta', 1]\text{, and} \\
            [1-\eps', 1] & \text{for all } x \in [-t-\delta', t+\delta']
        \end{cases}
    \end{align*}
\end{lemma}

In particular, one obtains the quantum recommendation systems result from QSVT is by preparing the state $\ket{A_{i,*}}$ and applying a block-encoding of $p(A)$ to get a copy of $\ket{p(A)A_{i,*}}$, where $p$ is the polynomial from Lemma~\ref{rec-polynomial}.\footnote{The original paper goes through a singular value estimation procedure; see the discussion in Section 3.6 of \cite{gslw18} for more details.}
We obtain the same guarantee classically:

\begin{corollary}[Dequantizing recommendation systems]
\label{cor:dq-rec-sys}
    Suppose we are given a matrix $A\in \mathbb{C}^{m \times n}$ such that $0.01 \leq \norm{A} \leq 1$ and $0 < \eps, \sigma < 1$, with $\bigO{\nnz(A)}$ time pre-processing.
    Then there exists an algorithm that, given an index $i \in [m]$, computes a vector $y$ such that $\Norm{ y - p(A) A_{i,*}  } \leq \eps\norm{A_{i,*}}$ with probability at least $0.9$, where $p(x)$ is the rectangle polynomial from \cref{rec-polynomial}, with parameters $t = \sigma$, $\delta' = \sigma/6$, $\eps' = \eps$.
    Further, the running time to compute such a description of $y$ is
    \begin{align*}
        \bigOt[\Big]{\frac{\fnorm{A}^4}{\sigma^{11}\eps^2}},
    \end{align*}
    and from this description we can sample from $y$ in $\bigOt*{\frac{\fnorm{A}^4\norm{A_{i,*}}^2}{\sigma^8\eps^2\norm{y}^2}}$ time.
\end{corollary}

Thinking of $p(A)A$ as the low-rank approximation of $A$ this algorithm gives access to, we see that the error guarantee of \cref{cor:dq-rec-sys} implies the error guarantee achieved by prior quantum-inspired algorithms~\cite{tang18a,cglltw19}.

\begin{remark}[Detailed comparison to {\cite[Theorem~26]{cchw20}}] \label{cchw-comparison-recsystems}
The recommendation systems algorithm of Chepurko, Clarkson, Horesh, Lin, and Woodruff outputs a rank-$k$ matrix $Z$ such that
\begin{align} \label{eq:classical-lra}
    \fnorm{A - Z}^2 \leq (1 + \bigO{\eps})\fnorm{A - A_k}^2,
\end{align}
where $A_k$ is the best rank-$k$ approximation to $A$, in $\bigOt{\frac{k^3}{\eps^{6}} + \frac{k\fnorm{A}^4}{(\fnorm{A - A_k}^2/k + \sigma_k^2)\sigma_k^2\eps^4}} = \bigOt{\frac{k^3}{\eps^{6}} + \frac{k\fnorm{A}^4}{\sigma_k^4\eps^4}}$ time\footnote{For ease of comparison, we suppose we know $\sigma_k$ and $\fnorm{A - A_k}^2$ exactly, in which case $\psi_\lambda \leq \psi_k$ and $k \leq \psi_k$, so the running time is $\bigOt{k^3/\eps^6 + k\psi_\lambda\psi_k/\eps^4}$, using the notation from \cite{cchw20}.}, under the assumption that $\eps \leq \frac{\fnorm{A_k}^2}{\fnorm{A}^2}$.
This guarantee is the typical one for low-rank approximation problems.
The authors use that $\ell_2^2$ importance sampling sketches oversample ridge leverage score sketch in certain parameter regimes, which is why dependences on the norm of the low-rank approximation $\fnorm{A_k}$ and its residual $\fnorm{A - A_k}$ appear.
For the recommendation systems task, it's not clear whether this style of guarantee is as good as the guarantee achieved by the quantum algorithm.
The quantum guarantee implies the classical one in the regime where additive and relative error are comparable: $\fnorm{A - Z} \leq \fnorm{A - p(A)A} + \fnorm{p(A)A - Z} \leq \fnorm{A - A_k} + \bigO{\eps\fnorm{A}}$, where the second inequality uses that $p(A)A$ is approximately a low-rank approximation that thresholds at $\sigma$, and so is at least as good of an approximation to $A$ as $A_k$, at the cost of being potentially higher rank.
Going from the classical guarantee to the quantum one loses a quadratic factor: it can be the case that the guarantee from \cref{eq:classical-lra} holds but $\fnorm{A_k - Z} = \bigOmega{\sqrt{\eps}\fnorm{A}}$.\footnote{
    Take, for example, $A = [\begin{smallmatrix}
        2 \\ & 1
    \end{smallmatrix}]$ and the rank $k=1$ approximation $Z = zz^\dagger$ where $z = [\begin{smallmatrix}
        \sqrt{2+\sqrt{\eps}} \\ \sqrt{\eps}
    \end{smallmatrix}]$.
    More generally, by \cite[Theorem 4.7]{tang18a}, \cref{eq:classical-lra} implies $\fnorm{Z - A_{\geq\sigma}} = \bigO{\sqrt{\eps}\fnorm{A - A_k}}$ (for certain forms of $Z$).
}
To achieve the precise guarantee of the quantum algorithm would require setting $\eps \leftarrow \eps^2$, giving a running time of
\begin{align*}
    \bigOt[\Big]{\frac{k^3}{\eps^{12}} + \frac{k\fnorm{A}^4}{\sigma_k^4\eps^8}}.
\end{align*}
We improve over this running time when $\eps = \bigO{\frac{k^{3/10}\sigma^{11/10}}{\fnorm{A}^{4/10}}}$ or when $\eps = \bigO{k^{1/6}\sigma^{7/6}}$.
\end{remark}

\subsection{Linear Regression}

Next, we consider linear regression: given a matrix $A$, a vector $b$, and a parameter $\kappa > 0$, the quantum algorithm obtains a state close to $\ket{f(A)b}$, where $f(x)$ is a polynomial approximation to $1/x$ in the interval $[-1, -1/\kappa] \cup [1/\kappa , 1]$~\cite[Theorem 41]{gslw18}.
Formally, this polynomial is the Chebyshev truncation of $(1 - (1-x^2)^b)/x$ for $b = \ceil*{\kappa^2\log(\kappa/\eps)}$, multiplied by the rectangle function from \cref{rec-polynomial} to keep the truncation bounded.

\begin{lemma}[Polynomial approximations of $1/x$, {\cite[Lemma 40]{gslw18}}, following \cite{cks17}] \label{inv-polynomial}
    Let $\kappa > 1$ and $0 < \eps < \tfrac12$.
    There is an odd polynomial $p(x)$ of degree $\bigO{\kappa \log(\frac{\kappa}{\eps})}$ with the properties that
    \begin{itemize}
        \item $\abs{p(x) - 1/x} \leq \eps$ for $x \in [-1, -1/\kappa] \cup [1/\kappa, 1]$;
        \item $\abs{p(x)} = \bigO{\kappa\log\frac{\kappa}{\eps}}$.
    \end{itemize}
\end{lemma}

We obtain the following corollary by invoking \cref{main-svt-alg} with the polynomial from \cref{inv-polynomial}.

\begin{corollary}[Dequantizing linear regression]
\label{cor:dequant-regression}
    Given a matrix $A \in \mathbb{C}^{m \times n}$ such that $0.01 \leq \norm{A} \leq 1$; a vector $b \in \mathbb{C}^m$; and parameters $\eps, 1/\kappa$ between $0$ and $1$, with $\bigO{\nnz(A) + \nnz(b)}$ time pre-processing.
    Then there exists an algorithm that outputs a vector $y$ such that $\norm{ y - p(A)b } \leq \eps\norm{b}/\kappa$ with probability at least $0.9$, where $p$ is the polynomial from \cref{inv-polynomial} with parameters $\kappa, \eps$.
    Further, the running time to compute a description of $y$ is
    \begin{align*}
        \bigOt[\Big]{\frac{\kappa^{11}\fnorm{A}^4}{\eps^2}},
    \end{align*}
    and from this description we can output a sample from $y$ in $\bigOt*{\frac{\kappa^{10}\fnorm{A}^4\norm{b}^2}{\eps^2\norm{y}^2}}$ time.
\end{corollary}

\begin{remark}[Comparison to {\cite[Theorem 24]{cchw20}}]
Chepurko, Clarkson, Horesh, Lin, and Woodruff solves the regularized regression problem, where the goal is to output a vector close to $x^* = \arg\min_{x} \norm{Ax - b}^2 + \lambda\norm{x}^2$ for a given $\lambda > 0$.
Their algorithm outputs a vector $y = (SA)^\dagger v$ such that
\begin{align*}
    \norm{y - x^*} \leq \eps\Big(\norm{x^*} + 2\frac{\norm{b}\norm{x^*}}{\norm{AA^+x^*}} + \frac{1}{\sqrt{\lambda}} \norm{\Pi_{\lambda, \bot} b}\Big),
\end{align*}
where $\Pi_{\lambda, \bot}$ is the projection onto the left singular vectors of $SA$ whose singular values are at most $\sqrt{\lambda}$.
This algorithm achieves a running time of $\widetilde{O}(\fnorm{A}^4(\norm{A}^2 + \lambda)^2\log(d)/((\sigma^2 + \lambda)^4\eps^4))$, where $\sigma$ is the minimum singular value of $A$.
Since $x^* = f(A)b$ is singular value transformation for the function $f(x) = x/(x^2 + \lambda)$, this is different from the setting we consider.
They become comparable when we take $\lambda \to 0$, in which case we must assume that $b$ is in the image of $A$, so that $\norm{\Pi_{\lambda, \bot} b}$ tends to zero, and we depend on the minimum singular value of $A$.
As discussed in \cref{intro-regression-comparison}, these assumptions are strong, and different from the setting we consider.
\end{remark}

\subsection{Hamiltonian Simulation}

Finally, we give a classical algorithm for Hamiltonian simulation: for a Hermitian matrix $H \in \mathbb{C}^{n\times n}$; a vector $b \in \mathbb{C}^n$; and a time $t \in \mathbb{R}$, the goal is to produce a description of $e^{\ii Ht}b$.
We begin by describing our polynomial approximation to $e^{\ii x} = \cos(x) + \ii \sin(x)$.

\begin{lemma}[Polynomial approximation to trigonometric functions, {\cite[Lemma 57]{gslw18}}] \label{trig-polynomials}
Given $t \in \mathbb{R}$ and $\eps \in (0, 1/e)$, let $r= \bigTheta{ t + \frac{\log(1/\eps)}{\log\log(1/\eps) } }$.
Then, the following polynomials $c(x)$ (even with degree $2r$) and $s(x)$ (odd with degree $2r+1$),
\begin{align*}
    c(x) &= J_{0}(t) - 2\sum_{i \in [1,r]} \Paren{-1}^{i} J_{2i}(t) T_{2i}(x) \\
    s(x) &= 2 \sum_{i \in [0,r]} \Paren{-1}^{i} J_{2i+1}(t) T_{2i+1}(x),
\end{align*}
satisfy that $\supnorm{\cos(tx) - c(x)} \leq \eps$ and $\supnorm{\sin(tx) - s(x)} \leq \eps$.
Here, $J_i(x)$ is the $i$-th Bessel function of the first kind~\cite[{(\href{https://dlmf.nist.gov/10.2.E2}{10.2.2})}]{DLMF}.
\end{lemma}

\begin{corollary}[Dequantizing Hamiltonian simulation]
\label{cor:dequant-ham}
    Given a Hermitian Hamiltonian $H \in \mathbb{C}^{n\times n}$ such that $0.01 \leq \norm{H} \leq 1$; a vector $b \in \mathbb{C}^n$; a time $t > 0$; and $0 < \eps < 1$, with $\bigO{\nnz(A) + \nnz(b)}$ pre-processing.
    Then there exists an algorithm that outputs a vector $y$ such that $\norm{y - e^{\ii Ht}b} \leq \eps\norm{b}$ with probability $\geq 0.9$.
    Further, the running time to compute such a description of $y$ is
    \begin{align*}
        \bigOt[\Big]{\frac{t^{11}\fnorm{H}^4}{\eps^2}},
    \end{align*}
    and from this description we can output a sample from $y$ in $\bigOt[\Big]{\frac{t^{8}\fnorm{H}^4\norm{b}^2}{\eps^2\norm{y}^2}}$ time.
\end{corollary}
\begin{proof}
Let $c(x)$ and $s(x)$ be the polynomials from \cref{trig-polynomials}.
We apply \cref{main-svt-alg} to get descriptions of $c$ and $s$ such that $\norm{y_c - c(H) b} \leq \eps\norm{b}$ and $\norm{y_s - s(H) b} \leq \eps\norm{b}$.
Then
\begin{align*}
    e^{\ii Ht}b
    = \cos(Ht)b + \ii \sin(Ht)b
    \approx_{\eps\norm{b}} c(H)b + \ii s(H)b
    \approx_{\eps\norm{b}} y_c + \ii y_s.
\end{align*}
This gives us a description $\bigO{\eps}$-close to $e^{\ii Ht}$.
Using \cref{Ax-sampling}, we can get a sample from this output in the time described, by combining the two descriptions of $y_c$ and $y_s$.
\end{proof}

\section*{Acknowledgments}
\addcontentsline{toc}{section}{Acknowledgments}

ET and AB thank Nick Trefethen, Simon Foucart, Alex Townsend, Sujit Rao, and Victor Reis for helpful discussions.
ET thanks t.f.\ for the support.
AB is supported by Ankur Moitra's ONR grant. 
ET is supported by the NSF GRFP (DGE-1762114).

\addcontentsline{toc}{section}{References}
\printbibliography


\end{document}